%% file: sc_main.tex
\documentclass[10pt,twocolumn,superscriptaddress,aps,pra]{revtex4-2}

\usepackage{amsmath,amsfonts,amssymb}
\DeclareMathOperator{\Tr}{Tr}
\usepackage{mathtools}
\usepackage{setspace}
\usepackage{stmaryrd}

\usepackage{enumitem}
\setlist{nosep}
\newlist{enuma}{enumerate}{10}
\setlist[enuma]{label*=\arabic*.}
\setlistdepth{10} 

\usepackage{braket}
\usepackage{array}
\usepackage{physics}

\usepackage{amsthm} 

\usepackage[thinc]{esdiff} 

\usepackage{dsfont} 
\usepackage{bm}

\usepackage{graphicx}
\usepackage{xcolor}
\usepackage[normalem]{ulem} 

\usepackage[colorlinks, linkcolor=brown, citecolor=blue, urlcolor=magenta]{hyperref}
\usepackage{url}

\usepackage{verbatim}
\usepackage{alltt}
\usepackage{moreverb}

\usepackage{subfiles} 

 \usepackage{mdframed} 
 
\usepackage{float}

\allowdisplaybreaks

\usepackage[absolute,overlay]{textpos} 

\def\cmntsoff{}

\input{estimation_defs}

\newcommand{\lb}{{\rm lb}}
\newcommand{\ub}{{\rm ub}}
\newcommand{\EF}{\text{estimation factor}}
\newcommand{\thresh}{{\rm{th}}}
\newcommand{\obs}{{\rm{obs}}}
\newcommand{\nll}{{\rm{null}}}
\newcommand{\KL}{{\rm KL}}
\newcommand{\pA}{{\rm A}}
\newcommand{\pB}{{\rm B}}
\newcommand{\pC}{{\rm C}}

\begin{document}
\title{An efficient method for spot-checking quantum properties with sequential trials}
\author{Yanbao Zhang}\thanks{yzh@ornl.gov}
\affiliation{Quantum Information Science Section, Computational Sciences and Engineering Division, Oak Ridge National Laboratory, Oak Ridge, Tennessee 37831, USA}
\author{Akshay Seshadri} 
\affiliation{Department of Physics, University of Colorado, Boulder, Colorado 80309, USA}
\affiliation{National Institute of Standards and Technology, Boulder, Colorado 80305, USA}
\author{Emanuel Knill}
\affiliation{National Institute of Standards and Technology, Boulder, Colorado 80305, USA}
\affiliation{Center for Theory of Quantum Matter, University of Colorado, Boulder, Colorado 80309, USA}

\begin{abstract}
In practical situations, the reliability of quantum resources can be compromised due to complex generation processes or adversarial manipulations during transmission. Consequently, the trials generated sequentially in an experiment may exhibit non-independent  and non-identically distributed (non-i.i.d.) behavior. This non-i.i.d. behavior can introduce security concerns and result in faulty estimates when performing information tasks such as quantum key distribution, self-testing, verifiable quantum computation, and resource allocation in quantum networks. To certify the performance of such tasks, one can make a random decision in each trial, either  spot-checking some desired property or utilizing the quantum resource for the given task. However, a general method for certification with a sequence of non-i.i.d. spot-checking trials is  still missing. Here, we develop such a method.  This method not only works efficiently with a finite number of trials but also yields asymptotically tight certificates of performance.  Our analysis shows that even as the total number of trials approaches infinity, only a constant number of trials needs to be spot-checked on average to certify the average performance of the remaining trials at a specified confidence level. 
  \end{abstract}

\maketitle

\begin{textblock}{13.3}(1.4,15)\noindent\fontsize{7}{7}\selectfont\textcolor{black!30}{This manuscript has been co-authored by UT-Battelle, LLC, under contract DE-AC05-00OR22725 with the US Department of Energy (DOE). The US government retains and the publisher, by accepting the article for publication, acknowledges that the US government retains a nonexclusive, paid-up, irrevocable, worldwide license to publish or reproduce the published form of this manuscript, or allow others to do so, for US government purposes. DOE will provide public access to these results of federally sponsored research in accordance with the DOE Public Access Plan (http://energy.gov/downloads/doe-public-access-plan).}\end{textblock}

Given a finite number of experimental trials, one would like to
certify a specific property of a quantum resource with high
confidence. This is crucial not only for characterizing the quantum
resource itself but also for assessing its potential performance in a
subsequent quantum information processing task.  The property of
interest is often quantifiable by the expectation value of an
observable $O$ given the underlying quantum state $\rho$.  This
captures many cases such as certifying state
fidelity~\cite{fidelity_rev}, entanglement~\cite{ent_rev},
steerability~\cite{ste_rev}, or Bell nonlocality~\cite{bell_rev}.
Given a sufficiently informative measurement configuration, one can
then determine a function $X$ of measurement outcomes
such that the mean $\Exp(X)$ satisfies $\Exp(X)=\Tr(\rho O)$.
With such an $X$, the certification problem can
be formulated as obtaining a high-confidence lower or upper bound on
$\Exp(X)$ from a sequence of trials, each involving the generation and
measurement of the state $\rho$.  Ideally, the trials generate
independent and identically distributed (i.i.d.) random variables,
each with the distribution of $X$ determined by the state $\rho$ and
the measurement configuration.  However, in real life, both the
quantum states generated and measurements performed can autonomously
drift over time or be subject to manipulations by untrusted parties,
resulting in a sequence of non-i.i.d. random variables.  To counter
this practical problem, one can employ concentration inequalities
tailored for non-i.i.d. sequential random variables, such as the
Azuma-Hoeffding inequality~\cite{hoeffding1963, azuma1967}.  These
concentration inequalities have been used for certifying Bell
nonlocality~\cite{Gill2003a, Gill2003b, Zhang2011, Zhang2013,
  Bierhorst2014, Bierhorst2015, Elkousst2016} and device-independent
quantum randomness~\cite{Pironio2010, Pironio2013}, for example. They
can also help the security analysis of quantum key
distribution~\cite{Boileau2005, Tamaki2009, Tamaki2014, Tamaki2018,
  Mizutani2019,Pereira2019} and self-testing~\cite{Tan2017,
  Bancal2021, Gocanin2022} with finite data.

In many applications, it is desirable to measure only a subset of the
states generated in a set of trials for the purpose of certification,
either because the measurement is expensive or because the leftover
states are to be used in a protocol without having been destroyed or
perturbed by the certification measurement.  An example of the
  former is certifying randomness from random circuit
  sampling~\cite{liu2025certified}. Examples of the latter 
include quantum key distribution~\cite{qkd_rev2020}, quantum
randomness expansion~\cite{qrand_rev2017},
self-testing~\cite{supic:2020}, verifiable quantum
computation~\cite{ver_qc_rev2017, mahadev:2018}, and resource
allocation in quantum networks~\cite{qnet_rev2023}.  Strategies for
measuring random subsets of the generated states are referred to as
spot-checking strategies.

To relax the i.i.d. assumption in spot-checking, previous works for
quantum cryptography~\cite{Bouman2010, Tomamichel2012, Tomamichel2017}
and quantum state certification~\cite{Takeuchi2019} resort to
concentration inequalities for sampling without replacement, for
example, the Serfling inequality~\cite{Serfling1974, Bardenet2015}.
These inequalities work well for parallel trials. For parallel trials,
the experiment is equivalent to one where the global state across all
trials is generated in the first step, and then the trial-wise
measurements for each spot-checked trial are made in parallel in the
second step.  Since the trials in current quantum experiments are
generated in sequence rather than in parallel, concentration
inequalities for sampling without replacement cannot be directly
applied without introducing additional assumptions. For example,
Refs.~\cite{Bouman2010, Tomamichel2012, Tomamichel2017} assumed that
the experimental device used is memoryless.  The recent
work~\cite{Gocanin2022} on self-testing considers spot-checking with a
sequence of trials, where one randomly decides, before each trial,
whether or not to spot-check and measure the trial's state.  The
method developed there allows for non-identical trials but still
assumes that these trials are independent.

In this work, our focus is on spot-checking with sequential trials.
We assume that in each trial there is an observable random
variable whose value or past-conditional expectation relates to a
parameter or property to be certified.  The random variable may be
associated with the outcomes of a measurement that is performed only
when the trial is spot-checked.  With the assumptions discussed above,
concentration inequalities for sampling without replacement can only
be used to estimate the average of unobserved random
variables~\cite{Tomamichel2012, Tomamichel2017}, while the method of
Ref.~\cite{Gocanin2022} is applicable only for estimating the average
of the past-conditional expectations of unobserved random variables.
We introduce the method of estimation factors to obtain
high-confidence bounds for both estimation problems without the above
assumptions, given the observed values of the non-i.i.d. random
variables that are spot-checked. The estimation-factor method has two
additional advantages. First, it can significantly reduce the total
number of trials required to certify a fixed confidence bound at a
given confidence level.  Second, the method permits early
stopping. For example, it is possible to stop as soon as the number of
trials that are unchecked (that is, not spot-checked) reaches a
predetermined number. This is suitable for use in a block protocol
that requires a fixed-length input. Methods such as the one in
Ref.~\cite{Gocanin2022} require a fixed total number of trials,
resulting in a random number of unchecked trials.

\noindent{\textit{Estimation-factor method.}}
We consider an experiment consisting of $n$ sequential trials indexed
by $i\in\{1,\ldots,n\}$. The $i$'th trial makes available a real
random variable $X_{i}$ that is observed only if the trial is spot-checked. We
assume that $X_{i}$ is bounded on at least one side. This ensures that
conditional expectations of $X_{i}$ have well-defined values in
$[-\infty,+\infty]$. The choice to spot-check is determined by a
second, 0/1-valued random variable $Y_{i}$, where the $i$'th trial is
spot-checked if $Y_{i}=0$. Let $\bm{X}$ and $\bm{Y}$ denote the
sequences of random variables $(X_i)_{i=1}^n$ and $(Y_i)_{i=1}^n$.
Similarly, $\bm{X}_{<i}$ and $\bm{Y}_{<i}$ denote the sequences
$(X_j)_{j=1}^{i-1}$ and $(Y_j)_{j=1}^{i-1}$ preceding the $i$'th
trial. We refer to the combination of $\bm{X}_{<i}$, $\bm{Y}_{<i}$,
and any other information obtained before each trial $i$ as the past
of the trial, denoted by $\Past_i$. Both $X_i$ and $Y_i$ can depend on
their past, but we require that $X_{i}$ and $Y_{i}$ are conditionally
independent of each other given the past. This ensures that the
past-conditional expectation of $X_{i}$ is independent of whether or
not the $i$'th trial is spot-checked.  We refer to this requirement as
the free-choice assumption.  This assumption also conveys that before
each trial and given the past, one has the freedom to determine
whether or not to spot-check the trial.

We use $\Exp$ and $\Prob$ to denote expectations and probabilities,
respectively.  To simplify the treatment of conditional probabilities,
in this work we focus on the situation where the probability
distributions involved are discrete. In quantum applications of
interest, either the random variable $X_i$ itself or its
past-conditional expectation, denoted by
$\Theta_{i}=\Exp (X_i | \Past_i )$, reflects the quality of the
quantum state generated in the $i$'th trial.  For example, in security
or error analysis~\cite{Bouman2010, Tomamichel2012, Tomamichel2017},
$X_i$ can be a binary random variable with values $1$ or $0$,
indicating whether or not an error occurs. Another example is state
certification~\cite{supic:2020, Yu:2022}, where the focus is on the
fidelity of the state generated in the $i$'th trial with respect to a
target state, given the past. This is readily relatable to a
past-conditional expectation $\Theta_{i}$.  

Our aim is to construct a confidence bound for the sum of the past-conditional 
expectations $\Theta_{i}$ for unchecked trials $i$, which is expressed as
\begin{equation} \label{main-eq:weighted_sum_means}
S_n=Y_1\Theta_1+Y_2\Theta_2+\ldots+Y_n\Theta_n.
\end{equation}
Another quantity of interest is the sum of the unobserved random variables $X_{i}$, 
expressed as
\begin{equation} \label{main-eq:weighted_sum_rvs}
S'_n=Y_1X_1+Y_2X_2+\ldots+Y_nX_n.
\end{equation}
As it turns out, the confidence bounds for $S_{n}$ and $S'_{n}$
obtained using the estimation-factor method introduced here are
identical at the same confidence level.

Since upper confidence bounds for a quantity can be obtained from
lower confidence bounds for the negative of the quantity, we focus on
constructing lower confidence bounds.  When a confidence interval is
desired, upper and lower confidence bounds can be combined using the
union bound. To obtain lower confidence bounds, we assume that
\(X_{i}\geq b\) for some known constant \(b\). While we treat \(b\) as
constant here, the estimation-factor method can be straightforwardly
adapted to scenarios where \(b\) depends on the trial index.

To construct a lower confidence bound on $S_{n}$, we choose 
a power $\beta>0$ and bound the monotonically decreasing function 
$\exp(-\beta S_n)$ of $S_n$ from above.  To achieve this for a 
sequence of  non-i.i.d. trials, we determine a sequence of non-negative functions
$T_i(X_i,Y_i)$, which can depend on the past of the $i$'th trial, such that 
\begin{equation} \label{main-eq:ef_chain}
 \Exp \qty(  e^{-\beta S_n} \prod_{i=1}^n T_i(X_i,Y_i) )   \leq 1. 
\end{equation}
Intuitively, the constraint ensures that it is unlikely for both
$e^{-\beta S_n}$ and $\prod_{i=1}^n T_i(X_i,Y_i)$ to be large
simultaneously.  To formalize this intuition, we apply Markov's
inequality $\Prob(Z>1/\eps )\leq \eps \Exp(Z)$ for $\eps\in (0,1)$ to
the non-negative random variable
$Z= e^{-\beta S_n}\prod_{i=1}^n T_i(X_i,Y_i)$. By manipulating the
inequality inside $\Prob(.)$, we obtain
$\Prob \big(S_n < S_{\text{lb},n}\big)= \Prob(Z>1/\eps ) \leq \eps$,
where
 \begin{equation} \label{main-eq:conf_bound}
  S_{\lb,n} = \sum_{i=1}^n \ln \big(T_i(X_i,Y_i)\big)/\beta+\ln(\eps)/\beta.
\end{equation}
Therefore, $S_{\text{lb},n}$ serves as a lower confidence bound on
$S_n$ with confidence level $(1-\eps)$.  To ensure that 
Eq.~\eqref{main-eq:ef_chain} is satisfied, 
it suffices to ensure that for every \(i\) and every probability 
distribution satisfying the assumptions specified above, the
inequality
\begin{align} \label{main-eq:ef_def}
\Exp\big(T_{i}(X_i, Y_i)e^{-\beta Y_{i}\Theta_{i}}|\Past_i \big) &\leq 1
\end{align}
holds; see App.~\ref{sect:EF_method} for details.
We refer to a non-negative function $T_i(X_i, Y_i)$ satisfying Eq.~\eqref{main-eq:ef_def} 
as an estimation factor of power $\beta$ for estimating $Y_i\Theta_i$.  
Accordingly, the defining inequality in Eq.~\eqref{main-eq:ef_def} is termed the
estimation-factor inequality.  We can follow the same steps to define 
estimation factors for estimating $S'_{n}$, replacing $e^{-\beta Y_{i}\Theta_{i}}$ with
  $e^{-\beta Y_{i} X_{i}}$ in Eq.~\eqref{main-eq:ef_def},  and obtain the 
  same expression as Eq.~\eqref{main-eq:conf_bound} for the lower 
  confidence bound on $S'_n$ with confidence level $(1-\eps)$.

Once a lower confidence bound $S_{\lb,n}$ for $S_n$ is obtained, it can be converted into a 
lower confidence bound for the average of $\Theta_i$  over unchecked trials, 
which is often the actual quantity of interest in applications. Let $C_n = \sum_{i=1}^{n}Y_{i}$ 
be the number of unchecked trials.  One is inclined to write $S_{n}/C_{n}$ for the average and 
use $S_{\lb,n}/C_{n}$ as the corresponding lower confidence bound, but it is necessary to 
account for the possibility that $C_{n}=0$, in which case all trials are spot-checked. 
For this purpose, we define the average of interest as $\bar{\Theta}=b$ when $C_{n}=0$
and $\bar{\Theta} = S_{n}/C_{n}$ otherwise. A lower confidence bound on $\bar{\Theta}$ with confidence 
level $(1-\eps)$ is then $\bar{\Theta}_{\lb}=b$ when $C_{n}=0$ and $\bar \Theta_{\lb}=S_{\lb,n}/C_n$ otherwise. 
To verify that $\bar{\Theta}_{\lb}$ is a valid lower confidence bound, observe that
 \begin{align}\label{main-eq:conf_bound_conversion}
    \Prob(\bar{\Theta} < \bar{\Theta}_{\lb})
    &= \Prob(C_{n}\neq 0, S_{n}< S_{\lb, n})\nonumber\\
    &\leq\Prob(S_{n}< S_{\lb,n}) 
    \leq \epsilon.
 \end{align}

\noindent{\textit{Constructing estimation factors.}}
We construct estimation factors for the basic spot-checking strategy
where the distribution of the spot-checking variable \(Y_{i}\) in
each trial $i$ is fixed and independent of the past, with a known
probability $\omega_i \in (0, 1)$ that the trial is spot-checked,
  namely that \(Y_{i}=0\).  In App.~\ref{sect:EFs_with_bias}, we
consider the general situation where the past-conditional distribution
of \(Y_{i}\) is constrained but not precisely known. The
estimation-factor inequality in Eq.~\eqref{main-eq:ef_def} applies to
each individual trial, so we consider a generic trial with random
variables \(X\) and \(Y\), where we omit explicit mention of trial
indices and the past.  Since $X$ is explicitly observed only when
$Y=0$, the estimation factor $T(X,Y)$ must not depend on $X$ when
$Y\neq 0$. As a result, $T(X,Y)$ can be expressed as
\begin{equation}
T(X,Y)=
\begin{cases}
T'(X), &  \text{if\ }  Y=0, \\
 t,  & \text{if\ } Y=1,
\label{main-eq:ef_form}
\end{cases}
\end{equation}
where both $T'(X)$ and $t$ are required to be non-negative. For a generic trial, 
the free-choice assumption states that $X$ and $Y$ are independent, where 
$\Prob(Y=0)=\omega$. Therefore, the estimation-factor inequality in 
Eq.~\eqref{main-eq:ef_def} becomes
\begin{equation} \label{main-eq:ef_simp}
  (1-\omega)te^{-\beta  \Exp(X)} +\omega  \Exp\big( T'(X) \big) \leq 1.
\end{equation}
This inequality must be satisfied by all distributions of \(X\) on
\([b,\infty)\), including the deterministic ones for which
\(\Prob(X=x)=1\) for some \(x\geq b\).  Substituting these deterministic distributions into
Eq.~\eqref{main-eq:ef_simp}, and since \(x\geq b\) is arbitrary, we
find that the function \(T'\) must satisfy
\(T'(x)\leq F'_{\beta, t}(x) \coloneqq \qty(1-(1-\omega)t e^{-\beta
  x})/\omega\) for \(x\geq b\).  The non-negativity of estimation
factors requires \(t\in[0, e^{\beta b}/(1-\omega)]\).  Since the
function $F'_{\beta, t}$ is concave,  Jensen's inequality implies  
\(\omega \Exp(T'(X)) \leq \omega \Exp(F'_{\beta, t}(X)) \leq \omega
F'_{\beta, t}(\Exp(X))= 1- (1-\omega)te^{-\beta \Exp(X)}\).
Therefore, all $T(X, Y)$ of the form in Eq.~\eqref{main-eq:ef_form}
with $T' \leq F'_{\beta, t}$ satisfy the estimation-factor inequality
in Eq.~\eqref{main-eq:ef_simp}.  Because increasing the values of
\(T'\) increases the lower confidence bounds obtained, we may set
\(T'= F'_{\beta, t}\), where $\beta>0$ and
$t\in [0, e^{\beta b}/(1-\omega)]$, for better bounds.  The resulting
estimation factors are denoted by $F_{\beta, t}(X,Y)$ and are called
extremal.

If we wish to estimate \(S'_{n}\) instead of \(S_{n}\), the estimation-factor
inequality for a generic trial becomes
\begin{equation} \label{main-eq:ef_simp_pr}
  (1-\omega)t\Exp(e^{-\beta  X}) +\omega  \Exp\big( T'(X) \big) \leq 1.
\end{equation}
By considering deterministic distributions at \(x\ge b\) as
before, we find that to satisfy Eq.~\eqref{main-eq:ef_simp_pr}, it is also
necessary that the function \(T'\leq F'_{\beta,t}\).  Provided that the function \(F'_{\beta,t}\) is
non-negative on \(x\in [b,\infty)\), by the concavity of \(F'_{\beta,t}\), we can
once again set \(T'=F'_{\beta,t}\) to satisfy Eq.~\eqref{main-eq:ef_simp_pr}. Consequently,
the estimation factors $F_{\beta, t}(X,Y)$ are extremal for estimating 
both $S_n$ and $S'_n$, and these estimation factors yield the same 
lower confidence bounds for both estimation problems.

To determine estimation factors in an experiment, we can proceed trial
by trial.  The power \(\beta\) must be chosen independently of the
data to which the estimation factors will be applied, preferably
before the experiment is performed.  The choice may be based on
independently acquired calibration data, or modeling and
simulations. Before each trial $i$ and based on its past, we choose
\(t_{i}\) to determine the estimation factor
$T_{i}(X_{i},Y_{i})=F_{\beta, t_i}(X_{i},Y_{i})$ to be used in the
$i$'th trial.  To optimize \(\beta\) and \(t_{i}\), we can maximize
the predicted expectation of the lower confidence bound using all
available information.
Whether or not the predictions are accurate, the estimation 
factors and lower confidence bounds obtained remain valid.

It is often desirable to optimize estimation-factor choices for the
case where the experiment performs as planned, which typically means
that the \(X_{i}\) are close to being i.i.d. and each trial is
spot-checked with a constant probability \(\omega\). For this purpose,
it is not necessary to perform fine-grained optimization at each
trial. A simple strategy is to collect a set of \(n_{\text{c}}\)
calibration trials, with each trial being spot-checked, before the
planned experiment.  To determine the estimation factor to be used 
for analyzing each post-calibration trial, we can maximize
the expected lower confidence bound assuming that the \(X_{i}\) for
these trials are sampled  i.i.d. from the empirical frequencies of 
the calibration data.  For this, we use the
extremal estimation factors \(F_{\beta, t}\) and optimize the expected
lower confidence bound over both \(\beta\) and \(t\); see
  App.~\ref{sect:numerical_opt}.  This strategy can be further
simplified by computing \(\beta\) and \(t\) directly from
calibration-based estimates of the mean and variance of \(X\), or even
with fixed choices for \(\beta\) and \(t\) without relying on
calibration, depending on how tight the lower confidence bounds need
to be. See App.~\ref{sect:analytical_opt} and the end of
App.~\ref{sect:tightness} for details.  When using estimates of the
mean and variance, we find that with regularization to account for
rare events and boundary cases, moderate values of \(n_{\text{c}}\)
suffice. In App.~\ref{sect:calibratevar} we determine bounds on the
needed values of \(n_{\text{c}}\).

The flexibility to choose the spot-checking probability \(\omega_{i}\)
and the estimation factor \(T_{i}\) before each trial $i$ can be
exploited for early stopping. Early stopping involves using current
and past information to decide whether to ignore or terminate future
trials.  This decision can be effectively implemented by choosing
\(\omega_{j}=1\) and \(T_{j}=1\), because for \(\omega_{j}=1\), in
Eq.~\eqref{main-eq:ef_def} we always have \(Y_{j}=0\), and the
quantity \(e^{-\beta Y_{j}\Theta_{j}}\) is identically \(1\).  Since
these estimation factors contribute nothing to the resulting
confidence bound, and these future trials do not contribute to either
\(S_{n}\) or \(S_{n}'\), it is not necessary to physically execute
these trials. One application of early stopping is to ensure that,
with overwhelming probability, the number of unchecked trials after
the experiment is some positive number $m$ fixed beforehand.  For this
purpose, given the spot-checking probability \(\omega\) in each
executed trial, it suffices to set \(n\) sufficiently larger than
\(m/(1-\omega)\) before the experiment starts.  See
App.~\ref{sect:early_stop} for details.

Beyond early stopping, several other properties of the estimation-factor 
method can also be established. In App.~\ref{sect:constant_omega_n}, 
we show that for a  fixed confidence level and target gap between 
a quantity averaged over unchecked trials and its corresponding lower 
confidence bound, a constant expected number of
spot-checked trials suffices even as the total number of trials
approaches infinity. We also prove the asymptotic tightness and
finite-data efficiency of the confidence bounds obtained with
estimation factors in App.~\ref{sect:tightness} and
App.~\ref{sect:efficiency}, respectively.  

\begin{figure}[htb!]
  \begin{center}
  \includegraphics[scale=0.45]{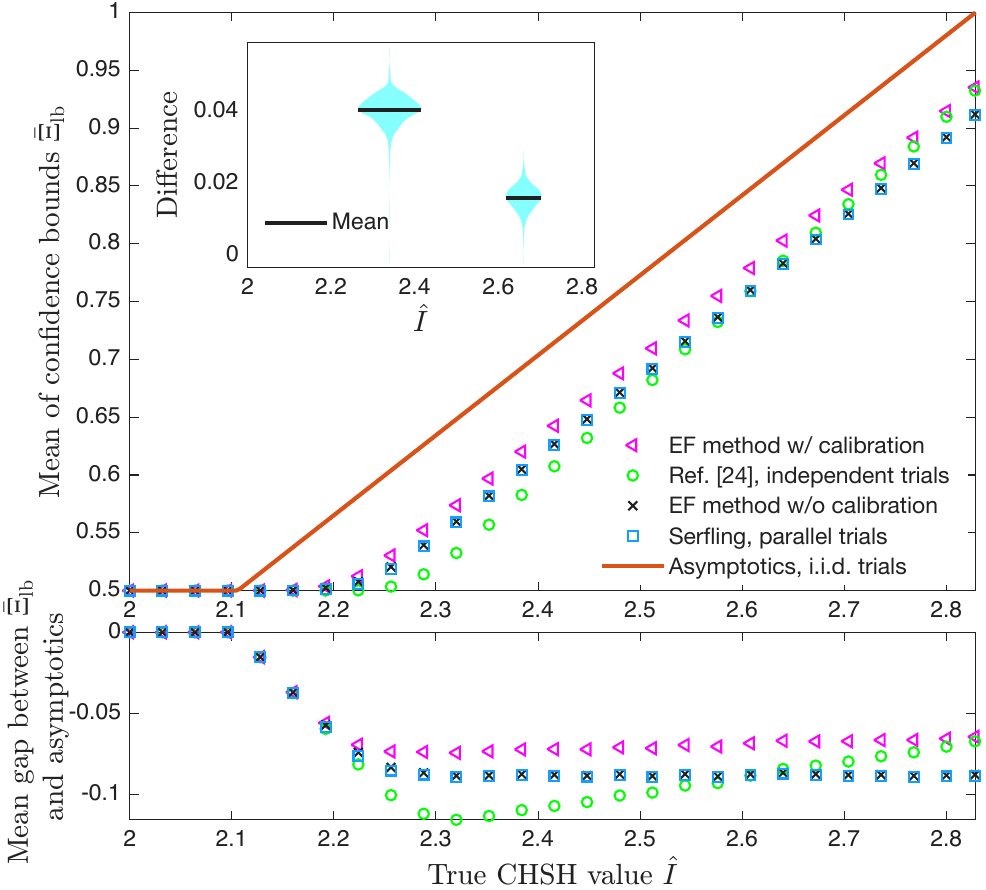}
  \end{center}
  \vspace{-0.4cm}
    \caption{Lower confidence bounds $\bar \Xi_{\lb}$ on the average extractability $\bar \Xi$ of the Bell state from 
    the past-conditional states generated in the unchecked trials as a function of the true CHSH value $\hat I$ 
    underlying the spot-checked trials. We fix the total number of trials at $n=10^5$ and set the confidence level 
    at $(1-\epsilon)=99\SI{\%}$.  Since a separable state can achieve the trivial extractability $0.5$, we 
    truncate all confidence bounds at this value. The upper panel shows the mean of $\bar \Xi_{\lb}$ across 
    $1000$ independent datasets for each method considered. For reference, the asymptotic bound from 
    Ref.~\cite{kaniewski:2016} is also shown. The inset violin plot shows the distribution of the differences between 
    the confidence bounds obtained by the estimation-factor (EF) method with calibration and the method of 
    Ref.~\cite{Gocanin2022} for each dataset when $\hat I=2.34$ or $2.66$. 
    The lower panel illustrates the mean gap between $\bar \Xi_{\lb}$ and the asymptotic bound.  
    Note that in this figure, as well as in Fig.~\ref{main-fig:finite_efficiency} below, we also illustrate the 
    performance of the Serfling inequality~\cite{Serfling1974} solely for comparison.
    For explanations of the method of Ref.~\cite{Gocanin2022}, the application of the Serfling inequality 
    and how this figure was obtained, see App.~\ref{sect:comparison}.
}
  \label{main-fig:tight_bound}
\end{figure}

\noindent{\textit{Application example.}}
The estimation-factor method can be applied to situations where
multiple parties wish to assure that they share a target entangled
state with good average fidelity by destructively spot-checking a
faction of states repeatedly shared among them. 
 Here, we consider the specific case where two parties aim to share Bell pairs.
Self-testing~\cite{supic:2020} can establish a device-independent
lower bound on the fidelity of a Bell pair extractable from some
subsystem at each party.  This fidelity is known as the extractability
$\Xi$ of the pair.  The extractability lower bound can be related
linearly~\cite{supic:2020, kaniewski:2016} to an expectation \(\hat I\) 
of a binary random variable, where \(\hat I\) is known as the CHSH value 
and indicates a violation of the CHSH inequality~\cite{chsh:1969}
when \(\hat I >2\).  Accordingly, 
the extractability lower bound conditional on the past of a
trial $i$ can also be expressed as an expectation of a
binary random variable \(X_{i}\); see App.~\ref{sect:comps}. 

We aim to estimate the past-conditional extractability 
averaged over a finite number of trials, 
denoted by $\bar \Xi$. Previous statistical methods for this purpose 
either do not work with spot-checking~\cite{Bancal2021} or require 
the assumption that the states generated in sequential trials are 
independent~\cite{Gocanin2022}. In contrast, the estimation-factor method
enables estimating the average extractability $\bar \Xi$ over unchecked non-i.i.d. trials.

We compare the estimation-factor method with the method 
of Ref.~\cite{Gocanin2022} and with the Serfling 
inequality~\cite{Serfling1974,Tomamichel2017}. The methods 
compared with require additional assumptions or address a different problem as noted in 
the introduction; see App.~\ref{sect:comparison} for further details.
To illustrate the performance of the methods considered, 
we simulate $1000$ independent datasets, each with 
$n=10^5$ i.i.d. trials with a spot-checking probability 
of $\omega=0.1$ and a true CHSH value $\hat I \in [2, 2\sqrt{2}]$. 
From each dataset, we determine a lower confidence bound 
$\bar \Xi_{\lb}$  using each method under its respective
assumptions.  We then plot the mean of the confidence bounds 
across datasets for each method, as shown in the upper panel of
Fig.~\ref{main-fig:tight_bound}. 
We also plot the asymptotic lower
bound \(\check\Xi\)  on the average extractability $\bar \Xi$ derived in
Ref.~\cite{kaniewski:2016} for an infinite number of i.i.d. trials.
We use the same estimation factor for all trials in each 
dataset. We determine this estimation factor prior to the experiment, 
either by optimization using a set of $100$ i.i.d. calibration 
trials sampled according to the true $\hat I$, 
or by using a suboptimal choice without calibration, as described 
in App.~\ref{sect:comps}.  The results in Fig.~\ref{main-fig:tight_bound}
indicate that estimation factors yield tighter confidence bounds, 
especially for practically relevant values of $\hat I$.

Alternatively, one can consider the minimum number of
trials required by each method to ensure that the expected lower
 confidence bound on the average extractability $\bar \Xi$ exceeds 
 a target threshold.  This minimum number of trials is estimated 
 assuming i.i.d. trials with a spot-checking probability of
 $\omega=0.1$,  based on the true underlying distributions.
  As shown in Fig.~\ref{main-fig:finite_efficiency}, a significant 
  improvement is achieved with estimation factors.  Notably, unlike
  estimation factors, the method of Ref.~\cite{Gocanin2022} cannot
  reach the asymptotic lower bound of \(\check\Xi=0.9111\) given the 
  true CHSH value  \(\hat{I}=2.7\) considered in 
  Fig.~\ref{main-fig:finite_efficiency}; specifically, this
  method diverges before the expected lower confidence bound
  increases to \(\check \Xi-\delta_{\text{th}}\) with \(\delta_{\text{th}}=0.0098\).  
  Note that the value $\hat{I}=2.7$ is chosen in accordance
  with the recent detection-loophole-free Bell tests reported in
  Refs.~\cite{nadlinger:2022, zhang:2022}.

\begin{figure}[htb!]
  \begin{center}
    \includegraphics[scale=0.45]{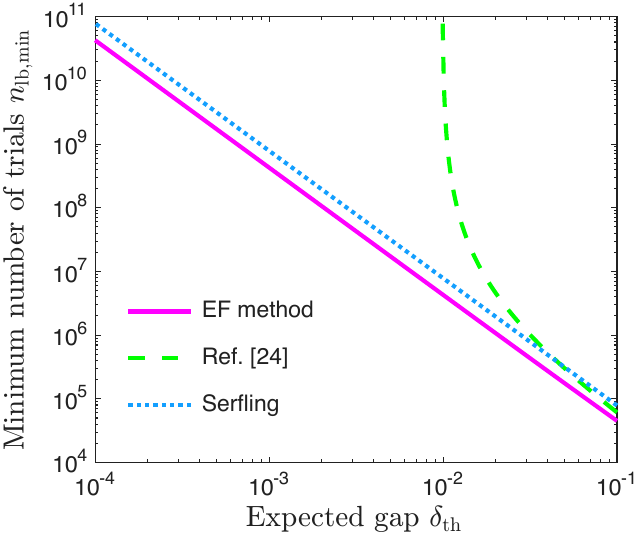}
  \end{center}
  \vspace{-0.5cm}
    \caption{Minimum number of trials $n_{\text{lb}, \min}$ required to 
  ensure that the expected lower confidence bound on $\bar \Xi$ exceeds
  \(\check \Xi-\delta_{\text{th}}\), where \(\check \Xi=0.9111\) is the 
  asymptotic lower bound on $\bar \Xi$ corresponding to a true 
  CHSH value \(\hat{I}=2.7\) 
  and \(\delta_{\text{th}}\) is the varying expected gap. We set the
      confidence level to \((1-\epsilon)=99\SI{\%}\) and assumed i.i.d. trials
      with spot-checking probability $\omega=0.1$. 
      See App.~\ref{sect:comps} for details. }    
     \label{main-fig:finite_efficiency}
\end{figure}

\noindent{\textit{Discussion.}}
The estimation-factor method is also applicable to 
 quantum key distribution; see App.~\ref{sect:appls} for details. 
 In some quantum information tasks, the performance in a trial $i$ is not directly determined by the random variable $X_i$
or its past-conditional expectation $\Theta_{i}$, but rather by some function of $X_i$ or $\Theta_{i}$. Without loss of generality, 
we consider the case where the quantity of interest in each trial $i$ is  $g(\Theta_{i})$.  In such cases, the estimation-factor method remains
applicable provided that the function $g$ satisfies certain properties.  
For example, if $g$ is a convex function, we have
\begin{align}\label{main-eq:cvx_est}
  S_{g,n}&=Y_1 g(\Theta_1)+Y_2 g(\Theta_2)+\ldots+Y_n g(\Theta_n) \notag \\
         &\geq C_n g(\bar \Theta), 
\end{align}
where $C_n=\sum_{i=1}^n Y_i$ and $\bar \Theta=S_n/C_n$.
Furthermore, if $g$ is a monotonically non-decreasing (or non-increasing) function, we can leverage
a lower (or an upper) confidence bound on $\bar \Theta$ to obtain  
a lower confidence bound on $S_{g,n}$, and thus a lower confidence bound on the average of $g(\Theta_{i})$ over unchecked trials. 

In quantum applications of spot-checking, the states produced in
unchecked trials might be measured or used in a way that is
incompatible with the measurement that produces \(X_{i}\) in
spot-checked trials. Because \(\Exp(X_{i}|\Past_{i})\)  is
well-defined whether or not \(X_{i}\) is realized, the confidence
bounds on \(S_{n}\) are still meaningful. One interpretation of the
quantity \(S'_{n}\) can be obtained by virtually timeshifting all
incompatible measurements until after the end of the spot-checking experiment, for
which purpose we imagine a virtual perfect quantum memory. In this
case,  the relevant joint distributions of the \(X_{i}\) are definable
even if not realized, and \(S'_{n}\) would be the sum of the
unobserved \(X_{i}\) if they had been realized. A confidence bound 
on \(S'_{n}\) thus still provides useful information.

\noindent{\textit{Conclusion.}} We introduced the estimation-factor method for
certifying the performance of a sequence of tasks repeatedly applied
to states from a source by spot-checking only a fraction of those
states.  The method works without an i.i.d. assumption for the
states produced at different times. Good results can be obtained by
spot-checking only a constant number of states, thereby minimizing
interruption to the underlying tasks.  The estimation-factor method
demonstrates exceptional performance, as illustrated by the example
of self-testing state extractability in the spot-checking scenario
without imposing unwanted assumptions.  Several open problems
remain.  One is to extend the method to certify multiple quantum
properties simultaneously in adversarial settings, in the spirit of
Ref.~\cite{Huang:2020} but using only sparse spot-checking and
without the i.i.d. assumption.  Another is to move beyond the
current reliance on (one-sided) bounded random variables and handle
cases where the relevant variables may be unbounded but have finite
variances.


\noindent{\textit{Acknowledgements.}} 
We thank Carl Miller and Go Kato for valuable comments on the paper.
We also thank Carl Miller, Shawn Geller, and Aneesh Ramaswamy for 
assistance with reviewing the paper before submission.  Y.Z. additionally 
thanks Valerio Scarani and Yeong-Cherng Liang for 
historical and stimulating discussions on self-testing.  
This work was performed in part at Oak Ridge National
Laboratory, operated by UT-Battelle for the U.S. Department of Energy
under contract no. DE-AC05-00OR22725. This work also includes
contributions of the National Institute of Standards and Technology (NIST),
which are not subject to U.S. copyright.  Y.Z. acknowledges funding
provided by the U.S. Department of Energy, Office of Science, Advanced
Scientific Computing Research (Field Work Proposal ERKJ381).
A.S. acknowledges support from the Professional Research Experience
Program (PREP) operated jointly by NIST and the University of
Colorado.
\newpage

\bibliographystyle{unsrt}
\bibliography{spotcheck_est}

\clearpage
\newpage
\onecolumngrid

\appendix

\tableofcontents

\section*{Appendix}
\label{sect:appendix}

\subfile{spotcheck_est_apq.tex}

\end{document}


%% file: estimation_defs.tex
\providecommand{\ignore}[1]{}

\newif\ifcmnt
\cmnttrue
\ifdefined\cmntsoff\cmntfalse\fi

\ifcmnt
    \providecommand{\aucmnt}[1]{#1}

    \providecommand{\Pcolor}{\color{gray}}
\else
    \providecommand{\aucmnt}[1]{}

    \providecommand{\Pcolor}{}
\fi

\newcommand{\Pc}[1]{\aucmnt{{\Pcolor \textbf{[}Permanent comment: {\color{gray}#1}\textbf{]}}}}

\newcommand{\SI}[1]{\;\mathrm{#1}}

\newcommand{\Prob}{\mathbb{P}}
\newcommand{\Exp}{\mathbb{E}}
\newcommand{\Var}{\mathrm{Var}}

\newcommand{\past}{{\textit{past}}}
\newcommand{\Past}{{\textit{Past}}}

\newcommand{\knuth}[1]{\left\llbracket #1 \right\rrbracket}

\newcommand{\eps}{\epsilon}

\newcommand{\rls}{\mathbb{R}}

\newcommand{\cF}{\mathcal{F}}

\newcommand{\cL}{\mathcal{L}}
\newcommand{\cM}{\mathcal{M}}

\newcommand{\cO}{\mathcal{O}}

\newcommand{\cT}{\mathcal{T}}

\newcommand{\cV}{\mathcal{V}}

\newtheorem{theorem}{Theorem}[section]
\newtheorem*{theorem*}{Theorem}
\newtheorem{proposition}[theorem]{Proposition}
\newtheorem{proposition*}{Proposition}
\newtheorem{lemma}[theorem]{Lemma}
\newtheorem*{lemma*}{Lemma}

\newcounter{protocol}
\renewcommand{\theprotocol}{\Alph{protocol}} 
\newenvironment{protocol}[1][]{
  \refstepcounter{protocol}%
  \begin{mdframed}[linewidth=0.8pt,roundcorner=5pt]
  \noindent\textbf{Protocol~\theprotocol. #1} \par
}{
  \end{mdframed}
}


%% file: spotcheck_est_apq.tex
\title{An efficient method for spot-checking quantum properties with sequential trials: Appendix}
\author{Yanbao Zhang}\thanks{yzh@ornl.gov}
\affiliation{Quantum Information Science Section,  Computational Sciences and Engineering Division,  Oak Ridge National Laboratory, Oak Ridge, Tennessee 37831, USA}
\author{Akshay Seshadri} 
\affiliation{Department of Physics, University of Colorado, Boulder, Colorado 80309, USA}
\affiliation{National Institute of Standards and Technology, Boulder, Colorado 80305, USA}
\author{Emanuel Knill}
\affiliation{National Institute of Standards and Technology, Boulder, Colorado 80305, USA}
\affiliation{Center for Theory of Quantum Matter, University of Colorado, Boulder, Colorado 80309, USA}

\maketitle

\ifSubfilesClassLoaded{%
  \onecolumngrid
  \appendix
  \tableofcontents
}{}

\section{Notation}

 A random variable (RV) is a measurable function from an underlying
measurable space to a measurable space. In this work, RVs need not be
real-valued, and the probability distribution of a given RV may be
different depending on context.  We use the convention that RVs are
denoted by capital letters and their values or instances by the
corresponding lower-case letters.  For simplicity, we assume that the
underlying probability space is discrete, a realistic restriction for
digitized observations in experiments.  This ensures that conditional
probabilities are well-defined, and the probability distributions of
RVs are supported on countable subsets of their range.  For
simplicity, we also assume that singletons are measurable.  We call a
probability distribution on a measurable space discrete if there
exists a countable subset that has probability \(1\).

The experiments considered here consist of sequential trials, where
each trial instantiates one or more RVs, both observed and
unobserved. For our analysis, we fix the total number \(n\) of trials
ahead of time. Trials are indexed by positive integers. The RVs
instantiated at trial \(i\) are denoted by \(X_{i}, Y_{i}, \ldots\).
We use upper-case bold-face letters to denote sequences of RVs. For
example \(\bm{X}=(X_{i})_{i=1}^{n}\). We use the notation
\(\bm{X}_{<i}=(X_{j})_{j=1}^{i-1}\) for the initial segment of
\(\bm{X}\) prior to trial $i$, and similar notation applies to other sequence RVs.  
The ``past'' of the \(i\)'th trial refers to the collection of all RVs that were
instantiated before trial \(i\), whether or not they were observed, and
normally includes \(X_{j},Y_{j},\ldots\) for \(j<i\) as well as any
relevant history prior to the start of the experiment. We use the special
symbol \(\Past_{i}\) to denote the RV that captures the past of trial
\(i\) and \(\past_{i}\) to denote its values.

We typically use the term ``distribution'' to refer to ``probability
distribution'' and denote distributions on measurable spaces by
\(\mu,\nu,\ldots\).  The joint distribution of the RVs over all trials
in a particular situation is usually not explicitly specified. We
define models as sets of such joint distributions. For each scenario
that we consider, we assume a model \(\cM\) and refer to the
distributions in the model as allowed distributions.  The probability
of an event \(\Phi\) is denoted by \(\Prob(\Phi)\), or by
\(\Prob_{\mu}(\Phi)\) if the distribution \(\mu\) needs to be
specified. Conditional probabilities are written as
\(\Prob(\Phi|\Lambda)\) for conditioning on an event \(\Lambda\) and
are defined to be zero if \(\Prob(\Lambda)=0\).  Expectations of an RV
\(Z\) are accordingly denoted by \(\Exp(Z)\) or
\(\Exp_{\mu}(Z)\). Expectations conditional on an event \(\Phi\) are
denoted by \(\Exp(Z|\Phi)\) or \(\Exp_{\mu}(Z|\Phi)\), where we use
the same convention for the case where the conditioning event has
probability zero. For an RV \(Z\), the event \(Z=z\) is abbreviated as
\(z\) when mentioned as an argument to \(\Prob\) or \(\Exp\).  Thus
\(\Prob(z|w)\) is the probability of \(Z=z\) given the event \(W=w\).
We define \(\knuth{\Phi}\) to be the indicator function of the event
\(\Phi\), that is, \(\knuth{\Phi} = 1\) on the event \(\Phi\) and
\(\knuth{\Phi}=0\) on the complement of \(\Phi\).

For RVs \(Z\) and \(W\) where \(Z\) is real-valued, the conditional
expectation of \(Z\) given \(W\), denoted by \(\Exp(Z|W)\), is an RV
whose value is \(\Exp(Z|w)\) when \(W\) takes value \(w\), that is, on
the event \(W=w\). For this purpose, we define \(\Exp(Z|w)=0\) when
\(\Prob(w)=0\). 
\Pc{This is a definition for our purposes that works
  for discrete probability distributions and agrees with the usual one
  a.e.}  Conditional expectations introduce RVs whose values depend on
the underlying distribution. To avoid this distributional dependence,
we exclude such RVs from contributing to \(\Past_{i}\). This ensures
that \(\Past_{i}\), when viewed as a function on the underlying
measurable space, does not depend on the specific probability
distribution under consideration.

\section{Preliminaries}
\label{app:preliminaries}

The scenarios considered in our work involve two RVs in each trial
$i$. The first is the spot-checking RV \(Y_{i}\) with values
\(y_{i}\in\{0,1,\ldots\}\), and the second is the real-valued RV
\(X_{i}\) of interest.  In each trial, \(Y_{i}\) is always observed,
while \(X_i\) is observed and the trial is spot-checked only when
\(Y_{i}=0\).  We define \(U_{i}=X_{i}\) or
\(U_{i}=\Exp(X_{i}|\Past_{i})\), depending on whether we are
interested in the unobserved \(X_{i}\) or their expectations given the
past. We aim to estimate the sum of \(U_{i}\) over the trials where
\(Y_{i}=1\). In contrast to the main text, we allow spot-checking
RVs to take more than two possible values. See
Sect.~\ref{sect:general_spotcheck} for motivation and an application.

When spot-checking quantum
states, unobserved \(X_{i}\) may not be physically realized and may be
unavailable due to a subsequent incompatible measurement of the state
produced in unchecked trial \(i\). Nevertheless, joint distributions
involving the unobserved \(X_{i}\) can be defined provided that it is
possible to virtually timeshift actions involving the unchecked states
until after the experiment, where we virtually invoke perfect quantum
memories for this purpose. For the case
\(U_{i}=\Exp(X_{i}|\Past_{i})\), it is not necessary to define these
joint distributions, as \(\Exp(X_{i}|\Past_{i})\) is well-defined
independent of the actual measurement or process applied to the state
of the \(i\)'th trial. In this case, \(\Past_{i}\) does not include
the uninstantiated values of unobserved \(X_{i}\).

The distributions in the model for a scenario are constrained. Such
constraints include the free-choice assumption as introduced in the
main text and assumptions on the spot-checking probability in each
trial. We require that the spot-checking probability in each trial
lies within the interval \((0,1)\) to ensure non-trivial results. When
the spot-checking probability is independent of the past and
known for all trials, we write \(\omega_i=\Prob(Y_{i}=0)\) for the
probability that we spot-check trial $i$.  
The main additional constraints are on the ranges of
the RVs, where unless mentioned otherwise, \(X_{i}\) is assumed to be
lower-bounded by \(b\). Typically, we can shift \(X_{i}\) so that the
lower bound is \(b=0\).

The problem considered here is to obtain lower confidence bounds on
the sum \(S_{U} =\sum_{i=1}^{n}\knuth{Y_i=1}U_{i}\).  Such lower
bounds are required to be correct with probability near \(1\) as
formalized next: For a given \(\epsilon\in (0,1)\), a random variable
\(S_{\lb}\) is a lower confidence bound on \(S_{U}\) with confidence
level \((1-\epsilon)\) if \(\Prob(S_{\lb} > S_{U}) \leq \epsilon\) for
every allowed distribution.  The parameter \(\epsilon\) is referred to
as the significance level or error bound for the estimate. The
confidence level \((1-\epsilon)\) can be interpreted as a lower bound
on the coverage probability of the confidence set \([S_{\lb},\infty)\)
for \(S_{U}\). In the main text, we used the notation \(S=S_{U}\) if
\(U=\Exp(X|\Past)\) and \(S'=S_{U}\) if \(U=X\).  When necessary, we
explicitly indicate the dependence of \(S\), \(S'\), and \(S_{\lb}\)
on \(n\) by adding the subscript \(n\) to each quantity.  While the
confidence bound \(S_{\lb}\) also depends on \(\epsilon\), we treat
this dependence as implicit throughout the paper. 

Estimation factors as defined in the main text and in the next section
provide lower confidence bounds.  In most cases, we wish to determine
good estimation factors in the sense that the lower confidence bounds
that they yield are as high as possible.  We typically fix the error
bound \(\epsilon\) and then minimize the expected value of
\(\Delta = S_{U}-S_{\lb}\) for the distribution that we expect to see
based on previously obtained information, for example, according to
calibration data or past performance.  The smaller the expectation of
\(\Delta\), the better the performance of the constructed confidence
bound \(S_{\lb}\). We remark that it is in principle possible for the
expectation of a lower confidence bound for \(S_{U}\) to be above that
of \(S_{U}\). However, such pathological behavior does not occur for
the lower confidence bounds obtained with estimation factors; see
Prop.~\ref{thm:soundness}.  The confidence bounds obtained always
satisfy the designed error bounds for every allowed distribution.
However, they generally perform optimally or near-optimally only for
distributions that are close to the one used for optimization. It is
usually the case that the apparatus is intended to produce
distributions for which each trial is independent and identical, where
the trial distributions are either designed or can be calibrated.
Although we may not trust that the apparatus performs ``honestly'' as
expected when it is used in applications, we use the designed or
calibrated trial distributions to optimize the expected confidence
bounds. If the apparatus behaves differently when it matters, the
actual confidence bounds may not be optimal, but we can still
guarantee that they satisfy the designed error bounds.

\section{Estimation-factor theorem and models for spot-checking}
\label{sect:EF_method}

An estimation factor of power \(\beta>0\) for estimating
\(\knuth{Y_i=1}U_{i}\) for model \(\cM\) at trial \(i\) is a
non-negative function \(T_{i}\) of \((\knuth{Y_{i}=0}X_{i}, Y_{i})\)
that may depend on \(\Past_{i}\) and satisfies the estimation-factor
inequality
\begin{align}
  \Exp(T_{i} e^{-\beta \knuth{Y_{i}=1}U_{i}}|\Past_{i})\leq 1
  \label{eq:efineq}
\end{align}
for every allowed distribution in the model \(\cM\). For simplicity of
notation, we identify \(T_{i}(X_{i},Y_{i})\) with
\(T_{i}(\knuth{Y_{i}=0}X_{i},Y_{i})\).  That is, when writing
\(T_{i}(X_{i},Y_{i})\), we implicitly assume that when \(Y_{i}>0\),
\(T_{i}\) does not depend on \(X_{i}\) and
\(T_{i}(X_{i},Y_{i})=T_{i}(0,Y_{i})\).  The terminology ``estimation
factor'' was first introduced in Ref.~\cite{knill:qc2017a}, where the
quantity being estimated can be interpreted as a sum of logarithms of
model-constrained conditional probabilities.  The choice of \(T_{i}\) 
may depend on observed values within \(\Past_{i}\). We can take 
advantage of this dependence to adapt \(T_{i}\) based on 
information gained before trial \(i\).

Estimation factors for the \(n\) trials determine the following lower 
confidence bound \(S_{\lb}\) on \(S_{U}\) at confidence level \((1-\epsilon)\):
\begin{align}
  S_{\lb} = \frac{\sum_{i=1}^{n}\ln(T_{i})+\ln(\epsilon)}{\beta}.
  \label{eq:conf_lb}
\end{align}
This
is a consequence of the estimation-factor theorem for spot-checking:

\begin{theorem}[Estimation-factor theorem] \label{thm:ef_confbound}
  With the notation introduced in the previous paragraphs, for all
  allowed distributions,
  \begin{align}
    \Exp\qty(\knuth{S_{\lb} \geq S_{U}}\Big|\Past_{1}) \leq \epsilon.
  \end{align}
\end{theorem}
In particular, the probability that \(S_{\lb} \geq S_{U}\) 
is at most \(\epsilon\), and
this probability bound holds even when conditioned on any
event that depends only on the past \(\Past_{1}\) of the experiment. 
Since the event  \(S_{\lb} > S_{U}\) implies  \(S_{\lb} \geq S_{U}\),
it follows that \(\Prob(S_{\lb} > S_{U})\leq \epsilon\) for every allowed distribution.
Therefore, \(S_{\lb}\) is a lower confidence bound on \(S_{U}\) with confidence 
level \((1-\epsilon)\).

\begin{proof}
  Let \(S_{k}=\sum_{i=1}^{k}\knuth{Y_i=1} U_{i}\) be the partial sum of the
  \(\knuth{Y_i=1}U_{i}\) and \(\cT_{k}=\prod_{i=1}^{k}T_{i}\) the partial
  product of estimation factors up to trial \(k\). We have
  \(S_{U}=S_{n}\).  We first show that \(\cT_{k}\) satisfies the inequality
  \(\Exp(\cT_{k} e^{-\beta S_{k}}|\Past_{1})\leq 1\) for all \(k\).   We proceed by
  induction. The inequality holds for \(k=1\) by Eq.~\eqref{eq:efineq}. For 
  \(k>1\),
  \begin{align}
    \Exp(\cT_{k} e^{-\beta S_{k}}|\Past_{1})
    &= \Exp(\Exp(\cT_{k-1}e^{-\beta S_{k-1}}T_{k}e^{-\beta \knuth{Y_k=1} U_{k}}|\Past_{k})|\Past_{1})
      \nonumber\\
    &= \Exp(\cT_{k-1}e^{-\beta S_{k-1}}\Exp(T_{k}e^{-\beta \knuth{Y_k=1} U_{k}}|\Past_{k})|\Past_{1})
      \nonumber\\
    &\leq\Exp(\cT_{k-1}e^{-\beta S_{k-1}}|\Past_{1})
      \nonumber\\
    &\leq 1.
  \end{align}
  The first line is due to the tower property of conditional
    expectations.  The second line follows from the first because
  \(\cT_{k-1}e^{-\beta S_{k-1}}\) is determined by
  \(\Past_{k}\).  This is because \(T_{i}\) and \(S_{i}\) for \(i\leq k-1\) 
  are  functions of RVs that are part of
  \(\Past_{k}\). Next, we used the estimation-factor inequality for
  \(T_{k}\) and applied the induction hypothesis.  
  
  When \(k=n\), we have \(\Exp(\cT_{n} e^{-\beta S_{n}}|\Past_{1})\leq 1\).
   We can therefore apply a general form of Markov's inequality to obtain 
  \(\Exp(\knuth{\cT_{n}e^{-\beta S_{n}}\geq 1/\epsilon}|\Past_{1})\leq
  \epsilon\).  For completeness, we prove this form of Markov's
  inequality.  Because
  \(\knuth{\cT_{n}e^{-\beta S_{n}}\geq 1/\epsilon} \leq \epsilon
  \cT_{n}e^{-\beta S_{n}}\), we have
  \begin{align}
    \Exp(\knuth{\cT_{n}e^{-\beta S_{n}}\geq 1/\epsilon}|\Past_{1})
    &\leq \Exp(\epsilon \cT_{n}e^{-\beta S_{n}}|\Past_{1})\nonumber\\
    &= \epsilon\Exp(\cT_{n}e^{-\beta S_{n}}|\Past_{1}) \leq \epsilon.
  \end{align}
  The inequality \(\cT_{n}e^{-\beta S_{n}}\geq 1/\epsilon\) is
  equivalent to the inequality
  \(S_{n}\leq \ln(\cT_{n}\epsilon)/\beta\), where the right-hand side
  is equal to \(S_{\lb}\) as defined in Eq.~\eqref{eq:conf_lb}.  This
  completes the proof of the theorem.
\end{proof}

Joint distributions of \(Y_{i},X_{i}\) are defined as distributions on
\(\cV=\{0,1,\ldots\}\times\rls\), the value space of
\((Y_{i},X_{i})\).  For a given global distribution \(\mu\) that
governs all trials in the experiment, we define \(\cM_{\mu, i}\) as
the set of distributions \(\nu\) on \(\cV\) for which there exists a
value \(\past_{i}\) of \(\Past_{i}\) that occurs with non-zero
probability according to \(\mu\), such that for all measurable subsets
\(\Phi\) of \(\cV\),
\(\Prob_{\nu}(\Phi)=\Prob_{\mu}((Y_{i},X_{i})\in\Phi|\past_{i})\).

In this work, the model \(\cM\) of the experiment is derived from a given
 set \(\cF\) of discrete distributions on \(\cV\). Given \(\cF\), we define
 \(\cM=\cM_{\cF}\) as the set of all discrete global distributions \(\mu\) 
 for which \(\cM_{\mu,i}\subseteq\cF\) for all \(i\). In general,  \(\cF\) 
 could depend on the trial index, but we do not require this here.  
 According to the following lemma, checking the
estimation-factor inequality for all distributions in \(\cM_{\cF}\) reduces to checking
unconditional estimation-factor inequalities for distributions in \(\cF\). For the rest
of this section, we consider a generic trial and omit trial indices for simplicity.

\begin{lemma}\label{lem:ucond_efineq}
  Let \(\cF\) be a set of discrete distributions on
  \(\cV=\{0,1,\ldots\}\times\rls\), \(\cM=\cM_{\cF}\) the derived
  model of the experiment, and \(T\) a non-negative function of
  \((Y,\knuth{Y=0}X)\).  Then \(T\) is an estimation factor of
    power \(\beta\) for estimating \(\knuth{Y=1}X\) for model \(\cM\)
    at each trial iff for all \(\nu\in\cF\),
    \begin{align}
      \Exp_{\nu}\qty(T e^{-\beta \knuth{Y=1}X})\leq 1.
      \label{eq:uncond_efineq}
    \end{align}
    \(T\) is an estimation factor of power \(\beta\) for estimating
    \(\knuth{Y=1}\Exp(X|\Past)\) for model \(\cM\) at each trial iff
    for all \(\nu\in\cF\),
    \begin{align}
      \Exp_{\nu}\qty(T e^{-\beta \knuth{Y=1}\Exp_{\nu}(X)})\leq 1.
      \label{eq:uncond_efineqE}
    \end{align}  
\end{lemma}
If \(T\) satisfies Eq.~\eqref{eq:uncond_efineq} or Eq.~\eqref{eq:uncond_efineqE}
for all \(\nu\in\cF\), we say that \(T\) is an estimation factor of power \(\beta\) 
for estimating \(\knuth{Y=1} U\) with \(U=X\) or \(U=\Exp(X|\Past)\), respectively, 
for the model \(\cF\) of the generic trial.

We also remark that if Eq.~\eqref{eq:uncond_efineq} is satisfied for
all \(\nu\in\cF\), then it is satisfied for all \(\nu\)
expressible as convex combinations of distributions in \(\cF\).  This
is because the product \(T e^{-\beta \knuth{Y=1}X}\), as a function on
the underlying measurable space, does not depend on the associated
distribution \(\nu\), and so the expectation
\(\Exp_{\nu}\qty(T e^{-\beta \knuth{Y=1}X})\) is linear in the
probabilities of events according to \(\nu\).  The convex combination
may involve a countable number of terms.

\begin{proof}
  Let \(U=X\) or \(U=\Exp(X|\Past)\).  The estimation-factor
  inequality in Eq.~\eqref{eq:efineq} is equivalent to the statement
  that for all \(\mu\in\cM\) and all values \(\past\) of \(\Past\)
  with \(\mu(\past)>0\), we have
  \(\Exp_{\mu}(T e^{-\beta \knuth{Y=1}U}|\past)\leq 1\).  For a given
  \(\mu\) and value \(\past\) of \(\Past\), let \(\nu\) be the
  distribution of \((Y,X)\) with respect to \(\mu\) conditional on
  \(\Past=\past\).  Then \(\nu\in\cF\) and
  \(\Exp_{\mu}(T e^{-\beta \knuth{Y=1}U}|\past)=\Exp_{\nu}(T e^{-\beta
    \knuth{Y=1}V})\), where \(V=X\) for \(U=X\) and
    \(V=\Exp_{\nu}(X)\) for \(U=\Exp(X|\Past)\).  Consequently, if
  Eq.~\eqref{eq:uncond_efineq} or Eq.~\eqref{eq:uncond_efineqE}
  holds for all \(\nu\in\cF\), then so does the corresponding 
  estimation-factor inequality in Eq.~\eqref{eq:efineq}.  For the
  converse, it suffices to observe that for every (discrete)
  \(\nu\in\cF\) and value \(\past\) of \(\Past\), there exists a
  distribution \(\mu\in\cM\) such that the distribution of \((Y,X)\)
  given \(\Past=\past\) is \(\nu\).  For example, we can choose the
  distribution \(\mu\) such that each trial is independent and
  identical with \(\cM_{\mu,i}=\nu\) for all \(i\).
\end{proof}

The model \(\cF\) that we consider in this work consists of all
discrete distributions for which \(Y\) and \(X\) are independent, the
probability of spot-checking (that is, of \(Y=0\)) is fixed to be
\(\omega\in (0,1)\), and \(X\) is lower-bounded by \(b\).  We denote
this model by \(\cF_{\omega}\). The independence of \(Y\) and \(X\) is
equivalent to the condition that the distributions in the derived
model of the experiment, \(\cM_{\cF}\), satisfy the free-choice
assumption.  It is convenient to standardize the lower bound to
\(b=0\) by shifting \(X\).  For this purpose, we observe that \(T\) is
an estimation factor for estimating \(U\) iff
\(Te^{-\beta b \knuth{Y=1}}\) is an estimation factor for estimating
\((U-b)\). Therefore, when constructing estimation factors or checking
the estimation-factor inequality, we can, without loss of generality,
use the shifted RVs.  The confidence bounds obtained for the sum of
the shifted RVs and for the sum of the original RVs differ by
\(\sum_{i=1}^{N}b \knuth{Y_{i}=1}\), which is random but directly
observed in the experiment. Moreover, to optimize and compare these
bounds, we usually focus on the difference between \(S_{U}\) and
\(S_{\lb}\), which cancels out the offset from shifting the RVs. In
this case, optimizing and comparing confidence bounds for the original
and shifted RVs are equivalent. For these reasons, unless otherwise
stated, we assume \(b=0\) for the remainder of the appendix.

We end this section by formalizing the observation that for
any set of discrete distributions on \(\cV=\{0,1,\ldots\}\times\rls\) 
as specified in the lemma below, such as \(\cF_{\omega}\),
whether \(T\) is an estimation factor does not depend
on whether \(U=X\) or \(U=\Exp(X|\Past)\). 

\begin{lemma}\label{lem:Uequiv}
  Let \(I\) be an interval of the form \(I=[b,c]\) or
  \(I=[b,\infty)\), and let \(J\) be a subset of \([0,1]\).  Define \(\cF\) as
  the set of all discrete distributions on \(\cV=\{0,1,\ldots\}\times\rls\) for which the
  values of \(X\) are in \(I\), and \(Y\) is independent of \(X\) with
  \(\Prob(Y=0)\in J\).  Let \(\cM=\cM_{\cF}\) be the derived model 
  of the experiment, and let \(T\) be an
  estimation factor of power \(\beta\) for estimating
  \(\knuth{Y=1}U\).  Then, \(T\) is an estimation factor with
  \(U=X\) iff it is an estimation factor with
  \(U=\Exp(X|\Past)\). 
\end{lemma}
   In this lemma, we allow the spot-checking probability
  \(\Prob(Y=0)\) to be arbitrary within a given subset of
  \([0,1]\). We exploit this option in Sect.~\ref{sect:EFs_with_bias}.

\begin{proof}
  We freely use Lem.~\ref{lem:ucond_efineq} to drop conditioning on
  the past and check estimation-factor inequalities on only
  distributions \(\nu\in\cF\).  Let \(T\) be an estimation factor with
  \(U=\Exp(X|\Past)\) for \(\cM\).  Given \(\nu\in\cF\), construct the
  distribution \(\nu'\) for which \(X\) is deterministic according to
  \(\Prob_{\nu'}(X=\Exp_{\nu}(X))=1\) and for all values \(y\) of
  \(Y\), \(\Prob_{\nu'}(y)=\Prob_{\nu}(y)\).  The distribution
  \(\nu'\) is also in \(\cF\) because \(I\) is an interval. Every
  distribution in \(\cF\) for which \(X\) is deterministic can be
  obtained with this construction for some \(\nu\in\cF\).  Let \(\cF'\) be the set
  of such deterministic distributions.  For every \(\nu'\in\cF'\),
  Eq.~\eqref{eq:uncond_efineq} holds with \(U=\Exp_{\nu'}(X)\).
  Furthermore, the left-hand side of the inequality in
  Eq.~\eqref{eq:uncond_efineq} is unchanged under the replacement of
  \(\Exp_{\nu'}(X)\) by \(X\), since \(X\) is deterministic.  We
  conclude that if \(T\) is an estimation factor with
  \(U=\Exp(X|\Past)\) for \(\cM\) and thus for \(\cF\), then it
  is an estimation factor with \(U=X\) for \(\cF'\).  We have that
  \(\cF\) is contained in the set of all convex combinations of
  distributions in \(\cF'\), where the convex combination may involve
  a countable number of terms.  Since Eq.~\eqref{eq:uncond_efineq} is
  preserved under such convex combinations (see the remark after
  Lem.~\ref{lem:ucond_efineq}), we can conclude that \(T\) is an
  estimation factor with \(U=X\) for \(\cF\). Consequently, by
  Lem.~\ref{lem:ucond_efineq}, \(T\) is also an estimation factor with
  \(U=X\) for \(\cM\).
  
  For the converse, consider an estimation factor \(T\) with \(U=X\).
  From Eq.~\eqref{eq:uncond_efineq}, for all \(\nu\in\cF\)
  \begin{align}
    1&\geq \Exp_{\nu}\qty(T(X,Y) e^{-\beta \knuth{Y=1}X})
       \notag\\
     &=\Exp_{\nu}(\knuth{Y= 0}T(X,0)) + T(0,1) \Exp_{\nu}\qty(\knuth{Y=1}e^{-\beta X})
       + \Exp_{\nu}\qty(\knuth{Y>1} T(0,Y)) \notag \\
     & =\Exp_{\nu}(\knuth{Y= 0}T(X,0)) + \Prob_{\nu}(\knuth{Y=1})T(0,1) \Exp_{\nu}\qty(e^{-\beta X})
       + \Exp_{\nu}\qty(\knuth{Y>1} T(0,Y)).
  \end{align}
  The second line follows from the fact that the estimation factor \(T\) depends on \(X\) only when \(X\) is spot-checked, that is, when \(Y=0\).
  The last line follows from the independence of \(X\) and \(Y\) according to \(\nu\). Because \(T\geq 0\), we have \(T(0,1)\geq 0\). 
  By the convexity of the function \(f(x)=e^{-x}\) and Jensen’s inequality, \(\Exp_{\nu}\qty(e^{-\beta X})\geq e^{-\beta \Exp_{\nu}(X)}\).    
   Consequently, continuing from the last line we obtain
  \begin{align}
    1 &\geq \Exp_{\nu}(\knuth{Y= 0}T(X,0)) + \Prob_{\nu}(\knuth{Y=1})T(0,1) e^{-\beta \Exp_{\nu}(X)}
        + \Exp_{\nu}\qty(\knuth{Y>1} T(0,Y))    \notag \\  
      &= \Exp_{\nu}(\knuth{Y= 0}T(X,0)) + T(0,1) \Exp_{\nu}\qty(\knuth{Y=1} e^{-\beta \Exp_{\nu}(X)})
        + \Exp_{\nu}\qty(\knuth{Y>1} T(0,Y)) \notag\\
      &= \Exp_{\nu}\qty(T(X,Y) e^{-\beta\knuth{Y=1}\Exp_{\nu}(X)}).
  \end{align}
  It follows that \(T\) is an estimation factor with \(U=\Exp(X|\Past)\).
\end{proof}

Given that the sets of distributions on \(\cV\) considered in this
work satisfy the conditions of Lem.~\ref{lem:Uequiv}, we consider only the case
\(U=X\) for the remainder of this appendix.

\section{Constructing estimation factors for spot-checking: Given a reference distribution of \texorpdfstring{\(X\)}{X}}
\label{sect:numerical_opt}

From now on, unless stated otherwise, we assume that in each trial
\(i\), \(X_i\) is lower-bounded by \(0\) and \(Y_i\) is a $0/1$-valued
binary RV.  Thus, \(\knuth{Y_i=1}=Y_i\) and \(\knuth{Y_i=0}=(1-Y_i)\).
Because of Lem~\ref{lem:Uequiv}, without loss of generality, we
consider only the case \(U=X\). 
 
Since in most practical situations the spot-checking probability is
constant across trials, we consider the spot-checking scenario with
\(\Prob(Y_{i}=0)=\omega\) for all \(i\).  We show how to optimize
estimation factors prior to an experiment, given the significance
level \(\epsilon\) and a target (or reference) distribution for the
\(X_{i}\).  We consider two goals.  The first is to minimize the
expected gap \(\Exp(S_{U}-S_{\lb})\) between the estimated quantity
\(S_{U}\) and the confidence bound \(S_{\lb}\) obtained for a
given number \(n\) of trials. The second is to minimize the required
number of trials, \(n\), that ensures the expected gap per unchecked
trial smaller than or equal to a given threshold \(\delta_{\thresh}\).
For these goals, we examine the performance under the assumption that
trials are independent and identically distributed (i.i.d.) with the
target distribution, as most quantum experiments aim for such
trials. At the same time, we emphasize that the constructed estimation
factors remain valid and applicable even when the trials to be
analyzed are not i.i.d., with past-conditional distributions that can
differ from the target distribution. If the past-conditional
distribution of \(X_{i}\) in each experimental trial is close to the
target distribution, the constructed estimation factors are expected
to achieve near-optimal performance.  Thus, one practical approach for
choosing the target distribution and thus the estimation factor to be
used prior to conducting the experiment is to perform calibration runs
or theoretical simulations based on predictive models.

In this work, we are evaluating the performance of estimation factors
by the expectation of the confidence bound obtained. This evaluation
does not take into consideration the distribution of the confidence
bound, so one cannot conclude that typically obtained confidence
bounds are good without additional arguments. We leave this for future
work, but note here that for the assumptions involving i.i.d. samples
from the reference distribution and estimation factors that do not
depend on the trial, the confidence bound is a sum of independent
RVs. This suggests that the confidence bound concentrates around its
expected value with a distribution that has relative width
\(1/\sqrt{n}\). To quantify the concentration requires estimating the
variance of the logarithm of the estimation factor used. If estimation
factors are derived from calibration trials, similar concentration
results can be applied to the calibration trials.

\subsection{Estimation factors for minimizing the expected gap}
\label{sect:numerical_opt_smax}

For the first goal, we need to minimize the expected gap by maximizing
the expectation of the confidence bound \(S_{\lb}\) on
\(S_{U}=\sum_{i=1}^{n}Y_{i}X_{i}\), where the \(X_{i}\) have the
target distribution independent of the past. Because
the trials are i.i.d. and \(S_{\lb}\) is a sum over trials, we can
choose the same estimation factor for each trial independent of the
past.  We then optimize the contribution to \(S_{\lb}\) from each
trial separately and, therefore, consider a generic trial, dropping
the trial indices and conditioning on the past.  Let \(\nu\) be the
joint distribution of \(Y\) and \(X\) in a generic trial, where \(Y\)
and \(X\) are independent with \(\Prob(Y=0)=\omega\) and \(X\)
following the target distribution.  We use the expectation of
\(S_{\lb}/n\) as the objective function in the optimization problem:
\begin{align}
  \text{Maximize over \(\beta, T\):\ }
  & \frac{1}{\beta}\qty(\Exp_{\nu}(\ln(T(X,Y)))+\ln(\epsilon)/n), 
    \nonumber\\
  \text{subject to:\ }
  &
    \beta> 0,\; T\geq 0,\;\text{and for all \(\nu'\in\cF_{\omega}\):\ }
    \Exp_{\nu'}\qty(T(X,Y)e^{-\beta YX})\leq 1.
    \label{prob:beta-t_optimization}
\end{align}
As shown in the main text, it suffices to consider estimation factors of the form
\begin{align}
  T(X,Y) &= (1-Y)\qty(1-(1-\omega)te^{-\beta X}) /\omega+ Yt,
           \label{eq:t_maximal}
\end{align}
where \(t\in[0, 1/(1-\omega)]\) and $\beta>0$. This reduces the dimensionality of
the optimization problem to two, with the feasible region for \(\beta\) and \(t\) forming
a convex set. Substituting the estimation factor of Eq.~\eqref{eq:t_maximal} in
Problem~\eqref{prob:beta-t_optimization}, we obtain for the objective function
\begin{align} \cO(\beta,t)
  &=
    \frac{1}{\beta}\qty(\omega\qty(\Exp_{\nu}\qty(\ln(1-(1-\omega)te^{-\beta X})) -\ln(\omega)) +
    (1-\omega)\ln(t) + \ln(\epsilon)/n).
    \label{eq:obetat}
\end{align}

If desired, estimation factors can be adapted during the run of
sequential trials.  To do this, the power \(\beta\) must be fixed
independently of the sequential trial results to be analyzed so that
the confidence bound \(S_{\lb}\) is valid.  Adaptation therefore
involves optimizing over \(t\) while keeping \(\beta\) fixed.  In view
of this, we solve the optimization
Problem~\eqref{prob:beta-t_optimization} in two steps: First, we
optimize over \(t\) for a fixed \(\beta\), and then we optimize over
\(\beta\). For the first step, we modify the objective function above
to retain only the terms that depend on \(t\).  The optimization
problem now is
\begin{align}
  \textrm{Maximize over \(t\):\ }
  &
    \cO(t) = \omega\Exp_{\nu}\qty(\ln(1-(1-\omega)te^{-\beta X})) +
    (1-\omega)\ln(t),
    \nonumber\\
  \textrm{subject to:\ }
  &
    t\in I\coloneqq  [0, 1/(1-\omega)].
    \label{eq:ot}
\end{align}

The objective function \(\cO(t)\) is an expectation over \(X\) of the function
\begin{align}
  h(t; X)&=\omega\ln(1-(1-\omega)t e^{-\beta X}) + (1-\omega)\ln(t).
\end{align}
Here, we use the notation \(h(t; X)\) to represent a function of \(t\)
that is parameterized by \(X\).  For any fixed \(X=x\geq 0\), the
function \(h(t; x)\) is strictly concave on the interval
\(J \coloneqq \qty(0,e^{\beta x}/(1-\omega))\). This is the interval on
  which \(h(t;x)\) is defined and differentiable, since for \(t\)
  outside this interval, at least one of the arguments of the
  logarithms is non-positive. The function \(h(t;x)\) diverges to
  \(-\infty\) at the boundaries of \(J\).  
Within this interval, \(h(t;x)\) achieves its maximum at
\(t_{1}\coloneqq e^{\beta x}\).  To verify these claims, we compute the
first derivative
\begin{align}
  h'(t;x) &=
    -\frac{\omega(1-\omega)}{e^{\beta x}-(1-\omega)t}
    + \frac{1-\omega}{t}.
    \label{eq:h't}
\end{align}
This is a sum of two strictly decreasing terms on \(J\), so \(h(t;x)\)
is strictly concave. One can see that when \(t=t_{1} \in J\),
\(h'(t;x)=0\), so \(h(t;x)\) achieves its maximum on \(J\) at
\(t_{1}\). In Problem~\eqref{eq:ot}, \(t\) is constrained to the interval
\(I=[0, 1/(1-\omega)]\).  If \(t_{1}\) is greater than
\(1/(1-\omega)\), the upper endpoint of \(I\), the maximum of
\(h(t;x)\) on \(I\) occurs at \(1/(1-\omega)\). In general, the
maximum of \(h(t;x)\) on \(I\) is achieved at
\(t_{0}\coloneqq \min(e^{\beta x}, 1/(1-\omega))\).  Since \(x\geq 0\)
and \(\beta, \omega>0\), we always have \(t_{0}\geq 1\).  The maximum
of \(h(t;x)\) on \(I\) satisfies
\begin{align}
  h(t_0; x) & = \omega\ln(1-(1-\omega)t_0 e^{-\beta x}) + (1-\omega)\ln(t_0 ) \nonumber \\
            & \leq (1-\omega)\ln(t_0 ) \nonumber \\
            & = (1-\omega)\ln(\min(e^{\beta x}, 1/(1-\omega)) ), \nonumber \\
            & \leq -(1-\omega)\ln(1-\omega),
\end{align} 
where the second line follows from the fact that
\((1-\omega)t_0 e^{-\beta x}\geq 0\), which implies
\(\ln(1-(1-\omega)t_0 e^{-\beta x})\leq 0\), and the last line follows
from the monotonicity of the logarithm. In view of the above and
  the fact that \(\cO(t)\) is a convex combination of
  \(h(t;x)\) over a distribution of \(x\), \(\cO(t)\) is strictly
concave with non-negative derivative at \(t=1\), and
therefore reaches its maximum on \(I\) at some \(t\geq
1\). Furthermore, for all \(t\in I \), we have
\begin{align}
  \cO(t)= \Exp_{\nu} \qty(h(t;X)) & \leq \Exp_{\nu} \qty(h(t_0;X)) \nonumber \\
  &\leq -(1-\omega)\ln(1-\omega).  
    \label{eq:otbnd}
\end{align}

The strict concavity of \(\cO(t)\) implies that if the expectation
\(\Exp_{\nu}\qty(\ln(1-(1-\omega)te^{-\beta X}))\) can be evaluated
efficiently, then numerically optimizing over \(t\) is also efficient.
To show this we analyze the behavior of \(\cO(t)\).  The expression
for \(\cO(t)\) in Problem~\eqref{eq:ot} shows that
\(\cO(t)\to -\infty\) as \(t \to 0\).  Let \(t_{\sup}\) be the
supremum of values for \(t\) such that
\(1-(1-\omega)te^{-\beta X} > 0\) almost surely with respect to
\(\nu\).  \(\cO(t)\) may converge to a finite value as \(t\)
approaches \(t_{\sup}\) from below.  Since \(\cO(t)\) is strictly
concave and diverges to \(-\infty\) at \(t=0\), it possesses a unique
global maximum \(t_{\beta} \in (0, t_{\sup}]\).  The maximum is either
at \(t_{\sup}\) or at the root of the first-order optimality condition
\(\cO'(t)=\Exp_{\nu} \qty(h'(t;X))=0\).  By verifying the conditions
for interchanging differentiation and expectation on compact sub-intervals
of \((0,t_{\sup})\), it can be shown that
\(\cO'(t)=\Exp_{\nu} \qty(h'(t;X))\) is well-defined and continuous on
\((0,t_{\sup})\). See, for example, Ref.~\cite{folland1984real}, Thm.~2.27.
In view of Eq.~\eqref{eq:h't}, the first-order
optimality condition can be rearranged as:
\begin{align}
  \Exp_{\nu}\qty(\frac{\omega t}{e^{\beta X}-(1-\omega)t}) &= 1. \label{eq:solve_h't}
\end{align}
Inspecting Eq.~\eqref{eq:solve_h't}, the term inside the expectation
is monotonically increasing in \(t\). Furthermore, since
\(e^{\beta X} \geq 1\), this term is strictly less than \(1\) when
\(t <1\). Consequently, the condition in Eq.~\eqref{eq:solve_h't} can
only be satisfied for \(t_{\beta} \geq 1\). We therefore have
\(t_{\beta}\in[1,t_{\sup}]\).  We now impose the constraint
\(t \leq 1/(1-\omega)\) from Problem~\eqref{eq:ot}. Since \(\cO(t)\)
increases monotonically on \((0, t_{\beta}]\), the optimal solution
\(t'_{\beta}\) to the problem constrained to the feasible interval is
\(t'_{\beta}=\min(t_{\beta}, 1/(1-\omega)) \geq 1\).

To optimize \(\cO(\beta,t)\) over \(\beta\) and \( t\), it suffices to
optimize \(\cO(\beta,t'_{\beta})\) over \(\beta\). To facilitate this
optimization, we first establish upper bounds on \(\cO(\beta,t)\). For
this, we substitute the upper bound on \(\cO(t)\) from
Eq.~\eqref{eq:otbnd} into the expression for \(\cO(\beta,t)\) to
obtain
\begin{align}
  \cO(\beta,t) &= \frac{1}{\beta}\qty(\cO(t) - \omega\ln(\omega) + {\ln(\epsilon)}/{n})\nonumber\\
               &\leq \frac{1}{\beta}\qty(-(1-\omega)\ln(1-\omega) - \omega\ln(\omega) + {\ln(\epsilon)}/{n})\nonumber\\
               &= \frac{H_{e}(\omega)+\ln(\epsilon)/n}{\beta},
                 \label{eq:obtbnd1}
\end{align}
where
\(H_{e}(\omega)=-(1-\omega)\ln(1-\omega) - \omega\ln(\omega)\). 
Another upper bound on \(\cO(\beta,t)\) can be derived by noting that the function
\(f(x)=\ln(1-(1-\omega)t e^{-\beta x})\) is concave in \(x\).  To see this,
compute the first derivative \(f'(x)=(1-\omega)\beta t/(e^{\beta x} - (1-\omega)t)\), which is 
monotonically decreasing in \(x\). By concavity and Jensen's inequality,
\begin{align}
  \cO(\beta,t) &\leq \frac{1}{\beta}\qty(\omega\ln(1-(1-\omega)t e^{-\beta\Exp_{\nu}(X)})-\omega\ln(\omega)+(1-\omega)\ln(t) + \ln(\epsilon)/n), \nonumber \\
               & = \frac{1}{\beta}\qty(h(t; \Exp_{\nu}(X))-\omega\ln(\omega)+ \ln(\epsilon)/n).  
\end{align}
As before, \(h\qty(t; \Exp_{\nu}(X))\) is maximized at
\(t=e^{\beta \Exp_{\nu}(X)}\). Therefore,
\begin{align}
  \cO(\beta,t) &\leq \frac{1}{\beta}\qty((1-\omega)\beta\Exp_{\nu}(X) + \ln(\epsilon)/n)
                 \nonumber\\
               &\leq (1-\omega)\Exp_{\nu}(X).
                 \label{eq:obtbnd2}
\end{align}
This implies that \(\cO(\beta,t)\) is bounded from above by the
expectation of \(YX\). Since \(\cO(\beta,t)\) is the additive
contribution from a generic trial to the expectation of the lower
confidence bound \(S_{\lb}\), the upper bound in
Eq.~\eqref{eq:obtbnd2} can be converted to an upper bound on the
expectation of \(S_{\lb}\). Reassuringly, the expectation of
\(S_{\lb}\) is bounded from above by that of \(S_{U}\), which is a
desirable property, though it is not necessarily satisfied by lower
confidence bounds in general.  This property is formally established in
Prop.~\ref{thm:soundness}.

We can now find the optimal value of \(\beta\) by maximizing
\(\cO(\beta, t'_{\beta})\) over \(\beta\in(0,\infty)\).  Provided that the
 expectation \(\Exp_{\nu}\qty(\ln(1-(1-\omega)te^{-\beta X}))\) and
 the solution \(t'_{\beta}\) can be evaluated efficiently, this maximization
 reduces to a one-dimensional numerical search.  The
  analyses in the subsequent sections show that the optimal value of
  \(\beta\) is small for typical distributions and reasonable values
  of \(n\), so in most cases the search can be confined to the
  interval \((0,1]\).  The following considerations can help to
  confine the search more generally. Substituting \(t=1\) in
  \(h(t;x)\) shows that
  \(h(1;x) = \omega\ln(1-(1-\omega)e^{-\beta x}) \geq
  \omega\ln(\omega)\).  Therefore
  \(\cO(\beta,t_{\beta}) \geq \frac{1}{n\beta}\ln(\epsilon)\), which
  goes to zero from below as \(\beta\rightarrow\infty\).  Since the
  analytic upper bound in Eq.~\eqref{eq:obtbnd1} also goes to zero as
  \(\beta\rightarrow\infty\), so does \(\cO(\beta,t'_{\beta})\).  If
  the numerator of the analytic upper bound in Eq.~\eqref{eq:obtbnd1} 
  is non-positive, both the upper and the lower bounds on  
  \(\cO(\beta,t'_{\beta})\) monotonically approach zero from below 
  as  \(\beta\rightarrow\infty\), which implies that the optimal value of \(\beta\) 
  diverges.  However, because the RV \(X\) is bounded from below by \(0\), 
  non-positive lower confidence bounds are trivial. If, in our search for the
  optimal value of \(\beta\), we find a value \(\beta_{1}\) of
  \(\beta\) for which \(0<o_{1}\coloneqq   \cO(\beta_{1},t'_{\beta_{1}})\),
  we can establish a rigorous cutoff as follows.  First, the existence of \(o_{1} > 0\) 
  implies that the numerator of the analytic upper bound in Eq.~\eqref{eq:obtbnd1},
  that is, \(\qty(H_{e}(\omega)+\ln(\epsilon)/n)\), must be positive.   
   Second, this analytic upper bound guarantees that for all 
   \(\beta > \beta_u\), the objective function \(\cO(\beta, t'_{\beta})<o_{1}\), 
   where \(\beta_u\) is determined by \((H_{e}(\omega)+\ln(\epsilon)/n)/\beta_u = o_{1}\). 
   Consequently, the search for the optimal value of \(\beta\)  
   can be safely confined to the interval \((0, \beta_u]\).

\subsection{Estimation factors for minimizing the number of trials}
\label{sect:numerical_opt_nmin}
 
Next we address the second goal of minimizing the number of trials,  \(n\),
required to ensure that the expected gap is at most
\(n(1-\omega)\delta_{\thresh}\), where \(\delta_{\thresh}>0\) is a fixed threshold.
We refer to \(\delta_{\thresh}\) as
the expected gap per unchecked trial.  For this, we choose 
estimation factors to minimize \(n\) subject to the inequality constraint 
\begin{align}
  0&\geq \Exp\qty(S_{U}-S_{\lb}-\delta_{\thresh}\sum_{i=1}^{n}Y_{i})\nonumber\\
   &=n\qty((1-\omega)\Exp_{\nu}(X) - \cO(\beta,t) - (1-\omega)\delta_{\thresh}).
     \label{eq:gapbnd}
\end{align}
Write \(n G(\beta,t,n)\) for the expression on the right-hand side of this equation. 
  To check whether this inequality can be satisfied for a given \(n\), it
suffices to maximize \(\cO(\beta,t)\) as before to determine the
minimum of \(G(\beta,t,n)\) over all feasible \(\beta\) and \(t\).
This minimum depends on \(n\) and is denoted by \(G(n)\).  The only
term in \(G(\beta,t,n)\) that depends on \(n\) is \(\cO(\beta,t)\),
specifically, the term \(\ln(\epsilon)/(n\beta)\) in the expression
for \(\cO(\beta,t)\) given in Eq.~\eqref{eq:obetat}.  Because
\(\ln(\epsilon)/(n\beta)\) increases with \(n\), the maximum of
\(\cO(\beta,t)\) also increases with \(n\), making \(G(n)\) a
decreasing function of \(n\).  As a result, the minimum \(n\) for
which \(G(n)\leq 0\) and the corresponding \EF{} can be determined quickly. 
Since the construction of $\EF$s for minimizing the number 
of trials builds directly on that for minimizing the expected gap,  
in the following sections we focus on the latter and 
illustrate the performance of the resulting $\EF$s.

\section{Constructing estimation factors for spot-checking: Given a reference mean and variance of \texorpdfstring{\(X\)}{X}}
\label{sect:analytical_opt}

We show that estimation factors can be constructed using only the mean
and variance of the reference distribution.  This avoids having to
 know or learn the complete target distribution and reduces dependence 
 of the estimation factor's performance on other features of the distribution.
This approach typically has minimal impact on performance, as measured
by the expectation of the resulting confidence bound \(S_{\lb}\). To
explain this phenomenon, we derive upper bounds on the expectation of
the gap \(\Delta=S_{U}-S_{\lb}\) between the quantity being estimated
and the resulting confidence bound, given that the target distribution
of \(X\) has mean \(\theta\) and variance \(\sigma^{2}\).  The upper
bound derived is of order \(\sqrt{n}\); see Eq.~\eqref{eq:expgap}
below.  General statistical considerations indicate that the order
  of this bound cannot be improved.  In this sense, we say that the
confidence bound \(S_{\lb}\) is asymptotically tight.

As in the previous section, we consider the case of i.i.d. trials,
where at each trial the pair \((Y,X)\) is drawn from a fixed
distribution \(\nu\). Under the distribution \(\nu\), \(Y\) and \(X\)
are independent, \(\Prob(Y=0)=\omega\), and \(X\) follows the target
distribution with mean \(\theta\) and variance \(\sigma^{2}\). For
this case, the expectation \(\Exp_{\nu}(\Delta)\) is the same
regardless of whether \(U=\Exp(X|\Past)\) or \(U=X\).  To upper-bound
\(\Exp_{\nu}(\Delta)\), we assume that \(\theta\) and \(\sigma^{2}\)
are finite. Since the case \(\sigma^{2}=0\) is trivial, for much of
the analysis below, we also assume \(\sigma^{2}>0\).  We show that the
upper bound on \(\Exp_{\nu}(\Delta)\) obtained is quite insensitive to
discrepancies between the assumed (or estimated) and true values of
the mean and variance.  We remark that the bounds derived in this
section remain valid even if the distribution of \(X\) varies from
trial to trial, as long as the mean of \(X\) is fixed at \(\theta\)
and the variance is bounded from above by \(\sigma^{2}\).  This
observation is relevant to the results in Sect.~\ref{sect:tightness}.

Given \(\beta\) and \(t\), we have
\(\Exp_{\nu}(\Delta)=n((1-\omega)\theta - \cO(\beta,t))\) with
\(\cO(\beta,t)\) defined in Eq.~\eqref{eq:obetat}.  The only term of
\(\Delta\) that depends on the distribution of \(X\) but is not
determined by the mean \(\theta\) is the term
\(\Exp_{\nu}\qty(\ln(1-(1-\omega)t e^{-\beta X}))\) in \(\cO(\beta,t)\).
With foresight, we introduce the parameter \(\theta_{e}\geq 0\) and define
\(t=e^{\beta\theta_{e}}\geq 1\).  The parameter \(\theta_{e}\) serves
as an estimate of the mean \(\theta\). In view of the upper bound on
\(t\) when optimizing \(\cO(\beta,t)\), we also constrain
\(t< 1/(1-\omega)\). While \(t=1/(1-\omega)\) yields a valid
estimation factor, we exclude this boundary point as it induces
singular behavior in the subsequent analysis.  We define the function
\( \cL(\beta,\theta_{e})\) as follows:
\begin{align}
  \cL(\beta,\theta_{e})
  &= \beta(1-\omega)\theta - \beta \cO(\beta,t) + \ln(\epsilon)/n
    \nonumber\\
  &=
    \omega\qty(\ln(\omega) - \Exp_{\nu}\qty(\ln(1-(1-\omega)e^{-\beta (X-\theta_{e})})))
    + \beta(1-\omega)(\theta-\theta_{e}),
    \label{eq:def_lbetatheta}
\end{align}
so that
\begin{align}
  \Exp_{\nu}(\Delta)
  &=\frac{1}{\beta}\qty(n\cL(\beta,\theta_{e})  - \ln(\epsilon)).
    \label{eq:expdelta_lbetatheta}
\end{align}
\begin{lemma}\label{lem:lbetatheta_bounds}
  Define
  \(B(z_{0})=\omega^{2}e^{z_{0}}/\qty(1-(1-\omega)e^{z_{0}})^{2}\), where \(z_{0} \in [0, \ln\qty(1/(1-\omega)))\). Then
  \begin{align}
    0 \leq \cL(\beta,\theta_{e}) \leq
    \frac{\beta^{2}(1-\omega)}{2\omega} B(\beta\theta_{e})\qty(\sigma^{2}+(\theta-\theta_{e})^{2}).
    \label{eq:lem_lbnds}
  \end{align}
\end{lemma}

The function \(B(z_{0})\) is non-negative and satisfies
\(B(z_{0})=1+O(z_{0})\) for small \(z_{0}\). Here and below, \(O(.)\)
denotes the standard big-O notation.  Moreover, this function is
finite for \(z_{0}\in [0,\ln(1/(1-\omega)))\) and diverges as
\(z_{0}\) approaches \(\ln(1/(1-\omega))\). On the right-hand side of
Eq.~\eqref{eq:lem_lbnds}, the point of divergence corresponds to the
value of \(\theta_{e}\) that approaches its upper bound, for which
\(e^{\beta\theta_{e}}=1/(1-\omega)\).  Thus \(\cL(\beta,\theta_{e})\)
is finite for all \(\theta_{e}\in [0, \ln(1/(1-\omega))/\beta)\).
 
\begin{proof}
  The lower bound in the lemma follows from the inequality in the
  first line of Eq.~\eqref{eq:obtbnd2}, which is applied to the
  definition of \(\cL(\beta,\theta_{e})\) in terms of \(\cO(\beta,t)\).

  To prove the upper bound, we first obtain an upper bound on the
  quantity \(-\ln(1-(1-\omega) e^{-\beta (x-\theta_{e})})\).  For this
  purpose, consider the function \(f(z)=-\ln(1-ae^{-(z-z_{0})})\) where
  \(0< a < 1\), \(z,z_{0}\geq 0\) and \(1-ae^{z_{0}} > 0\). Direct
  calculation shows that \(f(0)=-\ln(1-ae^{z_{0}})\), \(f(\infty)=0\),
  $f'(z) = -a e^{-(z-z_{0})}/(1-a e^{-(z-z_{0})})$, and
  $f''(z) = a e^{-(z-z_{0})}/(1-a e^{-(z-z_{0})})^{2}$. Therefore, the
  first derivative $f'(z)$ is negative and monotonically
  increasing, and the second derivative $f''(z)$ is positive and
  monotonically decreasing. Thus the maximum of \(f''(z)\) on the domain
  is \(f''(0)\). By performing a Taylor expansion with a second-order
  remainder and further upper-bounding the remainder, we get
  \begin{align} f(z) &\leq f(z_{0}) + f'(z_{0})(z-z_{0}) + \max_{u\geq
                       0} f''(u) (z-z_{0})^{2}/2 \nonumber\\ &= -\ln(1-a) -
                                                               \frac{a}{1-a}(z-z_{0}) + \frac{a
                                                               e^{z_{0}}}{2(1-ae^{z_{0}})^{2}}(z-z_{0})^{2}.
                                                               \label{eq:fx_expansion0}
  \end{align} For the quantity of interest, \(a=(1-\omega)\),
  \(z_{0}=\beta\theta_{e}\), and \(z=\beta x\), where the constraints on \(a\) and \(z_{0}\) 
  are satisfied by the definition of \(\theta_{e}\).  Substituting these
  values, we obtain the upper bound
  \begin{align} -\ln(1-(1-\omega) e^{-\beta (x-\theta_{e})}) &\leq
                                                               -\ln(\omega) - \frac{\beta(1-\omega)}{\omega}(x-\theta_{e})
                                                               +\frac{\beta^{2}(1-\omega)}{2\omega^{2}}B(\beta\theta_{e})(x-\theta_{e})^{2}.
                                                               \label{eq:fx_expansion}
  \end{align} 
  Substituting this bound into the expression for
  \(\cL(\beta,\theta_{e})\) and then evaluating the expectation
  yields the upper bound in the lemma. 
\end{proof}

By substituting the upper bound from Lem.~\ref{lem:lbetatheta_bounds} into the expression
for \(\Exp_{\nu}(\Delta)\) in Eq.~\eqref{eq:expdelta_lbetatheta}, we obtain the following 
upper bound \(\bar\Delta\) on \(\Exp_{\nu}(\Delta)\):
\begin{align}
  \bar\Delta
  &= n\beta\frac{(1-\omega)}{2\omega}
    B(\beta\theta_{e})(\sigma^{2}+(\theta-\theta_{e})^{2})
    -\frac{\ln(\epsilon)}{\beta}.
    \label{eq:bardelta}
\end{align}
Rather than directly minimizing \(\Exp_{\nu}(\Delta)\) over \(\beta\)
and \(\theta_{e}\), one can instead minimize the upper bound
\(\bar\Delta\) for constructing estimation factors.  While this
approach may yield suboptimal estimation factors and, consequently,
suboptimal confidence bounds \(S_{\lb}\), it has the advantage that
only the mean and variance, rather than a full characterization, of
the target distribution of \(X\) are required. The analysis below
shows that the expectation of the so obtained \(S_{\lb}\) with respect
to the true distribution \(\nu\) approaches the best possible as the
number of trials \(n\) increases.

Anticipating that well-performing confidence bounds require 
\(\beta\) to be of order \(1/\sqrt{n}\), we neglect terms of order \(\beta^{2}\)
when minimizing \(\bar\Delta\).
Since \(B(\beta\theta_{e})=1+O(\beta)\) for small \(\beta\), we choose \(\theta_{e}\) and 
\(\beta\) to minimize
\begin{align}
  \bar\Delta_{1}
  &=  n\beta \frac{(1-\omega)}{2\omega} (\sigma^{2}+(\theta-\theta_{e})^{2})
    -\frac{\ln(\epsilon)}{\beta}.
\end{align}
Regardless of the choice of \(\beta\), this requires setting \(\theta_{e}=\theta\).
To choose \(\beta\), we use the fact that the minimum over \(x\) of \(a x+b/x\) for
positive \(a,b,x\) is achieved at \(x=\sqrt{b/a}\). Thus, we set \(\beta=\beta_1\),
where 
\begin{align} \label{eq:sub_opt_beta}
  \beta_1 &= \sqrt{\frac{2\omega \ln(1/\epsilon)}{\sigma^{2}{n(1-\omega)}}}.
\end{align}
The choices \(\theta_{e}=\theta\) and \(\beta=\beta_1\) are valid 
as long as \(t=e^{\beta \theta_e}\leq 1/(1-\omega)\). This condition is equivalent to  
requiring that $n\geq n_{\thresh}$, where $n_{\thresh} =2\omega \theta^2 \ln(1/\eps)/\qty((1-\omega)\sigma^2 \ln^2(1/(1-\omega)))$. 
Therefore, for large $n$, the above choices for \(\theta_{e}\) and \(\beta\) yield the bound
\begin{align}
  \Exp_{\nu}(\Delta)
  &\leq
    \sigma\sqrt{\frac{2n(1-\omega)\ln(1/\epsilon)}{\omega}}
    (1+O(1/\sqrt{n})).
    \label{eq:expgap}
\end{align}
This bound is proportional to \(\sigma\sqrt{n(1-\omega)/\omega}\),
whose square is \(n(1-\omega)\sigma^{2}/\omega\). The following
discussion suggests that this bound is asymptotically optimal for
constructing lower confidence bounds. Let
\(Q_{n\omega}=\sum_{i=1}^{n \omega} X_i\) and
\(Q_{n(1-\omega)}=\sum_{i=1}^{n (1-\omega)} X_i\) represent the sums
of \(n\omega\) and \(n(1-\omega)\) i.i.d. samples of \(X\),
respectively, where, for simplicity, we assume that $n \omega$ and
$n (1 - \omega)$ are integers.  In this context,
\(\qty((1-\omega)/\omega) Q_{n\omega}\) serves as an unbiased 
in-expectation estimate of \(Q_{n(1-\omega)}\) based on \(n \omega \)
i.i.d. samples of \(X\) accessible through spot-checking with
probability \(\omega\).  The quantity \(n(1-\omega)\sigma^{2}/\omega\)
can be interpreted as the variance of the difference between the sum
of interest, \(Q_{n(1-\omega)}\), and its unbiased estimate,
\(\qty((1-\omega)/\omega) Q_{n\omega}\).  To explain this
interpretation, observe that because the variance of \(X\) according
to the target distribution is \(\sigma^2\), we have that the variance
of \(Q_{n(1-\omega)}\) is \(n(1-\omega)\sigma^{2}\), while the
variance of \(\qty((1-\omega)/\omega) Q_{n\omega}\) is
\(n(1-\omega)^{2}\sigma^{2}/\omega\).  Therefore the sum of these
variances is \(n(1-\omega)\sigma^{2}/\omega\).  The factor
\(\sqrt{\ln(1/\epsilon)}\) in the bound of Eq.~\eqref{eq:expgap} can
be interpreted as being related to a Gaussian tail-probability bound.
Based on these observations, we do not expect that there exists a
statistical protocol for constructing lower confidence bounds that has
a significantly smaller expected gap between \(S_{U}\) and \(S_{\lb}\)
than what can be achieved using estimation factors.

Ideally, the reference distribution \(\nu\) should be chosen such
that its mean \(\theta\) and variance \(\sigma^{2}\) closely match
or are identical to those of the true distribution governing a
trial. However, in practice, the true distribution is often unknown
and may even drift over time. To address this challenge, we can
estimate the mean \(\theta\) and variance \(\sigma^{2}\) through
calibration or theoretical modeling conducted before an experiment.
Let the resulting estimates be \(\theta_{e}\) for the
mean \(\theta\) and \(\sigma_{e}^{2}\) for the variance
\(\sigma^{2}\). In view of the bounds obtained, a strategy is to set
\(\beta=\beta_{e}=\sqrt{{2\omega\ln(1/\epsilon)}/{\qty(\sigma_{e}^{2}{n(1-\omega)})}}\)
and \(t=t_{e}=e^{\beta_e \theta_{e}}\) in the expression for the
estimation factor \(T\) in Eq.~\eqref{eq:t_maximal}.  The resulting
expression \(T_{e}\) is a valid estimation factor provided that
\(t_{e} \leq 1/(1-\omega)\), which holds as along as
$n\geq \tilde{n}_{\thresh}$, where
$\tilde{n}_{\thresh} =2\omega\theta_{e}^2
\ln(1/\eps)/\qty((1-\omega)\sigma_{e}^2 \ln^2(1/(1-\omega)))$.  We
expect this estimation factor to perform well when the estimates
\(\theta_{e}\) and \(\sigma_{e}^{2}\) are close to \(\theta\) and
\(\sigma^{2}\). Nevertheless, regardless of accuracy, we can still
bound its performance according to the following proposition.

\begin{proposition}\label{prop:bndwithest}
  Let \(T_{e}\) be the estimation factor expressed in terms of
  \(\theta_{e}\) and \(\sigma_{e}^{2}\) as described in the previous
  paragraph, where we assume that \(n\geq \tilde{n}_{\thresh}\) to
  ensure the condition \(t_{e} \leq 1/(1-\omega)\) required by the
  construction.  Then, we have the bound
  \begin{align}\label{eq:deltae_bnd}
    \Exp_{\nu}(\Delta)
    &\leq
      \sigma\sqrt{\frac{n(1-\omega)\ln(1/\epsilon)}{2\omega}}
      \qty(
      B\qty(\frac{\theta_{e}}{\sigma_{e}}\sqrt{\frac{2\omega\ln(1/\epsilon)}{n(1-\omega)}})\qty(\frac{\sigma}{\sigma_{e}} +
      \frac{(\theta-\theta_{e})^{2}}{\sigma\sigma_{e}}) +
      \frac{\sigma_{e}}{\sigma}
      ).
  \end{align}
\end{proposition}

\begin{proof}
  It suffices to substitute \(\beta_{e}\) for \(\beta\) in Eq.~\eqref{eq:bardelta} and
  rewrite the expression for the upper bound \(\bar \Delta\) on \(\Exp_{\nu}(\Delta)\).
\end{proof}

In the case of large \(n\), \(B(\beta_{e}\theta_{e})\) approaches \(1\),
so that the bound in Prop.~\ref{prop:bndwithest} is minimized when
\(\theta_{e}=\theta\) and \(\sigma_{e}^{2}=\sigma^{2}\).
Moreover, when \(\theta_{e}\) deviates from \(\theta\) 
or \(\sigma_{e}^{2}\) deviates from \(\sigma^{2}\), the upper bound 
on the right-hand size of Eq.~\eqref{eq:deltae_bnd} exceeds
its minimum value only by additive terms that scale quadratically with
the relative deviations \(|\sigma_{e} -\sigma|/\sigma\)
and \(|\theta-\theta_{e}|/\sigma\).  In this sense, the upper bound on 
\(\Exp_{\nu}(\Delta)\) given in Eq.~\eqref{eq:deltae_bnd} is quite insensitive 
to discrepancies between the estimated and true values of the mean and variance
of \(X\).
For good performance, it therefore suffices to determine 
\(\sigma^{2}\) with small relative error and \(\theta\) with error small compared to \(\sigma\). 
A moderate number of i.i.d. calibration samples of \(X\)  is sufficient to achieve this goal.
In the next section, we formalize this claim.

\section{Requirements for estimating reference mean and variance of \texorpdfstring{\(X\)}{X} from calibration trials}
\label{sect:calibratevar}

To compute estimation factors by means of the strategy of
Sect.~\ref{sect:analytical_opt}, it is necessary to estimate the mean
and variance of the reference distribution \(\nu\).  This can be done
with calibration trials prior to running the spot-checking experiment.
The estimated mean and variance are then RVs depending on these
calibration trials. When their statistics are taken into account, the
expected gap \(\Exp(\Delta)\) with \(\Delta=S_{U}-S_{\lb}\) increases
compared to the best possible in Eq.~\eqref{eq:deltae_bnd}.  In this
section, we obtain bounds on \(\Exp(\Delta)\) that take into account
the statistics of the calibration trials and show that relatively few
calibration trials are required for good performance.  For this
purpose, in addition to having finite mean \(\theta\) and finite
variance \(\sigma^{2}\), we assume that the reference distribution's
third and fourth central moments, denoted by \(m_{3}\) and \(m_{4}\),
are finite.  Note that the analysis below explicitly permits the case \(\sigma^{2}=0\).

Our calculations in this section can be used to estimate the number of
calibration trials needed to ensure good confidence bounds. 
Although these calculations are based on using estimates of the mean 
and variance to optimize estimation factors, we expect that when using 
direct numerical optimization according to Sect.~\ref{sect:numerical_opt}, 
similar numbers of calibration trials are sufficient.

Suppose that \(\theta_{e}\) and \(\sigma_{e}^{2}\) are estimates
obtained from calibration trials.  Then, these quantities are themselves
RVs with distributions that depend on both the calibration
trials and the method used to obtain the estimates.  Our goal is 
to establish an upper bound on the expectation of the bound in Eq.~\eqref{eq:deltae_bnd},
where the expectation is taken with respect to the distributions of the random 
estimates \(\theta_{e}\) and \(\sigma_{e}^{2}\). For this purpose, we assume 
that the calibration trials consist of i.i.d. samples of \(X\) drawn from the 
reference distribution \(\nu\). 

To simplify the analysis, we obtain \(\theta_{e}\) as the sample mean
 \(\theta_{e}=\sum_{i=1}^{n_{a}}X_{i}/n_{a}\) of \(n_{a}\) i.i.d. samples 
 of \(X\), and  we obtain \(\sigma_{e}^{2}\) from an additional set of \(2n_{v}\) 
 i.i.d. samples of \(X\) according to
\(\sigma_{e}^{2}= \sum_{i=1}^{n_{v}}(X_{2i}-X_{2i-1})^{2}/(2n_{v})\).
Thus, the total number of calibration trials used is \(n_{a}+ 2n_{v}\).
In practice, we would use the usual sample-efficient estimates of the
mean and variance from a single set of samples of \(X\). The less
efficient method used here ensures independence between the estimates
\(\theta_{e}\) and \(\sigma_{e}^{2}\) and serves as a proof of principle.
The bounds obtained in our analysis are not tight and can readily be
improved. We aim to obtain interpretable bounds that can be
used in practice to choose the number of calibration trials as well
as other parameters used for the spot-checking experiment, 
such as \(n\) and \(\omega\).

In view of Eq.~\eqref{eq:deltae_bnd}, two regularizations are required
to ensure that the bound is finite with probability \(1\). The bound
in Eq.~\eqref{eq:deltae_bnd} diverges to infinity at
\(\sigma^{2}_{e}=0\).  The first regularization is required to avoid
this divergence. We replace \(\sigma^{2}_{e}\) with
\(\sigma_{\tilde e}^{2}=\sigma_{e}^{2}+r_{e}^{2}\), where
\(r_{e}^{2}\) is a constant that is chosen independent of the
calibration trials to ensure good performance.  Based on the
expression for \(\beta_{e}\) given above Prop.~\ref{prop:bndwithest},
we correspondingly define
\(\beta_{\tilde e} = \sqrt{2\omega\ln(1/\epsilon)/(\sigma_{\tilde
    e}^{2}{n(1-\omega)})}\). For the remainder of this section,
estimation factors and all related parameters are expressed in terms
of \(\beta_{\tilde e}\) rather than \(\beta_{e}\).

Let \(z_{\infty}\) be the positive point at which the function
\(B(z_0)\), as defined in Lem.~\ref{lem:lbetatheta_bounds}, diverges
to infinity.  Accordingly, we have \(z_{\infty}=\ln(1/(1-\omega))\).
The second regularization is required to ensure that \(B(z_0)\) is
bounded from above and not much larger than \(1\).  In view of the
monotonicity established below, this requires \(z_0\) to be restricted
by an upper bound that is strictly less than \(z_{\infty}\). This
regularization not only avoids divergence but also simplifies bounding
the expectation of terms involving \(B(z_0)\).  Moreover, since
\(z_0< z_{\infty}\) implies \(t=e^{z_0}<1/(1-\omega)\), the positivity
of the resulting estimation factors is also ensured.  The function
\(B(z_0)\) is monotonically increasing and convex on
\([0,z_{\infty})\).  To check this monotonicity and convexity,
consider the function \(f(y)=y/(1-y)^{2}\). Then
\((1-\omega)B(z_0)/\omega^{2}=f(y)\) with \(y=(1-\omega)e^{z_0}\).  We
have \(f'(y)=(1+y)/(1-y)^{3}>0\) and \(f'(y)\) is monotonically
increasing on the range \(0\leq y < 1\).  We also have that
  \(\frac{d}{dz_{0}}y\) is positive and monotonically increasing. 
Thus, \(f(y)\) is
monotonically increasing and convex, and  in 
consideration of the chain rule, the same holds for the function
\(B(z_0)\).  It follows that for any \(z_{u}\in [0,z_{\infty})\) and
\(z_0 \in [0, z_{u}]\),
\begin{align}\label{eq:B_fun_ub}
B(z_0)& \leq 1+\qty(B(z_{u})-B(0))(z_0-0)/(z_{u}-0) \nonumber\\
             & =1+\qty(B(z_{u})-1)z_0/z_{u}.
\end{align}
Based on this bound, we regularize and remove the terms involving
\(B(z_0)\) by ensuring that \(z_{0}\leq z_{u}\), where \(z_{u}\)
is chosen to satisfy \(B(z_{u})\leq 2\).  In the relevant
  expressions, we have \(z_{0}=\beta_{\tilde e}\theta_{e}\), so to
  ensure \(z_{0}\leq z_{u}\), we set
  \(\theta_{u}=z_{u}/\beta_{\tilde e}\) and replace \(\theta_{e}\) by
  \(\theta_{\tilde e} \coloneqq  \min(\theta_{u},\theta_{e})\). One such
choice of \(z_u\) and the corresponding \(\theta_{u}\) can be obtained
as follows: Write \(z_{u}=\ln((1-v \omega)/(1-\omega))\), where \(v\)
is constrained to be in the interval \((0,1]\) to ensure that
\(z_{u}\in [0, z_{\infty})\). Then
\begin{align}
  B(z_{u}) &=
         \frac{\omega^{2}(1-v\omega)}{(1-\omega)v^{2}\omega^{2}}
         \nonumber\\
       & =\frac{1}{v^{2}}\times \frac{1-v\omega}{1-\omega}.
\end{align}
If we further express \(v\) as \(v=(1-\omega)v_{0}+\omega\) and constrain \(v_{0}\) 
to be in the interval \( (0,1)\), \(B(z_{u})\) can be bounded as follows:
\begin{align}
  B(z_{u}) &=
         \frac{1}{((1-\omega)v_{0}+\omega)^{2}} \times \frac{1-\omega^{2}- (1-\omega)\omega v_{0}}{1-\omega}
         \nonumber\\
       &=
         \frac{1+\omega(1-v_{0})}{(v_{0}+\omega(1-v_{0}))^{2}}
         \nonumber\\
       &=\frac{1}{v_{0}+\omega(1-v_{0})}
         +\frac{1-v_{0}}{(v_{0}+\omega(1-v_{0}))^{2}}
         \nonumber\\
           &\leq
             \frac{1}{v_{0}} + \frac{1-v_{0}}{v_{0}^{2}}
             \nonumber\\
       &= \frac{1}{v_{0}^{2}}.
\end{align}
We can thus set \(v_{0}=1/\sqrt{2}\) in the chain of expressions for
\(z_{u}\) to ensure that \(B(z_{u})\leq 2.\)  In view of the monotonicity 
of the function \(B(z_0)\), it  suffices for \(z_{u}\) to satisfy
\begin{align}
  z_{u} &\leq \ln((1-((1-\omega)/\sqrt{2}+\omega)\omega)/(1-\omega))\nonumber\\
        &= \ln(1+(1-1/\sqrt{2})\omega).
     \label{eq:barnsbnd0}
\end{align}
By the concavity of \(\ln(1+x)\) in \(x\), we have
\(\ln(1+ax)\geq x \ln(1+a)\) for \(a\geq 0\) and \(0\leq x\leq
1\). Applying this with \(a=(1-1/\sqrt{2})\) and \(x=\omega\), we
obtain the inequality
\(\ln(1+(1-1/\sqrt{2})\omega)\geq \omega \ln(2-1/\sqrt{2}) \geq
\omega/4\), where we used \(\ln(2-1/\sqrt{2})\approx 0.2569\geq
1/4\). For simplicity, we therefore use the upper bound \(\omega/4\)
for \(z_{u}\). Since \(z_u=\beta_{\tilde e} \theta_{u}\) and
\(\beta_{\tilde e}\leq
\beta_{r}=\sqrt{2\omega\ln(1/\epsilon)/(r_{e}^{2}{n(1-\omega)})}\), it
suffices to choose
\begin{align}
  \theta_{u}&=\omega/(4\beta_{r})\nonumber\\
            &=\frac{1}{4}\sqrt{\frac{{r_{e}^{2}n\omega(1-\omega)}}{2{\ln(1/\epsilon)}}}
              \label{eq:barnsbnd}
\end{align}
to guarantee that \(B(\beta_{\tilde e}\theta_{u})\leq 2\). The bound
\(\theta_{u}\) grows with \(n\), so for large enough \(n\),
\(\theta_{u}\) is unlikely to be exceeded by either \(\theta_{e}\) or
\(\theta\).  In the next section, we analyze the situation where
  the expected number of spot-checked trials \(\bar n_{s}=n\omega\) is
  held constant as \(n\) increases. In this case, \(\theta_{u}\) is
asymptotically proportional to \(\sqrt{\bar n_{s}}\) as
\(n \to \infty\).  For the bounds below, we require
\(\theta_{u}\geq\theta\), which is ensured if \(n \omega\) is large
enough.

To continue, we rewrite and relax the bound in Eq.~\eqref{eq:deltae_bnd}.
Since \(\Exp_{\nu}(\Delta)\) is proportional to \(\sqrt{n/\omega}\) for
large \(n\) and small \(\omega\), we consider the normalized quantity
\(\tilde\Delta=\sqrt{\omega/n}\,\Exp_{\nu}(\Delta)\).
With the above regularizations, we replace \(\theta_{e}\) with
\(\theta_{\tilde e}=\min(\theta_{u},\theta_{e})\), \(\sigma_{e}\) with 
\(\sigma_{\tilde e}\), and \(\beta_{e}\) with  
\(\beta_{\tilde e} =\sqrt{2\omega\ln(1/\epsilon)/(\sigma_{\tilde e}^{2}{n(1-\omega)})}\). Then, from Eq.~\eqref{eq:deltae_bnd} we have
\begin{align}
  \tilde\Delta & \leq \sqrt{(1-\omega)\ln(1/\epsilon)/2}
                 \qty(B(\beta_{\tilde e} \theta_{\tilde e})\qty(\frac{\sigma^{2}}{\sigma_{\tilde e}} + \frac{(\theta-\theta_{\tilde e})^{2}}{\sigma_{\tilde e}}) + \sigma_{\tilde e})  \nonumber \\    
                & \leq    \sqrt{(1-\omega)\ln(1/\epsilon)/2}
                 \qty(\qty(1+ \frac{B(\beta_{r} \theta_{u})-1}{\beta_{r} \theta_{u}}\beta_{\tilde e} \theta_{\tilde e})\qty(\frac{\sigma^{2}}{\sigma_{\tilde e}} + \frac{(\theta-\theta_{\tilde e})^{2}}{\sigma_{\tilde e}}) + \sigma_{\tilde e})  
                \nonumber \\    
                 & \leq \sqrt{(1-\omega)\ln(1/\epsilon)/2}
                 \qty(\qty(1+4\sqrt{\frac{2\ln(1/\epsilon)}{{n\omega (1-\omega)}}}\,
                 \frac{\theta_{\tilde e}}{\sigma_{\tilde e}})\qty(\frac{\sigma^{2}}{\sigma_{\tilde e}} + \frac{(\theta-\theta_{\tilde e})^{2}}{\sigma_{\tilde e}}) + \sigma_{\tilde e}).
                \label{eq:bndontildedelta}
\end{align}
To obtain the second line, we applied the inequality in
Eq.~\eqref{eq:B_fun_ub} with \(z_0=\beta_{\tilde e}\theta_{\tilde e}\)
and \(z_{u}=\beta_{r}\theta_{u}\) and used the constraint
\( \beta_{\tilde e} \theta_{\tilde e} \leq \beta_{r} \theta_{u} =
\omega/4\), and to obtain the last line, we used
\(B(\beta_{r}\theta_{u})\leq 2\).  Accounting for statistical
fluctuations in calibration trials and further manipulating the last
line of Eq.~\eqref{eq:bndontildedelta}, we obtain that the expected
gap \(\Exp(\tilde\Delta)\) is bounded as
\begin{align}
  \Exp(\tilde\Delta)     
                 & \leq \sqrt{(1-\omega)\ln(1/\epsilon)/2}
                 \qty(\Exp \qty(\frac{\sigma^{2}}{\sigma_{\tilde e}} + \frac{(\theta-\theta_{\tilde e})^{2}}{\sigma_{\tilde e}})\,                 
                 +4\sqrt{\frac{2\ln(1/\epsilon)}{{n\omega (1-\omega)}}} \Exp \qty(\frac{\theta_{\tilde e} \sigma^{2}}{\sigma^2_{\tilde e}} + \frac{\theta_{\tilde e}(\theta-\theta_{\tilde e})^{2}}{\sigma^2_{\tilde e}})\,
                 +\Exp(\sigma_{\tilde e})).                           
                  \label{eq:exp_bndontildedelta}
\end{align}
The random quantities in this expression are \(\theta_{\tilde e}\) and \(\sigma_{\tilde e}\), 
which are independent.  To proceed we need estimates of the following expectations:
\begin{align}
  &\Exp(\theta_{\tilde e}), \Exp((\theta-\theta_{\tilde e})^{2}),  \Exp(\theta_{\tilde e}(\theta-\theta_{\tilde e})^{2}),\\
  &\Exp(\sigma_{\tilde e}), \Exp(1/\sigma_{\tilde e}), \Exp(1/\sigma_{\tilde e}^{2}).
\end{align}
For \(\theta_{u}\geq\theta\), the first three expectations are each
bounded from above by the corresponding expectations for \(\theta_{e}\).
To see this, consider the function \(F(\theta_{e})=G(\theta_{e})-G(\theta_{\tilde e})\), 
where  \(G(x)\) is chosen as \(G(x)=x\), \(G(x)=(\theta-x)^{2}\), 
or \(G(x)=x(\theta-x)^{2}\).  Then, in the case where \(\theta_{e}\leq\theta_{u}\),  
\( \theta_{\tilde e} =\theta_{e} \) and so \(F(\theta_{e})=0\). Conversely, if 
\(\theta_{e}>\theta_{u}\),  then \( \theta_{\tilde e} =\theta_{u} \) and 
\(F(\theta_{e}) = G(\theta_{e})-G(\theta_{u})\).  Since \(G(x)\) is
monotonically increasing for \(x\geq \theta\), \(F(\theta_{e})\) is positive
when \(\theta_{e}>\theta_{u}\) and \(\theta_{u}\geq \theta\).

For the rest of this section, we assume that
\(\theta_{u}\geq\theta\). As discussed above, this assumption is
satisfied for sufficiently large values of \(n \omega \), and can be
verified beforehand by means of Eq.~\eqref{eq:barnsbnd} if an
upper bound on \(\theta\) is known.  Since the bound on
\(\Exp(\tilde\Delta)\) in Eq.~\eqref{eq:exp_bndontildedelta} is
monotonically increasing in the expectations involving
\(\theta_{\tilde e}\), and these expectations are bounded from above
by the corresponding expectations involving \(\theta_{e}\), we replace
the former with the latter. These terms, namely
\(\Exp(\theta_{e}), \Exp((\theta-\theta_{e})^{2})\), and \(\Exp(\theta_{
  e}(\theta-\theta_{ e})^{2})\), can be expressed or bounded as 
\begin{align}
  \Exp(\theta_{e}) = \theta,\\
  \Exp((\theta-\theta_{e})^{2}) &= \sigma^{2}/n_{a},\\
  \Exp(\theta_{e}(\theta-\theta_{e})^{2})
                                &=
                                  \Exp(\theta(\theta-\theta_{e})^{2})-\Exp((\theta-\theta_{e})^{3})
                                  \nonumber\\
                                &=
                                  \theta\sigma^{2}/n_{a}-m_{3}/n_{a}^{2}\leq
                                  \theta\sigma^{2}/n_{a}+|m_{3}|/n_{a}^{2}.
\end{align}
These identities are obtained by expanding \(\theta-\theta_{e}=\frac{1}{n_{a}}\sum_{i=1}^{n_{a}}(\theta-X_{i})\) 
and applying the assumption of i.i.d. calibration trials. In more detail, we have
\begin{align}
  \Exp((\theta-\theta_{e})^{2})
  &=\Exp\qty(\qty(\frac{1}{n_{a}}\sum_{i=1}^{n_{a}}(\theta-X_{i}))^{2})
    \nonumber\\
  &=
    \Exp\qty(\frac{1}{n_{a}^{2}}\sum_{i=1}^{n_{a}}(\theta-X_{i})^{2})
    \nonumber\\
  &= \sigma^{2}/n_{a},
\end{align}
and for the third identity
\begin{align}
  \Exp((\theta-\theta_{e})^{3})
  &=\Exp\qty(\qty(\frac{1}{n_{a}}\sum_{i=1}^{n_{a}}(\theta-X_{i}))^{3})
    \nonumber\\
  &=
    \Exp\qty(\frac{1}{n_{a}^{3}} \sum_{i=1}^{n_{a}}(\theta-X_{i})^{3})
    \nonumber\\
  &=m_{3}/n_{a}^{2}.
\end{align}

For the expectations involving \(\sigma_{\tilde e}\), we use the bound
\(\sigma_{\tilde e}^{2}\geq r_{e}^2\). 
To handle reciprocal powers of \(\sigma_{\tilde e}\),  we
upper-bound \(1/(1+\delta)\) for \(\delta \geq l-1\), where \(l\in(0,1)\).
 We claim that 
 \begin{align} \label{eq:quadracticbound}
  1/(1+\delta) \leq
  f(\delta) \coloneqq 1-\delta+\delta^{2}/l.
\end{align}
To verify this inequality, observe that \(f(\delta)-1/(1+\delta) = -\delta^{2}/(1+\delta)+\delta^{2}/l\),
which is non-negative whenever \(l\leq (1+\delta)\), that is, for
the values of \(\delta\) being considered. 
We also need the variance of the estimate \(\sigma_{e}^{2}\). For two
i.i.d. instances \(X_{1}\) and \(X_{2}\) of \(X\), we have
\begin{align} \label{eq:m2}
  \Exp\qty((X_{1}-X_{2})^{2})&= \Exp\qty(\qty((X_{1}-\theta)-(X_{2}-\theta))^{2})\nonumber\\
                         & = 2 \sigma^{2},
\end{align}\begin{align} \label{eq:m4}
  \Exp\qty((X_{1}-X_{2})^{4})&= \Exp\qty(\qty((X_{1}-\theta)- (X_{2}-\theta))^{4})\nonumber\\
                         & = 2 m_{4} + 6 \sigma^{4}.
\end{align}
The variance of \(\sigma_{e}^{2}\) is given by
\begin{align}
\Var(\sigma_{e}^{2})&=\Exp(\sigma_{e}^{4})-\qty(\Exp(\sigma_{e}^{2}))^2 \nonumber \\
                                &=\Exp\qty(\qty(\sum_{i=1}^{n_{v}}\frac{(X_{2i}-X_{2i-1})^{2}}{2n_{v}})^2)-\qty(\Exp\qty(\sum_{i=1}^{n_{v}}\frac{(X_{2i}-X_{2i-1})^{2}}{2n_{v}}))^2
\end{align}
Under the assumption of i.i.d. calibration trials and using Eqs.~\eqref{eq:m2} and \eqref{eq:m4}, 
the first term \( \Exp(\sigma_{e}^{4}) \) is evaluated to be \( (2 m_{4} + 6 \sigma^{4} + 4 (n_{v}-1) \sigma^4)/(4 n_{v}) \),
while the second term is \( \qty(\Exp(\sigma_{e}^{2}))^2 = \sigma^4\). Therefore,
the variance of \(\sigma_{e}^{2}\) is \((m_{4}+\sigma^{4})/(2n_{v})\).

We can now estimate the expectations involving \(\sigma_{\tilde e}\).
Since the square-root function \(f(x)=\sqrt{x}\) is concave, Jensen’s inequality gives 
\(\Exp( \sqrt{\sigma_{e}^{2}+r_{e}^2} ) \leq \sqrt{\Exp(\sigma_{e}^{2})+r_{e}^2} =\sqrt{\sigma^{2}+r_{e}^2} \)
and \(\Exp\qty(\sqrt{1/(\sigma_{e}^{2}+r_{e}^2)}) \leq \sqrt{\Exp\qty(1/(\sigma_{e}^{2}+r_{e}^{2}))}\). As
 \((\sigma_{e}^{2}-\sigma^{2})/(\sigma^{2}+r_{e}^2) \geq r_{e}^{2}/(\sigma^{2}+r_{e}^{2})-1\), 
 we can apply Eq.~\eqref{eq:quadracticbound}  with \( \delta=(\sigma_{e}^{2}-\sigma^{2})/(\sigma^{2}+r_{e}^2)\)
 and \(l= r_{e}^{2}/(\sigma^{2}+r_{e}^{2})\) to obtain the bound 
 \begin{align}
 \frac{1}{1+(\sigma_{ e}^{2}-\sigma^{2})/(\sigma^{2}+r_{e}^2)} & \leq  1
                            -\frac{\sigma_{e}^{2}-\sigma^{2}}{\sigma^{2}+r_{e}^{2}}+
                            \frac{
                            (\sigma_{ e}^{2}-\sigma^{2})^{2}
                            /(\sigma^{2}+r_{e}^2)^{2}
                            }{
                            r_{e}^{2}/(\sigma^{2}+r_{e}^{2})
                            }.
\end{align}
By combining the above inequalities we obtain
\begin{align}
  \Exp(\sigma_{\tilde e}) & = \Exp( \sqrt{\sigma_{e}^{2}+r_{e}^2} )  \nonumber \\
                                       &\leq \sqrt{\sigma^{2}+r_{e}^2} ,\\
  \Exp\qty(\frac{1}{\sigma_{\tilde e}^{2}})
                          &= \Exp\qty(\frac{1}{\sigma_{e}^{2}+r_{e}^2}) \nonumber \\
                          &=\frac{1}{\sigma^{2}+r_{e}^2}
                            \Exp\qty(
                            \frac{1}{1+
                            (\sigma_{ e}^{2}-\sigma^{2})/(\sigma^{2}+r_{e}^2)}
                            )\nonumber\\
                          &\leq\frac{1}{\sigma^{2}+r_{e}^2}
                            \Exp\qty(1
                            -\frac{\sigma_{e}^{2}-\sigma^{2}}{\sigma^{2}+r_{e}^{2}}+
                            \frac{
                            (\sigma_{ e}^{2}-\sigma^{2})^{2}
                            /(\sigma^{2}+r_{e}^2)^{2}
                            }{
                            r_{e}^{2}/(\sigma^{2}+r_{e}^{2})
                            }
                            )\nonumber\\
                          &=\frac{1}{\sigma^{2}+r_{e}^2}
                            \qty(
                            1
                                                       + \frac{m_{4}+\sigma^{4}
                            }{
                            2 n_{v}r_{e}^{2}(\sigma^{2}+r_{e}^{2})}
                            ), 
                                                      \nonumber \\
  \Exp\qty(\frac{1}{\sigma_{\tilde e}})
                          &=
                            \Exp\qty(\sqrt{\frac{1}{\sigma_{e}^{2}+r_{e}^2}})
                            \nonumber\\
                          &\leq \sqrt{\Exp\qty(\frac{1}{\sigma_{e}^{2}+r_{e}^{2}})}
                            \nonumber\\
                          &\leq \frac{1}{\sqrt{\sigma^{2}+r_{e}^{2}}}
                            \sqrt{1
                            +  \frac{m_{4}+\sigma^{4}
                            }{
                            2 n_{v}r_{e}^{2}(\sigma^{2}+r_{e}^{2})} 
                            } \nonumber\\
                          & \leq
                            \frac{1}{\sqrt{\sigma^{2}+r_{e}^{2}}}
                            \qty(
                            1+  \frac{m_{4}+\sigma^{4}
                            }{
                            4 n_{v}r_{e}^{2}(\sigma^{2}+r_{e}^{2})}                           
                            ).
\end{align}
Note that the inequality in the last line follows from the fact that
\(\sqrt{1+x}\leq 1+x/2\) for all \(x \geq -1\).  To simplify
expressions, we write
\begin{align}
  c(n_{v}) &= \frac{m_{4}+\sigma^{4}
                            }{
                            2 n_{v}r_{e}^{2}(\sigma^{2}+r_{e}^{2})},
\end{align}
which is of order \( 1/n_{v}\).

We can now bound each term involving an expectation in
the upper bound on \(\Exp(\tilde\Delta)\) of Eq.~\eqref{eq:exp_bndontildedelta}
as follows:
\begin{align}
  \Exp(\sigma_{\tilde e})
  &\leq \sqrt{\sigma^{2}+r_{e}^{2}},\label{eq:deltacal_sum1}\\
  \Exp\qty(
  \frac{1}{\sigma_{\tilde e}}(\sigma^{2}+(\theta-\theta_{e})^{2})
  )
  &=
    \Exp\qty(\sigma^{2}+(\theta-\theta_{e})^{2})\Exp\qty(\frac{1}{\sigma_{\tilde e}})
    \nonumber\\
  &\leq
    \frac{\sigma^{2}}{\sqrt{\sigma^{2}+r_{e}^{2}}}\qty(1+\frac{1}{n_{a}})
    \qty(1+
    c(n_{v})/2)\nonumber\\
  &=
    \sqrt{\sigma^{2}+r_{e}^{2}}\qty(1-\frac{r_{e}^{2}}{\sigma^{2}+r_{e}^{2}})
    \qty(1+\frac{1}{n_{a}})
    \qty(1+
    c(n_{v})/2)
    ,\label{eq:deltacal_sum2}\\
 \Exp\qty( \frac{1}{\sigma_{\tilde e}^{2}}(\theta_{e}\sigma^{2}+\theta_{e}(\theta-\theta_{e})^{2}))
  &=\Exp\qty(\theta_{e}\sigma^{2}+\theta_{e}(\theta-\theta_{e})^{2})
    \Exp\qty(\frac{1}{\sigma_{\tilde e}^{2}})\nonumber\\
  &= \qty(\Exp(\theta_{e}\sigma^{2})+\Exp(\theta_{e}(\theta-\theta_{e})^{2}))
    \Exp\qty(\frac{1}{\sigma_{\tilde e}^{2}})\nonumber\\
  &\leq
    \frac{1}{\sigma^{2}+r_{e}^{2}}
    \qty(\theta\sigma^{2}+\frac{1}{n_{a}}\qty(\theta\sigma^{2}+ \frac{|m_{3}|}{n_{a}}))
    \qty(1+ c(n_{v})).\label{eq:deltacal_sum3}
\end{align}
Putting everything together, we find that the leading-order
term \(L_{0}\) in the upper bound on \(\Exp(\tilde\Delta)\) in Eq.~\eqref{eq:exp_bndontildedelta}, with respect to powers of
\(1/n_{a}\), \(c(n_{v})\) and  \(1/\sqrt{n}\), is
\begin{align}
  L_{0}
  &= \sqrt{2(1-\omega)\ln(1/\epsilon)}\sqrt{\sigma^{2}+r_{e}^{2}}
    \qty(1-\frac{r_{e}^{2}}{2(\sigma^{2}+r_{e}^{2})}). \label{eq:deltacal_zeroth}
\end{align}
This leading-order term is minimized when \(r_{e}=0\). However, the
contributions to the upper bound from degree one or higher terms in
\(1/n_{a}\), \(c(n_{v})\) and \(1/\sqrt{n}\) may diverge as \(r_{e}\)
goes to zero; see the next paragraph for the explicit expressions for these terms. 
This issue can be mitigated by wisely choosing
  \(r_{e}\) to limit relative increase in the leading-order term
  \(L_{0}\) for the anticipated variances.  Usually,
it is sufficient to limit the relative increase to a moderately small
value such as \(1/64\). To achieve this, it suffices to set
\(r_{e}^{2}\) to around \(1/4\) of a lower bound on the variance, as shown
by the following rewriting of the expression for \(L_{0}\):
\begin{align}
  L_{0}
  &= \sqrt{2(1-\omega)\ln(1/\epsilon)}\sigma
    \sqrt{1+r_{e}^{2}/\sigma^{2}}\qty(1-\frac{r_{e}^{2}}{2(\sigma^{2}+r_{e}^{2})})
    \nonumber\\
  &\leq \sqrt{2(1-\omega)\ln(1/\epsilon)}\sigma
    \qty(1+\frac{r_{e}^{2}}{2\sigma^{2}})
    \qty(1-\frac{r_{e}^{2}}{2(\sigma^{2}+r_{e}^{2})})\nonumber\\
  &=
    \sqrt{2(1-\omega)\ln(1/\epsilon)}\sigma
    \qty(
    1+\frac{r_{e}^{4}}{4\sigma^{2}(\sigma^{2}+r_{e}^{2})}
    )\nonumber\\
  &=
    \sqrt{2(1-\omega)\ln(1/\epsilon)}\sigma
    \qty(
    1+\frac{r_{e}^{4}}{4\sigma^{4}(1+r_{e}^{2}/\sigma^{2})}
    ),
\end{align}
where the inequality in the second line uses \(\sqrt{1+x}\leq 1+x/2\) for all \(x \geq -1\).
The relative increase compared to \(r_{e}=0\) is bounded by
  \({r_{e}^{4}}/{4\sigma^{4}(1+r_{e}^{2}/\sigma^{2})}\), which is less
  than \(1/64\) for \(r_{e}^{2}\leq \sigma^{2}/4\). 

In one expression and including the terms involving powers of
  \(1/n_{a}\), \(c(n_{v})\) and \(1/\sqrt{n}\), the upper bound on
\(\Exp(\tilde\Delta)\) that has been obtained is
\begin{align}
  \Exp(\tilde\Delta) 
  &\leq
    \sqrt{2(1-\omega)\ln(1/\epsilon)}\sqrt{\sigma^{2}+r_{e}^{2}}
    \nonumber\\
  &\phantom{\leq\;}
    \times\Bigg(
    \qty(1-\frac{r_{e}^{2}}{2(\sigma^{2}+r_{e}^{2})})
    + \frac{1}{2}\qty(1-\frac{r_{e}^{2}}{\sigma^{2}+r_{e}^{2}})
    \qty(
    \frac{1}{n_{a}}
    + \frac{c(n_{v})}{2 n_{a}}
    + \frac{c(n_{v})}{2}
    )
    \nonumber\\
  &\phantom{\leq\;\times\;()}+
    2\sqrt{\frac{1}{\omega n}}\sqrt{\frac{2\ln(1/\epsilon)}{(1-\omega)}}
    \frac{1}{(\sigma^{2}+r_{e}^{2})^{3/2}}
    \qty(\theta\sigma^{2}+\frac{\theta\sigma^{2}}{n_{a}}+
     \frac{|m_{3}|}{n_{a}^2})
    \qty(1+ c(n_{v}))
    \Bigg).
    \label{eq:tildedelta_fullbound}
\end{align}
For spot-checking, a relevant limit is \(n \to \infty\) while keeping
\(n \omega = \bar n_{s}\) constant.
In this limit, \((1-\omega)\) approaches \(1\), and 
the bound becomes
\begin{align}
  \lim_{n\rightarrow\infty, n\omega=\bar n_{s}}\Exp(\tilde\Delta) 
  &\leq
    \sqrt{2\ln(1/\epsilon)}\sqrt{\sigma^{2}+r_{e}^{2}}
    \Bigg(
    \qty(1-\frac{r_{e}^{2}}{2(\sigma^{2}+r_{e}^{2})})
    + \frac{1}{2}\qty(1-\frac{r_{e}^{2}}{\sigma^{2}+r_{e}^{2}})
    \qty(
    \frac{1}{n_{a}}
    + \frac{c(n_{v})}{2 n_{a}}
    + \frac{c(n_{v})}{2}
    )
    \notag\\
  &\hphantom{\leq \sqrt{2\ln(1/\epsilon)}\sqrt{\sigma^{2}+r_{e}^{2}}
    \Bigg(}
     {}+ 2\sqrt{\frac{2\ln(1/\epsilon)}{\bar n_{s}}}\frac{1}{(\sigma^{2}+r_{e}^{2})^{3/2}}
    \qty(\theta\sigma^{2}+\frac{\theta\sigma^{2}}{n_{a}}+ \frac{|m_{3}|}{n_{a}^2})\qty(1+c(n_{v}))   
    \Bigg).
    \label{eq:tildedelta_asymbound} 
\end{align}
Moderately small values of \(n_{a}\) and \(n_{v}\) are sufficient to
ensure that the relative increase in the upper bound due to
statistical fluctuations in calibration is small.

\section{Constant expected number of spot-checked trials}
\label{sect:constant_omega_n}

At the end of the previous section, we obtained a bound on the
asymptotic expected gap between \(S_{U}\) and \(S_{\lb}\) in the limit
where \(n\to \infty\) while the expected number of spot-checked
trials, given by \(\bar n_{s}=n \omega\), remains constant. For this,
we took into consideration the statistics of the mean and variance
estimates from calibration.  In the first part of this section, we
evaluate the asymptotic behavior in the same limit, but for the case
where the mean and variance are known without requiring either
calibration or regularization.  As before, we consider an anticipated
distribution \(\nu\) governing each i.i.d.  spot-checking trial, under
which the variable \(X\) in each trial has a finite mean \(\theta\)
and finite variance \(\sigma^{2}\). For the analysis below, we assume
both \(\theta\) and \(\sigma^{2}\) are known, and we exclude the
trivial case \(\sigma^{2}=0\), which induces singular behavior in the
absence of regularization.  Consistent with the results of the
previous section, we find that \(\Exp_{\nu}(\Delta)/n\) is of order
\(1/\sqrt{\bar n_{s}}\), proportional to the standard deviation for
the estimate of the mean from \(\bar n_{s}\) spot-checked samples.  So
far we have considered only the expected gap between \(S_{U}\) and
\(S_{\lb}\).  Therefore, in the second part of this section, we
address the question of whether the actual confidence bound obtained
has high probability of being near the expected one
\(\Exp_{\nu}(S_{\lb})\) by estimating the variance of the gap between
\(S_{U}\) and \(S_{\lb}\). We find that this variance is
asymptotically smaller than the square of the expected gap by a factor
of \(2\ln(1/\epsilon)\).

From Prop.~\ref{prop:bndwithest} and by substituting \(\bar n_{s}/n\) for \(\omega\), 
we obtain the following bound under the anticipated distribution \(\nu\):
\begin{align}
  \frac{1}{n}\Exp_{\nu}(\Delta)
  &\leq
          \sigma\sqrt{\frac{(1-\bar n_{s}/n)\ln(1/\epsilon)}{2\bar n_{s}}}
      \qty(B(\beta\theta)+1),
    \label{eq:bndwithest_constant}
\end{align}
where \(\beta=\frac{1}{n\sigma}\sqrt{\frac{2\bar n_{s}\ln(1/\epsilon)}{(1-\bar n_{s}/n)}}\).
This bound is obtained using the estimation factor \(T_{e}\)
defined prior to Prop.~\ref{prop:bndwithest}, together with the 
estimates \(\theta_{e}=\theta\) and \(\sigma_{e}=\sigma\).  
To ensure the positivity of \(T_{e}\), we require that both \(n\) and
\(\bar n_{s}\) are large enough.  In particular, following the arguments 
around Eq.~\eqref{eq:B_fun_ub} of the previous section, it suffices to
ensure that \(B(\beta\theta)\leq 2\).  From those arguments, this is
satisfied if \(\beta\theta \leq \bar n_{s}/(4n)\). We can rewrite this inequality  
to obtain the condition 
\begin{align}
  \bar n_{s}
  &\geq 32 \qty(\frac{\theta}{\sigma})^{2}\ln(1/\epsilon)\frac{1}{1-\bar n_{s}/n}.
\end{align}
Define \(\bar n_{s,\lb} = 32 (\theta/\sigma)^{2}\ln(1/\epsilon)\).
Then the condition is satisfied if \(\bar n_{s}>\bar n_{s,\lb}\) and
\(n\geq \bar n_{s}^{2}/(\bar n_{s}- \bar n_{s,\lb}) = \bar
n_{s}/(1-\bar n_{s,\lb}/\bar n_{s})\). For the remainder of this
section, we assume that these two inequalities are satisfied.

We consider quantities to lowest order in \(1/n\) and estimate to
lowest order in \(1/n\).  Let
\(c=(\theta/\sigma)\sqrt{2\ln(1/\epsilon)}\).  Then, the term
\(B(\beta\theta)\) in the bound of
  Eq.~\eqref{eq:bndwithest_constant} is
\(B(c\sqrt{\bar n_{s}}/n + O(1/n^{2}))\).  Using the definition
of the function \(B(z_0)\) in Lem.~\ref{lem:lbetatheta_bounds}
with \(z_0=c\sqrt{\bar n_{s}}/n + O(1/n^{2})\), we have
\begin{align}
  B\qty(c\sqrt{\bar n_{s}}/n + O\qty(1/n^{2}))
  &= (\bar n_{s}^{2}/n^{2}) (1+O(1/n))\qty(1-(1-\bar n_{s}/n)\qty(1+c\sqrt{\bar n_{s}}/n + O\qty(1/n^{2})))^{-2}\notag\\
  &=
    (\bar n_{s}^{2}/n^{2}) (1+O(1/n))\qty(\bar n_{s}/n - c\sqrt{\bar n_{s}}/n + O(1/n^{2}))^{-2}
    \notag\\
  &=    (\bar n_{s}^{2}/n^{2})(1+O(1/n))
    \qty((\bar n_{s}/n) (1 - c/\sqrt{\bar n_{s}})(1 + O(1/n)))^{-2}
    \notag\\
  &= (1 + c/(\sqrt{\bar n_{s}}-c))^{2}(1 + O(1/n)).
\end{align}
In view of this and applying \(\sqrt{1-\bar n_{s}/n} = 1 + O(1/n)\), 
the bound on the expected gap in Eq.~\eqref{eq:bndwithest_constant}
 can therefore be written as 
\begin{align}
  \frac{1}{n}\Exp_{\nu}(\Delta)
  & \leq \sigma\sqrt{\frac{\ln(1/\epsilon)}{2\bar n_{s}}}\qty(1 + O(1/n))
        \qty(1+ (1 + c/(\sqrt{\bar n_{s}}-c))^{2}(1 + O(1/n))) \notag \\
  &=  \sigma\sqrt{\frac{\ln(1/\epsilon)}{2\bar n_{s}}}
          \qty(2+2\qty(c/(\sqrt{\bar n_{s}}-c))\qty(1+ c/(2(\sqrt{\bar n_{s}}-c)))) \qty(1 + O(1/n))  \notag  \\
  &=   \sigma\sqrt{\frac{2\ln(1/\epsilon)}{\bar n_{s}}}
          \qty(1+\qty(c/(\sqrt{\bar n_{s}}-c))(1+c/(2 (\sqrt{\bar n_{s}}-c)))) \qty(1 + O(1/n))  \notag  \\      
  &=
          \sigma\sqrt{\frac{2\ln(1/\epsilon)}{\bar n_{s}}}
      \qty(
    1+\frac{\theta}{\sigma (\sqrt{\bar n_{s}}-c)}\sqrt{2\ln(1/\epsilon)}\qty(1+
    \frac{\theta}{\sigma (\sqrt{\bar n_{s}}-c)}\sqrt{\frac{\ln(1/\epsilon)}{2}})
    ) (1+O(1/n)).
    \label{eq:avgapns}
\end{align}
This indicates that the leading order of \(\Exp_{\nu}(\Delta)/n\)
scales as \(1/\sqrt{\bar n_{s}}\), which is proportional to the
standard deviation of the sample mean estimated from \(\bar n_{s}\)
spot-checked trials.  We note that the leading order is the same as
that obtained from Eq.~\eqref{eq:tildedelta_asymbound} when neither
regularization nor calibration is used.

It is necessary to check that the actual gap has an acceptably narrow
distribution around the expected gap.  For this we need finite bounds
on \(m_{3,\rm{abs}}\) and \(m_{4}\), the third absolute and the fourth
central moments of \(X\) with respect to \(\nu\). \ The third absolute
central moment \(m_{3,\rm{abs}}\) is finite if the fourth central
moments \(m_{4}\) and the variance \(\sigma^2\) are finite. It can be
bounded by the Cauchy-Schwartz inequality as follows:
\begin{align}
  m_{3,\rm{abs}} &= \Exp_{\nu}(|X-\theta|^{2}|X-\theta|)
          \notag\\
        &\leq \qty(\Exp_{\nu}|X-\theta|^4)^{1/2} \qty(\Exp_{\nu}|X-\theta|^{2})^{1/2}
          \notag\\
         &= m_{4}^{1/2}\sigma.
\end{align}
Here we consider only the case \(U=\Exp(X|\Past)=\theta\).  The
statistics of the gap are then completely determined by the statistics
of the logarithm of the estimation factor \(T_{e}(X,Y)\) used. The
contribution of the estimation factor to \(\Delta/n\) is the sum of
\(n\) i.i.d samples of \(-\ln\qty(T_{e}(X,Y))/(\beta n)\).  The
variance of \(\Delta/n\) is thus given by
\begin{align}
  \Var(\Delta/n)
  &= n\Var\qty(\knuth{Y=1}\theta/n-\ln\qty(T_{e}(X,Y))/(\beta n))
    \notag\\
  &=\frac{1}{\beta^{2}n}\Var\qty(\knuth{Y=1}\beta\theta-\ln\qty(T_{e}(X,Y)))
    \notag\\
  &=\frac{\sigma^{2}n}{2\bar n_{s}\ln(1/\epsilon)}\Var\qty(\knuth{Y=1}\beta\theta-\ln\qty(T_{e}(X,Y)))(1-\bar n_{s}/n)
    \label{eq:vardeltan}
\end{align}
after substituting the expression for \(\beta\) noted above.
Below we show that the variance \(\Var(\Delta/n)\) is close to 
\(\sigma^{2}/\bar n_{s}\), which is the variance of the average 
of \(\bar n_{s}\) i.i.d. spot-checked samples of \(X\), as one might 
expect. Comparing the square root of this variance to the expected
gap in Eq.~\eqref{eq:avgapns}, we find that the former is smaller by a factor
of \(\sqrt{2\ln(1/\epsilon)}\), suggesting that the actual gap
is typically not much larger or smaller than the expected gap.  

\newcommand{\Osmaller}{O(\text{smaller})} 
For the calculation of the variance \(\Var\qty(\knuth{Y=1}\beta\theta-\ln\qty(T_{e}(X,Y)))\), 
we compute only the leading-order term in \(1/n\), and within that,
the leading-order term in \(1/\bar n_{s}\).  Terms of 
higher order in \(1/n\) or of the same order in \(1/n\) but higher order
in \(1/\bar n_{s}\) are neglected.  For simplicity, we use 
the notation ``\(+\Osmaller\)'' to indicate the
suppression of such higher-order terms without 
specifying them explicitly. The degree of the
explicitly stated term determines the meaning of ``\(+\Osmaller\)''.

According to the law of total variance, 
\begin{align}
  \Var\qty(\knuth{Y=1}\beta\theta-\ln\qty(T_{e}(X,Y)))
  &=
    (\bar n_{s}/n)\Var\qty(\ln(T_{e}(X,0))|Y=0)
    +(1-\bar n_{s}/n)\Var\qty(\ln(T_{e}(X,1))|Y=1)
    \notag\\
  &\hphantom{=(\bar n_{s}/n)\Var\qty(\ln\qty(T_{e}(X,0)|Y=0))}
          + \Var\qty(\Exp_{\nu}\qty(\knuth{Y=1}\beta\theta-\ln(T_{e}(X,Y))|Y)) \nonumber \\
          &= (\bar n_{s}/n)\Var\qty(\ln(T_{e}(X,0))|Y=0) + \Var\qty(\Exp_{\nu}\qty(\knuth{Y=1}\beta\theta-\ln(T_{e}(X,Y))|Y)).
          \label{eq:total_variance}
\end{align}
For the second line, we applied the
identity \(\Var\qty(\ln(T_{e}(X,1))|Y=1)=0\), which holds because
\(T_{e}(X,1)\) does not depend on \(X\).
\(\Exp_{\nu}(\knuth{Y=1}\beta\theta-\ln(T_{e}(X,Y))|Y)\) is a binary RV that takes the value
\(\Exp_{\nu}(-\ln(T_{e}(X,0)))\) with probability \( \bar n_{s}/n\) and
the value \(\Exp_{\nu}(\beta\theta-\ln(T_{e}(X,1)))=0\) with
probability \((1-\bar n_{s}/n)\), where the latter equality follows from \(T_{e}(X,1)=t_{e}=e^{\beta \theta}\). We have
\begin{align}
     \ln(T_{e}(X,0)) &= \ln(\frac{n}{\bar n_{s}}) +  \ln(1- (1-\bar n_{s}/n)e^{-\beta (X-\theta)}).
\end{align}
We effectively already estimated \(-\Exp_{\nu}(\ln(T_{e}(X,0)))\) in
Eq.~\eqref{eq:avgapns}. For the chosen parameters and
from Eqs.~\eqref{eq:def_lbetatheta} and~\eqref{eq:expdelta_lbetatheta},  we have
\begin{align}
  \Exp_{\nu}\qty(\ln(T_{e}(X,0)))
  &=
    -\qty(\beta \Exp_{\nu}(\Delta) + \ln(\epsilon))/(n\omega)
    \notag\\
  &=-\sqrt{\frac{2\ln(1/\epsilon)}{\sigma^{2}\bar n_{s}}}\frac{\Exp_{\nu}(\Delta)}{n}
    +\frac{\ln(1/\epsilon)}{ \bar n_{s}}+\Osmaller.
\end{align}
Because \(0\leq\Exp_{\nu}(\Delta)/n \leq \sigma\sqrt{2\ln(1/\epsilon)/\bar n_{s}}+ \Osmaller\),
\begin{align}
  |\Exp_{\nu}(\ln(T_{e}(X,0)))|
  &\leq \frac{\ln(1/\epsilon)}{\bar n_{s}}  + \Osmaller.
\end{align}
The variance of a binary RV,  with probability \(\omega\) of taking the lower value and
a difference \(d\) between the two values, is given by \(\omega(1-\omega)d^{2}\).
From the above calculations, 
for the binary RV \(\Exp_{\nu}(\knuth{Y=1}\beta\theta-\ln(T_{e}(X,Y))|Y)\),  we have 
\(d \leq \ln(1/\epsilon)/\bar n_{s}+ \Osmaller\). Therefore, we obtain
\begin{align}
  \Var\qty(\Exp_{\nu}\qty(\knuth{Y=1}\beta\theta-\ln(T_{e}(X,Y))|Y))
  &\leq \frac{\ln^{2}(1/\epsilon)}{\bar n_{s} n} + \Osmaller.
  \label{eq:total_variance_part1}
\end{align}

It remains to bound
\begin{align}
  \Var(\ln(T_{e}(X,0))|Y=0) &= \Var\qty(\ln(1- (1-\bar n_{s}/n)e^{-\beta (X-\theta)}))\notag\\
                        &= \Var(f(\beta X)),
\end{align}
where the function \(f(z)\) is defined in the proof of Lem.~\ref{lem:lbetatheta_bounds}, 
with parameters \(a=(1-\bar n_{s}/n)\) and \(z_{0}=\beta \theta\) in the definition
 there.  In that proof, we showed that the function \(f(z)\) is convex and that its positive second derivative 
 is maximized at \(z=0\).   For any RV \(Z\) and constant \(z_0\), 
 we have  \(\Exp\qty((Z-z_{0})^2)=\Var(Z)+\qty(\Exp(Z)-z_0)^2\).  Therefore, we can upper-bound 
 the variance  \(\Var(f(\beta X))\) by \(\Exp_{\nu}\qty((f(\beta X)-f(\beta \theta))^{2})\).
By the properties of \(f(z)\) and in view of Eq.~\eqref{eq:fx_expansion0}, 
\begin{align}
  f(\beta X)-f(\beta\theta))
  &\geq f'(\beta\theta)(\beta X-\beta \theta)
    \notag\\
  f(\beta X)-f(\beta\theta))
  &\leq f'(\beta\theta)(\beta X-\beta \theta)
    +(f''(0)/2) (\beta X-\beta \theta)^{2}. 
\end{align}
Hence, 
\begin{align}
  |f(\beta X)-f(\beta\theta)| &\leq \max\qty(|f'(\beta\theta)(\beta X-\beta \theta)|, |f'(\beta\theta)(\beta X-\beta \theta)
    +(f''(0)/2)(\beta X-\beta \theta)^{2} |)\nonumber \\
  &\leq
  \beta|f'(\beta\theta)(X-\theta)|
    +\beta^{2}(f''(0)/2)(X-\theta)^{2}.
\end{align}
We square and take expectations to get 
\begin{align}
\Exp_{\nu}\qty((f(\beta X)-f(\beta \theta))^{2}) & \leq \Exp_{\nu} \qty(\qty(\beta|f'(\beta\theta)(X-\theta)|
    +\beta^{2}(f''(0)/2)(X-\theta)^{2})^2)  \nonumber \\
    & \leq \beta^2 \qty(f'(\beta\theta))^2 \Exp_{\nu}\qty((X-\theta)^{2})+\beta^3  |f'(\beta\theta)| f''(0) \Exp_{\nu}\qty(|X-\theta|^{3})+ \beta^4 \qty(f''(0))^2 \Exp_{\nu}\qty((X-\theta)^{4})/4 \nonumber \\
    & =   \beta^2 \qty(f'(\beta\theta))^2 \sigma^2 +\beta^3  |f'(\beta\theta)| f''(0) m_{3, \rm{abs}}+ \beta^4 \qty(f''(0))^2 m_4/4.
     \end{align}
Therefore, we have 
\begin{align}
  \Var(\ln(T_{e}(X,0))|Y=0)
  &\leq
    \beta^{2}\qty(
    \qty(f'(\beta\theta))^{2}\sigma^{2}
    + \beta |f'(\beta\theta)|f''(0) m_{3, \rm{abs}}
     +\beta^{2}\qty(f''(0))^{2}m_{4}/4
    ).
    \label{eq:varbound1}
\end{align}
For the contributions from the derivatives of \(f\), a direct calculation based on the 
expressions for \(f'(z)\) and \(f''(z)\) in the proof of Lem.~\ref{lem:lbetatheta_bounds} 
 shows that \(|f'(\beta\theta)| =a/ (1-a)= n/\bar n_{s}-1=n/\bar n_{s}+ \Osmaller\)
and \(f''(0)= aB(\beta\theta)/(1-a)^2 =n^2/\bar n^2_{s}+ \Osmaller\).
 Furthermore,  since \(\beta =\sqrt{2\bar n_{s}\ln(1/\epsilon)/(n^{2}\sigma^{2}(1-\bar n_{s}/n))}\),
the first term in the right-hand side of Eq.~\eqref{eq:varbound1} dominates,  and 
 the terms involving \(m_{3,\rm{abs}}\) and \(m_{4}\) can be neglected under
the assumption that  \(m_{3,\rm{abs}}\) and \(m_{4}\) are finite. We now have
\begin{align}
  \Var(\ln(T_{e}(X,0))|Y=0)
  & \leq  \beta^{2}\qty(f'(\beta\theta))^{2}\sigma^{2} + \Osmaller \nonumber \\
  &= \frac{2\ln(1/\epsilon)}{\bar n_{s}} + \Osmaller. 
  \label{eq:total_variance_part2}
  \end{align}
In view of Eqs.~\eqref{eq:total_variance},~\eqref{eq:total_variance_part1}, and~\eqref{eq:total_variance_part2}, 
\begin{align}
\Var(\knuth{Y=1}\beta\theta-\ln(T_{e}(X,Y))) & \leq  \frac{\bar n_{s}}{n} \qty(\frac{2\ln(1/\epsilon)}{\bar n_{s}} + \Osmaller) + \qty(\frac{\ln^{2}(1/\epsilon)}{\bar n_{s} n} + \Osmaller)  \nonumber \\
&=   \frac{2\ln(1/\epsilon)}{n} + \Osmaller.
\end{align}
 Substituting this expression in Eq.~\eqref{eq:vardeltan} gives
\begin{align}
    \Var(\Delta/n)
  &\leq
  \frac{\sigma^{2}}{\bar n_{s}} + \Osmaller.
\end{align}
Up to the neglected terms, this confirms that \(\Var(\Delta/n)\)
is the same as the variance of the average of \(\bar n_{s}\)
i.i.d. spot-checked samples of \(X\).

\section{Asymptotic tightness of estimation-factor-based confidence bounds}
\label{sect:tightness}

We say that a confidence bound is asymptotically tight if the expected
gap between the bound and the estimated quantity is of the same
order as the standard deviation of the distribution of the
quantity.  Sect.~\ref{sect:analytical_opt} establishes the
asymptotic tightness  of estimation-factor-based confidence bounds in the
setting where the trials are i.i.d.  In this section, we extend the result to the setting where,
in each trial $i$, the RV $X_i$ is non-negative with constant
past-conditional mean and a variance that is upper-bounded conditional
on the past.  In this setting, the trials are not necessarily i.i.d.

A desirable property of lower confidence bounds is that their
expectations are bounded from above by the expectation of the
estimated quantity.  We previously noted that this holds for 
estimation-factor-based lower confidence bounds as a 
consequence of Eq.~\eqref{eq:obtbnd2}. We begin
by formalizing this observation.

\begin{proposition}\label{thm:soundness}
  Consider a sequence of \(n\) trials that are not necessarily i.i.d.,
  where at the \(i\)'th trial,  the spot-checking probability conditional 
  on the past, \(\Prob(Y_{i}=0|\past_{i})\), is known and denoted by \(\omega_i\). 
   The confidence bound \(S_{\lb}\) provided by an arbitrary
  sequence of $\EF$s designed at the \(i\)'th trial for \(\omega_{i}\)
   satisfies $\Exp(S_{\lb}) \leq \Exp(S_{U})$.
\end{proposition}

\begin{proof}
  Since the confidence bound and the estimated quantity are both
    additive over trials, it suffices to consider the contributions
    for each trial separately.  We have
    \begin{align}
      \Exp(S_{U})
      &=\sum_{i=1}^{n}\Exp(Y_{i}X_{i})
        \notag\\
      &= \sum_{i=1}^{n}\Exp(\Exp(Y_{i}X_{i}|\Past_{i}))
        \notag\\
      &= \sum_{i=1}^{n}\Exp((1-\omega_{i})\Exp(X_{i}|\Past_{i}))\notag\\
      &=\Exp\qty(\sum_{i=1}^{n}(1-\omega_{i})\Exp(X_{i}|\Past_{i})),
    \end{align}
    where the outer expectation averages over the values of \(\Past_{i}\).   
    The contribution of trial \(i\) to the estimated quantity, conditional on
    a value \(\past_{i}\) of \(\Past_{i}\) is
    \((1-\omega_{i})\Exp(X_{i}|\past_{i})\). As shown in the main
    text, the confidence bounds are optimized for $\EF$s of the form
    in Eq.~\eqref{eq:t_maximal}, with \(X\), \(Y\), and \(t\) replaced
    by the trial-wise variables \(X_i\), \(Y_i\), and \(t_i\).  Then,
    the confidence bound \(S_{\lb}\) in Eq.~\eqref{eq:conf_lb} has
    expectation value
  \begin{equation}
    \label{eq:exp_conf_bound_true1}
    \Exp(S_{\lb})=\Exp\qty(\sum_{i=1}^{n}\qty(\omega_{i} \Exp \qty(\ln(F'_{\beta, t_i}(X_i))|\Past_{i}) +(1-\omega_{i})\ln(t_i) )/\beta+\ln(\eps)/\beta),
  \end{equation}
  where $F'_{\beta, t_i}(X_i)=\big(1-(1-\omega_{i})t_i e^{-\beta
      X_i}\big)/\omega_{i}$. Thus, it suffices to
  show that the following inequality holds for
    the contribution to \(\Exp(S_{\lb})\) from trial \(i\), conditional on \(\past_{i}\):
  \begin{equation} \label{eq:soundness_proof2}
    \Big(\omega_{i} \Exp\qty(\ln(F'_{\beta, t_i}(X_i))|\past_{i}) +(1-\omega_{i})\ln(t_i)+ \ln(\eps)/n \Big)/\beta  \leq (1-\omega_{i}) \Exp (X_i|\past_{i}).
  \end{equation}
  This inequality is established for a generic trial by Eq.~\eqref{eq:obtbnd2}.
\end{proof}

For asymptotic tightness, the main result is given by the next theorem.

\begin{theorem}\label{thm:tightness}
  Let \(\theta \geq 0\) and \(\sigma^2>0\) be finite.  Consider a
  sequence of \(n\) trials, where in each trial \(i\), the RV
  \(X_i \geq 0\) follows a distribution with unknown but constant
  past-conditional mean \(\Exp(X_i|\Past_i)=\theta\) and unknown
  past-conditional variance satisfying
  \(\Var(X_i |\Past_i) \leq \sigma^2\).  Let the spot-checking RV
  \(Y_i\) be independent of the past with constant 
  \(\Prob(Y_i=0)=\omega \in (0,1)\). 
  Suppose that an estimate \(\theta_e\geq 0\) of the unknown mean
  \(\theta\) is given before observing the sequence of trials.  Then
  there exists a sequence of identical \(\EF\)s that yields a
  confidence bound \(S_{\lb}\) satisfying
  \( \Exp(S_{\lb}) \geq \Exp(S_{U}) - \sqrt{n
    (1-\omega)\ln(1/\eps)/(2\omega)} \big((\theta-\theta_e)^2/\sigma +
  2\sigma \big) - O(1) \) for sufficiently large \(n\).
\end{theorem}

Three remarks are in order. First, the estimate \(\theta_{e}\) may be
chosen in an arbitrary way and need not be accurate. Second, if
\(\theta_{e}=\theta\), the theorem implies that the gap between
\(\Exp(S_{U})\) and \(\Exp(S_{\lb})\) is
\(\sigma \sqrt{2n (1-\omega)\ln(1/\eps)/\omega}+O(1)\). This gap is
essentially optimal, as it cannot be significantly reduced further in
view of the observations following Eq.~\eqref{eq:expgap}.  Third,  if 
the weak law of large numbers were applicable, based on this theorem we
expect that the confidence bound $S_{\lb}$ exceeds
$\Exp(S_{U})-O(\sqrt{n})$ with high probability for large \(n\).
Since the estimation factors depend on \(n\), the weak law of large
numbers is not immediately applicable, even for
i.i.d. trials. However, the variance calculation for the gap between \(S_{U}\) and
\(S_{\lb}\) given a constant expected number of spot-checked trials in the
second part of Sect.~\ref{sect:constant_omega_n} suggests that our
expectation is justified.

\begin{proof} The proof follows the same structure as the derivation
  of the upper bound obtained in Eq.~\eqref{eq:expgap}, except that here
  we consider non-i.i.d. trials.  We first
  note that $\Exp(S_U)=n(1-\omega)\theta$, regardless of whether
  \(U=\Exp(X|\Past)\) or \(U=X\) for a generic trial.

  Consider $\EF$s of the form given in Eq.~\eqref{eq:t_maximal} for a
  generic trial, which are parameterized by $\beta>0$ and $t\in [0,
  1/(1-\omega)]$.  Using a sequence of such identical $\EF$s, the
  expectation of the confidence bound $S_{\lb}$ is given in
    Eq.~\eqref{eq:exp_conf_bound_true1}, replacing \(\omega_{i}\)
    and \(t_i\) there with \(\omega\) and \(t\) here.    As in
  Sect.~\ref{sect:analytical_opt}, we set $t=e^{\beta \theta_e}$, where
  the estimate $\theta_e \geq 0$ is obtained before the trials 
  and satisfies the constraint $e^{\beta \theta_e} \leq 1/(1-\omega)$.
  Equivalently, we have $\beta \in (0, \beta_{\text{ub}}]$, where
  \begin{equation} \label{eq:beta_ub}
    \beta_{\text{ub}}=\ln(1/(1-\omega))/\theta_e.
  \end{equation} 
  The gap between $\Exp(S_U)$ and $\Exp(S_{\text{lb}})$ can then be written as
  \begin{align}\label{eq:gap_estimation}
    \Exp(S_{U})-\Exp(S_{\text{lb}})&= \sum_{i=1}^{n}\qty(\omega
                                     \ln(\omega)-\omega \Exp \qty(\ln(1-(1-\omega) e^{-\beta
                                     (X_i-\theta_e)}) )+\beta(1-\omega) (\theta-\theta_e)
                                     )/\beta+\ln(1/\eps)/\beta \notag \\ 
                                    &= \sum_{i=1}^{n} \Exp\qty(\Exp \qty(\omega
                                     \ln(\omega)-\omega \ln(1-(1-\omega) e^{-\beta
                                     (X_i-\theta_e)})+\beta(1-\omega) (\theta-\theta_e)
                                      | \Past_{i})) /\beta+\ln(1/\eps)/\beta \notag \\ 
                                     &= \sum_{i=1}^{n}
                                            \Exp\qty(\Exp\qty(\cL_{i}(\beta, \theta_e)| \Past_{i}))/\beta+\ln(1/\eps)/\beta,
  \end{align} 
  where the function $\cL(\beta, \theta_e)$ is defined for a
  generic trial in Eq.~\eqref{eq:def_lbetatheta}, and 
  $\cL_{i}(\beta, \theta_e)$ denotes its instantiation for trial $i$.
  Since the conditional distributions for trial \(i\)  given values \(\past_i\) of the past that occur with non-zero probability 
  are in the model \(\cF\) of the generic trial as defined in Sect.~\ref{sect:EF_method} and, by
  assumption, have mean \(\theta\) and variance \(\sigma_{i |\past_i}^{2} =\Var(X_i |\past_i)\),
  the bound of Lem.~\ref{lem:lbetatheta_bounds} applies conditional on the past
  with the given mean and variance.
  We therefore obtain for each trial \(i\) and past value \(\past_i\),
    \begin{align}
   \Exp(\cL_{i}(\beta,\theta_{e})|\past_{i}) & \leq \beta^2 \qty(\sigma_{i|\past_i}^{2}+(\theta-\theta_{e})^{2}) B(\beta\theta_{e})(1-\omega)/(2\omega) \notag \\
                                         & \leq \beta^2 \qty(\sigma^{2}+(\theta-\theta_{e})^{2}) B(\beta\theta_{e})(1-\omega)/(2\omega),
                                             \label{eq:gap_term_bound}
  \end{align}
  where the second line is a consequence of the constraint that $\sigma_{i |\past_i}^2\leq \sigma^2$ for all $i$ and \(\past_i\). 
  Taking the expectation over \(\Past_i\) and summing over \(i\) yields
  \begin{align}\label{eq:gap_bound} \Exp(S_{U})-\Exp(S_{\text{lb}}) &
                                                                      \leq \sum_{i=1}^{n} \beta \big( \sigma^2+(\theta-\theta_e)^2 \big)
                                                                      B(\beta\theta_e)(1-\omega)/(2\omega) + \ln(1/\eps)/\beta \notag \\ &
                                                                                                                                           = n\beta \big( \sigma^2+(\theta-\theta_e)^2 \big)
                                                                                                                                           B(\beta\theta_e)(1-\omega)/(2\omega) + \ln(1/\eps)/\beta.
  \end{align} 
  
  Next, we determine a sufficiently good choice for $\beta$.  For this,
  we set $\beta=\beta_1$, where $\beta_1$ is given in
  Eq.~\eqref{eq:sub_opt_beta}.  To be a feasible choice, $\beta_1$ must
  not exceed $\beta_{\text{ub}}$ introduced above.  Equivalently, we
  must have that $n\geq n_{\thresh}$, where $n_{\thresh}
  =2\omega\theta_e^2 \ln(1/\eps)/\qty((1-\omega)\sigma^2
  \ln^2(1/(1-\omega)))$.  Therefore, when $n$ is sufficiently large, we
  can set $\beta=\beta_1$ and then
  $B(\beta\theta_e)=1+O(1/\sqrt{n})$. Accordingly, the upper bound in
  Eq.~\eqref{eq:gap_bound} becomes
  \begin{align}\label{eq:estimation_gap2}
    \Exp(S_{n})-\Exp(S_{\text{lb}}) & \leq \sqrt{n(1-\omega)
                                      \ln(1/\eps)/(2 \omega \sigma^2)} \big( \sigma^2+(\theta-\theta_e)^2
                                      \big) \big(1+O(1/\sqrt{n})\big) + \sqrt{n
                                      (1-\omega)\ln(1/\eps)/(2\omega)}\sigma \notag \\ & = \sqrt{n
                                                                                         (1-\omega)\ln(1/\eps)/(2\omega)} \big((\theta-\theta_e)^2/\sigma +
                                                                                         2\sigma \big)+O(1),
  \end{align} which proves the theorem. Note that unlike the upper bound
  obtained in Eq.~\eqref{eq:expgap}, we do not set $\theta_e=\theta$ in
  this proof, as $\theta$ is treated as an unknown parameter;  furthermore, \(\sigma\)
  here is an upper bound on the past-conditional variance in each trial \(i\).
\end{proof}
  
The results above, together with those in
Sect.~\ref{sect:analytical_opt}, suggest a simple and practical
construction of estimation factors aimed at maximizing the expected
lower confidence bound with finite data. Specifically, to construct a
valid estimation factor according to Eq.~\eqref{eq:t_maximal}, we can
set $\beta=\min(\beta_1, \beta_{\text{ub}})$ and
$t=e^{\beta \theta_e}$, where \(\theta_{e}\) is an estimate of the
mean before observing the sequence of spot-checking trials, and
$\beta_1$ and $\beta_{\text{ub}}$ are given in
Eqs.~\eqref{eq:sub_opt_beta} and~\eqref{eq:beta_ub}. 
 The determination of \(\beta_{1}\) in
  this construction also requires a bound \(\sigma^{2}\) on the
  variance,  which need not be strict and may be derived from an
estimate of the variance before the trials.  This technique can work
well even when the assumptions on the means and variances in
Thm.~\ref{thm:tightness} are not strictly satisfied. 
More generally, good results can be obtained for a broad range of
choices for \(\theta_{e}\) and \(\sigma^{2}\).  For example, if in
each trial the RV $X_i$ is bounded as $0\leq X_i \leq u$, we can set
$\theta_e=u/2$. In this case, since $X_i$ is bounded, its variance
$\sigma_i^2$ satisfies $\sigma_i^2 \leq \sigma^2=u^2/4$.
Consequently, when $0\leq X_i \leq u$ for each trial $i$, we can set
\begin{equation}\label{eq:analytic_EF}
  \beta =\min(\beta_1, \beta_{\text{ub}})= \min\Big(\sqrt{\frac{8\omega  \ln(1/\eps)}{n  (1-\omega) u^2}}, \frac{2\ln(1/(1-\omega))}{u}\Big) \text{ and } t=e^{\beta u/2}
\end{equation} 
to construct a valid estimation factor that can be used to analyze a
general sequence of non-i.i.d. trials. Our
numerical study indicates that these estimation factors perform well,
if not asymptotically optimally, with finite data;  see
Fig.~\ref{main-fig:tight_bound} in the main text and 
Figs.~\ref{SI-fig:tight_bound} and~\ref{SI-fig:limit_bound} in Sect.~\ref{sect:comps} for 
illustrations.  Importantly, this construction requires no experimental
calibration or simulation.

\section{Finite-data efficiency of estimation factors}
\label{sect:efficiency}

In Sect.~\ref{sect:numerical_opt_nmin}, we determined the minimum
value of \(n\) that ensures the expected gap per unchecked trial does
not exceed a given threshold \(\delta_{\thresh}\), under the
assumption that the reference distribution is representative of the
true distribution. This minimum corresponds to the smallest \(n\) that
satisfies the inequality constraint in Eq.~\eqref{eq:gapbnd}. Here, we
simplify this inequality constraint and analyze the asymptotic
behavior of the minimum value as $\delta_{\thresh} \rightarrow 0$.

As in Sect.~\ref{sect:numerical_opt}, we consider the case of
i.i.d. trials, and  let \(\nu\) denote the joint reference distribution of \(Y\) and
\(X\) in a generic trial, where \(Y\) and \(X\) are independent,
\(\Prob(Y=0)=\omega\), and \(X\) is bounded from below by \(0\). 
 Suppose \(\Exp_{\nu}(X)=\theta\).  Then, we have
\(\Exp(S_{U})=n(1-\omega)\theta\), regardless of whether
\(U=\Exp(X|\Past)\) or \(U=X\). In addition,
\(\Exp(\sum_{i=1}^{n}Y_{i})=n(1-\omega)\).  With the $\EF$s used in
the proof of Thm.~\ref{thm:tightness}, the expectation
\(\Exp(S_{\lb})\) can be written as in
Eq.~\eqref{eq:exp_conf_bound_true1}, replacing \(\omega_{i}\)
    and \(t_i\) with \(\omega\) and \(t\).  Because the reference
distribution of \(X\) is known, its mean \(\theta\) is also
known. Thus, we can set \(t=e^{\beta \theta_e}\) with
\(\theta_{e}=\theta-\delta_{\thresh}\) in the expression for
\(\Exp(S_{\lb})\). Moreover, since the trials are i.i.d.,
we obtain the following simplified expression for \(\Exp(S_{\lb})\):
\begin{align}\label{eq:exp_conf_bound_true2-1}
\Exp(S_{\text{lb}}) &= n\omega \Exp_{\nu} \qty(\ln(F'_{\beta, t}(X)))/\beta +n(1-\omega)(\theta-\delta_{\thresh})+\ln(\eps)/\beta, 
\end{align}
where $F'_{\beta, t}(X)=\big(1-(1-\omega) e^{-\beta(X-\theta+\delta_{\thresh})}\big)/\omega$. 
The inequality constraint in Eq.~\eqref{eq:gapbnd} becomes
\begin{align}
  n\omega \Exp_{\nu} \qty(\ln(F'_{\beta, t}(X)))/\beta +\ln(\eps)/\beta \geq 0.
  \label{eq:gapbnd-2}
\end{align}
To ensure that \(t=e^{\beta \theta_e}\) is a valid choice, \(\beta\)
must not exceed \(\beta_{\text{ub}}\) given in Eq.~\eqref{eq:beta_ub}.
After optimizing the $\EF$ over \(\beta\), the corresponding minimum
value for \(n\) is given by
\begin{equation}\label{eq:min_exp_num}
  n_{\text{lb}, \min}=\inf_{0<\beta \leq  \beta_{\text{ub}} } \ln(1/\eps)\qty(\omega \Exp_{\nu} \Big(\ln(F'_{\beta, t}(X))\Big))^{-1}
  = \ln(1/\eps)\qty(\omega \sup_{0<\beta \leq  \beta_{\text{ub}} }  \Exp_{\nu} \Big(\ln(F'_{\beta, t}(X))\Big))^{-1},
\end{equation}
where in computing the infimum, we exclude cases where the argument
is negative.
For each value $x$ of the RV $X$, $F'_{\beta, t}(x)$ is a concave
function of $\beta$. Because \(\ln(z)\) is concave and monotonically
increasing in \(z\), $\ln\big(F'_{\beta, t}(x)\big)$ is also concave
in $\beta$.  Moreover, in view of the fact that a convex combination
of concave functions is itself concave, the expectation
$\Exp_{\nu} \Big(\ln\big(F'_{\beta, t}(X)\big)\Big)$ is a concave
function of $\beta$. In addition, the feasible region
$0<\beta \leq \beta_{\text{ub}}$ is convex.  Thus we can find the
minimum value $n_{\text{lb}, \min}$ by convex programming.  Numerical
evidence indicates that the above optimization problem yields results
indistinguishable from those obtained with the minimization problem
formulated in Sect.~\ref{sect:numerical_opt_nmin} for the parameters
we considered. The results shown in
Fig.~\ref{main-fig:finite_efficiency} of the main text and in
Fig.~\ref{SI-fig:finite_efficiency} of Sect.~\ref{sect:comps} were
obtained by solving the optimization problem in
Eq.~\eqref{eq:min_exp_num}.

We now analyze the asymptotic behavior of $n_{\text{lb}, \min}$ as $\delta_{\thresh} \rightarrow 0$,  
assuming that the variance $\Exp_{\nu} (X-\theta)^2$ is upper-bounded by $\sigma^2$. 
In view of the inequality established in Eq.~\eqref{eq:fx_expansion} and  
$F'_{\beta, t}(x)=\big(1-(1-\omega) e^{-\beta(x-\theta+\delta_{\thresh})}\big)/\omega$, we have that for all $x\geq 0$,
\begin{equation} \label{eq:taylor_expan3}
  \ln(F'_{\beta, t}(x)) \geq \beta(x-\theta+\delta_{\thresh})(1-\omega)/\omega - \beta^2(x-\theta+\delta_{\thresh})^2 B\big(\beta(\theta-\delta_{\thresh})\big)(1-\omega)/(2\omega^2),
\end{equation}
where the non-negative function \(B(z)\), defined in
Lem.~\ref{lem:lbetatheta_bounds}, satisfies $B(z)=1+O(z)$ as
$z \to 0$. Therefore,
\begin{align} \label{eq:exp_log_EF}
  \Exp_{\nu} \big(\ln(F'_{\beta, t}(X))\big) & \geq \beta(\Exp_{\nu}(X)-\theta+\delta_{\thresh})(1-\omega)/\omega - \beta^2 \Exp_{\nu}(X-\theta+\delta_{\thresh})^2 B\big(\beta(\theta-\delta_{\thresh})\big)(1-\omega)/(2\omega^2) \notag \\
                                             & \geq \beta(\theta-\theta+\delta_{\thresh})(1-\omega)/\omega - \beta^2 \big(\Exp_{\nu}(X-\theta)^2+\delta_{\thresh}^2 \big) B\big(\beta(\theta-\delta_{\thresh})\big)(1-\omega)/(2\omega^2) \notag \\
                                             & \geq \beta\delta_{\thresh}(1-\omega)/\omega - \beta^2(\sigma^2+\delta_{\thresh}^2) B\big(\beta(\theta-\delta_{\thresh})\big)(1-\omega)/(2\omega^2).
\end{align}
Here, the second line follows from the facts that $\Exp_{\nu}(X)=\theta$ and $\Exp_{\nu} (X-\theta+\delta_{\thresh})^2=\Exp_{\nu} (X-\theta)^2+\delta_{\thresh}^2$,
and the last line is a consequence of the assumption that $\Exp_{\nu} (X-\theta)^2 \leq \sigma^2$. 
Next, we need to determine a sufficiently good choice for $\beta$. 
For this, we observe that when $\beta$ is small enough,  the lower bound in Eq.~\eqref{eq:exp_log_EF} 
is well approximated by $\beta\delta_{\thresh}(1-\omega)/\omega - \beta^2(\sigma^2+\delta_{\thresh}^2) (1-\omega)/(2\omega^2)$.  
This approximate lower bound is maximized at 
$\beta=\beta_2$, where  $\beta_2=\omega \delta_{\thresh}/( \sigma^2+\delta_{\thresh}^2)$. 
To be a feasible choice, $\beta_2$ must not exceed $\beta_{\text{ub}}$ given in 
Eq.~\eqref{eq:beta_ub} with \(\theta_{e}=\theta-\delta_{\thresh}\). This is satisfied as long   
as $\delta_{\thresh} \leq \sigma^2 \ln\big(1/(1-\omega)\big)/(\omega \theta )$.   
Therefore, when $\delta_{\thresh}$ is small enough,  the following inequality holds:
\begin{align} \label{eq:sup_exp_log_EF}
  \sup_{0<\beta \leq  \beta_{\text{ub}} }  \Exp_{\nu} \big(\ln(F'_{\beta, t}(X))\big) & \geq  \beta_2 \delta_{\thresh}(1-\omega)/\omega - \beta_2^2(\sigma^2+\delta_{\thresh}^2)\big(1+O(\beta_2)\big)(1-\omega)/(2\omega^2)  \notag \\
                                                                                      &= (1-\omega) \delta_{\thresh}^2/( \sigma^2+\delta_{\thresh}^2)-(1-\omega) \delta_{\thresh}^2\big(1+O(\delta_{\thresh})\big)/\big(2( \sigma^2+\delta_{\thresh}^2)\big) \notag \\
                                                                                      & =(1-\omega) \delta_{\thresh}^2\big(1-O(\delta_{\thresh})\big)/\big(2( \sigma^2+\delta_{\thresh}^2)\big).
\end{align}  
In view of Eq.~\eqref{eq:min_exp_num}, we then have
\begin{equation}\label{eq:min_exp_num_bound2}
  n_{\text{lb}, \min}\leq 2\ln(1/\eps) ( \sigma^2+\delta_{\thresh}^2) \big(1+O(\delta_{\thresh})\big)/\big(\omega (1-\omega) \delta_{\thresh}^2\big). 
\end{equation}
Consequently, as $\delta_{\thresh} \rightarrow 0$,
\(n_{\lb,\min}=O(1/\delta_{\thresh}^{2})\), consistent with the
standard \(1/\sqrt{n}\) scaling of measurement precision for \(n\)
i.i.d. samples. It is in this
sense that the estimation-factor method is efficient with finite data.

We remark that the upper bound in Eq.~\eqref{eq:min_exp_num_bound2}
remains valid under the relaxed assumptions in
Thm.~\ref{thm:tightness}, as can be seen by incorporating the bound
from Eq.~\eqref{eq:gap_bound} into the above derivation of the upper
bound on \(n_{\lb,\min}\). Moreover, the above choices $\beta=\beta_2$
and $t=e^{\beta_2 \theta_{e}}$ provide a simple and practical way to
construct estimation factors aimed at minimizing the number of trials
required to satisfy the prescribed gap \(\delta_{\thresh}\).  This
  construction may be compared to the one described in the last
  paragraph of Sect.~\ref{sect:tightness}.

\section{Protocols}
\label{sect:protocols}

A basic spot-checking protocol consists of four phases: a preparation
phase, which may include calibration; a setup phase, where the
estimation-factor power \(\beta\) and an initial estimation factor are
determined; an experiment phase that consists of a sequence of
spot-checking trials, during which estimation factors may optionally
be updated based on earlier data; and an analysis phase, which
produces lower confidence bounds at specified error bounds from the
values of the estimation factors evaluated for each trial.  The
preparation and setup phases may use a series of calibration trials
that are acquired before the spot-checking trials.  When updates are
performed, they involve the same kind of optimization of estimation
factors as is performed during the setup phase, except that the power
\(\beta\) is fixed.  Although updates can be applied on a
trial-by-trial basis, it is usually more effective to perform them at
less frequent intervals. For simplicity, here we describe the
protocols where the same spot-checking probability and estimation
factor are used for each spot-checking trial. 

A variety of protocols can be derived from the general structure
outlined in the previous paragraph. We view each variant as a
combination of two subprotocols, where the first,
Protocol~\ref{prot:EF_construction} (Estimation-Factor Construction),
covers the preparation and setup phases, and the second,
Protocol~\ref{prot:SC_analysis} (Spot-Checking Experiment and
Analysis), covers the experiment and analysis phases.  The first
subprotocol has multiple options for constructing estimation factors,
depending on whether the goal is to optimize the confidence bound
given a fixed number of trials, or to minimize the number of trials
given a target gap per unchecked trial for the
confidence bound.  The construction further depends on whether
calibration data is used and, if so, which properties of the estimated
reference distribution are leveraged.


\pagebreak

\begin{protocol}[Estimation-Factor Construction]\label{prot:EF_construction} 
\begin{enuma}
\item \textbf{Input}:  
   \begin{itemize} 
    \item Lower bound \(b \in (-\infty,+\infty)\) and upper bound \(u \in (b,\infty]\) on the values of \(X\) (without loss of generality, \(b=0\));
    \item Spot-checking probability \(\omega \in (0,1)\);
    \item Error bound \(\epsilon\);
    \item Number of spot-checking trials \(n\). \textit{Note: This may be ignored depending on the choices below.}
    \item (Optional) number of calibration trials \(n_c\) with \(n_c \geq 2\).   
    \textit{Note: Sect.~\ref{sect:calibratevar} provides a conservative guideline for choosing \(n_c\) to meet protocol goals.}   
 \end{itemize}
 
\item \textbf{Preparation (optional calibration)}:  
  \begin{enuma}
      \item Acquire \(n_{c}\) instances \(\bm{x}=(x_{i})_{i=1}^{n_{c}}\) of \(X\).
      \item Compute the sample mean \(\theta_{e}=\sum_{i=1}^{n_{c}} x_{i}/n_{c}\)
        and unbiased sample variance \(\sigma_{e}^{2}=\sum_{i=1}^{n_{c}}(x_{i}-\theta_{e})^{2}/(n_{c}-1)\).
      \item Estimate a reference distribution \(\nu_{e}\) of \(X\).  
        The distribution can be estimated with the empirical frequencies of the calibration trials, or by fitting a model.
       \item Return \(\bm{x},\theta_{e},\sigma_{e}^{2}\), \(\nu_{e}\). If no calibration is used, proceed without these estimates. 
  \end{enuma}
       
\item \textbf{Setup} (choose \emph{one} option depending on the goal):
  \begin{enuma}
    \item \textbf{Maximizing the expected lower confidence bound.}  
      Given \(n\), choose power \(\beta\) and estimation factor \(T\) by one of the following three methods:
      \begin{enuma}
        \item \label{prot:numerical} Optimize numerically using \(\nu_e\) (Sect.~\ref{sect:numerical_opt_smax});  
        \item \label{prot:variance} Construct analytically using  \(\theta_e\) and \(\sigma_e^2\) (Sect.~\ref{sect:analytical_opt}; see also Sect.~\ref{sect:calibratevar});  
        \item \label{prot:fixed} If \(u<\infty\), compute \(\beta\) and \(t\) via 
        Eq.~\eqref{eq:analytic_EF}, and then determine \(T\) according to Eq.~\eqref{eq:t_maximal}.  \textit{Note: This option does not require calibration data.}    
      \end{enuma}

    \item \textbf{Minimizing the number of trials.}  
      Given a gap threshold \(\delta_{\thresh}\) and the reference distribution \(\nu_{e}\), determine the minimum number of trials \(n\) required and 
      the corresponding \(\beta\) and \(T\) (Sect.~\ref{sect:numerical_opt_nmin}). Alternatively,  given \(\delta_{\thresh}\), 
      \(\theta_e\), and \(\sigma_e^2\), where \(\delta_{\thresh} \leq \sigma_{e}^2 \ln\big(1/(1-\omega)\big)/(\omega \theta_e)\), 
      compute \(\beta\) and \(t\) according to the last paragraph of Sect.~\ref{sect:efficiency},  and then construct the corresponding \(T\) via Eq.~\eqref{eq:t_maximal}.  In this 
      case, the minimum \(n\) required can be estimated according to the bound in Eq.~\eqref{eq:min_exp_num_bound2}
      with the replacement of \(\sigma^2\) there by \(\sigma_{e}^2\) here.   
      \textit{Note: If no calibration data are available and \(u<\infty\), the estimates \(\theta_e\) and \(\sigma_e^2\) may be taken as
      \(\theta_e=u/2\) and \(\sigma_e^2=u^2/4\).}   
     \end{enuma}

\item \textbf{Output}:  
  Estimation factor \(T\), power \(\beta\), and, if the number of trials was minimized, the required minimum number of trials \(n\).
\end{enuma}
\end{protocol}

\begin{protocol}[Spot-Checking Experiment and Analysis]\label{prot:SC_analysis} 

\begin{enuma}
\item \textbf{Input}:
    \begin{itemize} 
    \item Spot-checking probability \(\omega \in (0,1)\);
    \item Error bound \(\epsilon\);
    \item Number of spot-checking trials \(n\) (from Protocol~\ref{prot:EF_construction}); 
    \item Estimation factor \(T\) (from  Protocol~\ref{prot:EF_construction});
    \item Power \(\beta\) (from Protocol~\ref{prot:EF_construction}). 
    \end{itemize}

\item \textbf{Experiment}:  
  Perform \(n\) spot-checking trials. For each trial \(i\), record the spot-checking outcome \(y_i\) 
  and the observed value \((1-y_i)x_i\).  
  Collect \(\bm{y}=(y_i)_{i=1}^n\) and \(\bm{x}_{\obs}=((1-y_i)x_i)_{i=1}^n\).  

\item \textbf{Analysis}:  
  Compute the confidence bound \(S_{\lb}\) according to Eq.~\eqref{eq:conf_lb}.  

\item \textbf{Output}:  
  Confidence bound \(S_{\lb}\) at confidence level \((1-\epsilon)\).
\end{enuma}
\end{protocol}

\section{Other spot-checking scenarios}

The spot-checking protocols developed above can be readily
generalized to other scenarios. In this section, we consider three
such scenarios.  First, we consider the scenario where the
spot-checking probability in each trial can vary adversarially
within specified bounds. Second, we consider the scenario where the
spot-checking RV \(Y\) takes more than two values.  Third, we
discuss early stopping used to ensure, with high probability, 
that the number of unchecked trials equals a predetermined value.

\subsection{Spot-checking with varying probabilities}
\label{sect:EFs_with_bias}

We consider the case where the spot-checking probability may vary from
trial to trial. We assume that the true probability $\tilde\omega$ for
a given trial is unknown but constrained to the interval
$[\omega', \omega]$, where $0 < \omega' \leq \omega < 1$. We define
the trial bias as $\delta = \omega - \tilde\omega$, which implies
$\delta \leq \delta_{\max} \coloneqq \omega - \omega'$. Note that
$\delta$ is a trial-dependent RV that may be chosen adversarially and
depend on past events; for the present purposes, we deviate from our
notational convention by denoting this random variable with a
lower-case letter.  This situation corresponds to the model
\(\cM_{\cF}\) of allowed distributions for a sequence of trials, where
the trial model \(\cF\) consists of all discrete distributions of
\(Y\) and \(X\) for which \(Y\) and \(X\) are independent, the
probability that \(Y=0\) is in the interval \([\omega',\omega]\), and
\(X\) is lower-bounded by \(b\). As usual, we set \(b=0\) without loss
of generality. For the definitions of the models \(\cF\) and
\(\cM_{\cF}\), see the paragraphs preceding
Lem.~\ref{lem:ucond_efineq}.  Since Lem.~\ref{lem:Uequiv} is
applicable to these models, the same estimation factor can be used
whether we estimate \(YX\) or \(Y \Exp(X|\Past)\).

The goal is to construct estimation factors for the model \(\cM_{\cF}\).
According to Lem.~\ref{lem:ucond_efineq}, it
suffices to consider a generic trial with distribution \(\nu\) of
\(X\) and \(Y\) constrained in the trial model \(\cF\) and to
  construct the corresponding estimation factors. We can generalize the
construction in the main text. As a function of \(X\) and \(Y\), we
can write the estimation factor in the form \(T=(1-Y)T'(X)+ Yt\).
By considering distributions in \(\cF\) where \(X\) is deterministic, 
the estimation-factor inequality in Eq.~\eqref{eq:uncond_efineq} 
implies that for all \(x\geq 0\)
\begin{align}
  1 &\geq
      \max_{\tilde\omega\in[\omega',\omega]}
      \qty(1-\tilde\omega) t e^{-\beta x} + \tilde\omega T'(x) \nonumber\\
    & = \max\qty((1-\omega')te^{-\beta x} + \omega' T'(x),
      (1-\omega) t e^{-\beta x} + \omega T'(x)).
\end{align}
This inequality requires that
\begin{align}
  T'(x) &\leq \min\qty(\qty(1-(1-\omega')te^{-\beta x})/\omega',
          \qty(1-(1-\omega)te^{-\beta x})/\omega).
\end{align}
Furthermore, the non-negativity of estimation factors imposes the 
condition  \(0\leq t\leq \min(1/(1-\omega),1/(1-\omega'))=1/(1-\omega')\).
As a consequence, we can generalize the definition of \(F'_{\beta,t}(x)\) 
in the main text by setting
\begin{align}
  F'_{\beta,t}(x)&=\min\qty(\qty(1-(1-\omega')te^{-\beta x})/\omega',
                   \qty(1-(1-\omega)te^{-\beta x})/\omega).
\end{align}
This is the minimum of two concave functions and therefore concave.
As in the main text, in view of the concavity of \(F'_{\beta,t}(x)\), 
we conclude that the non-negative function \(T=(1-Y)F'_{\beta,t}(X)+Y t\) is
an extremal estimation factor for estimating \(YU\) with \(U=X\) or \(U=\Exp(X|\Past)\).
This reduces the task of optimizing estimation factors in
the presence of trial bias to a two-parameter optimization problem.

Instead of direct optimization over extremal estimation factors, one
can modify an estimation factor \(T_{0}\) optimized for
\(\tilde\omega=\omega\), that is, for the trial model
\(\cF_{\omega}\).  To obtain an estimation factor for \(\cF\) from
\(T_{0}\), it suffices to let \(T=T_{0}/f_{\max}\), where \(f_{\max}\)
is an upper bound on the maximum expectation of \(T_{0}\) for
distributions in \(\cF\).  The estimation factor \(T_{0}\) is of the
form
\begin{align}
  T_{0}&=(1-Y)\qty(1-(1-\omega)te^{-\beta X})/\omega + Yt.
\end{align}
To determine a suitable value for \(f_{\max}\), consider a
distribution \(\nu'\) in \(\cF\) with
\(\Prob_{\nu'}(Y=0)=\omega-\delta\) and bound the corresponding
expectation of \(T_{0}\).  Let \(\nu\) be the distribution in
\(\cF_{\omega}\) with the same marginal on \(X\). Then,
\(\Exp_{\nu}(T_{0})\leq 1\) and
\begin{align}
  \Exp_{\nu'}(T_{0})
  &= (\omega-\delta)\qty(1-(1-\omega)t\Exp_{\nu'}\qty(e^{-\beta X}))/\omega
    +(1-\omega+\delta) t
    \notag \\
  &= (\omega-\delta)\qty(1-(1-\omega)t\Exp_{\nu}\qty(e^{-\beta X}))/\omega
    +(1-\omega+\delta) t
    \notag \\
  &=
    \Exp_{\nu}(T_{0})
    -\delta\qty(1-(1-\omega)t\Exp_{\nu}\qty(e^{-\beta X}))/\omega
    +\delta t
    \notag\\
  &\leq
    1-\delta\qty(1-(1-\omega)t\Exp_{\nu}\qty(e^{-\beta b}))/\omega
    +\delta t
    \notag\\
  &=
    1+\delta (t-(1-(1-\omega)t)/\omega)
    \notag\\
  &= 1+\delta(t-1)/\omega.
    \label{eq:ef_vio}
\end{align}
Since \(\delta \in [0, \delta_{\max}]\), we can set \(f_{\max}=\max(1+\delta_{\max}(t-1)/\omega,1)\) and
\(T=T_{0}/f_{\max}\) to obtain an estimation factor for \(\cF\). We expect this estimation factor to
perform well for small bias \(\delta_{\max}\).

\subsection{Spot-checking with more than two choices}
\label{sect:general_spotcheck}

In the entanglement-based BB84 quantum key distribution (QKD) protocol~\cite{BBM:1992}, during
each trial (or round), Alice and Bob independently and randomly
measure their own particles along either the key-generation
$\mathds{X}$-basis or the complementary $\mathds{Z}$-basis. If both
parties choose the $\mathds{X}$-basis, the trial is a key-generation
trial, and a raw key is shared between Alice and Bob. If both use the
$\mathds{Z}$-basis, the trial becomes a parameter-estimation trial,
assisting Alice and Bob in estimating the amount of private
information in their key-generation trials.  In the case that the
basis choices of Alice and Bob do not match in a trial, the standard
QKD implementation and security analysis discard the corresponding
trial, constituting the sifting step in a QKD protocol.  For more
details, see Box 1 in Ref.~\cite{Tomamichel2012}.  This can be 
interpreted as a spot-checking protocol in which the spot-checking 
RV \(Y\) has three possible values. 
Specifically, in the QKD protocol,  \(Y=0\) and the trial is spot-checked if
both parties measure in the \(\mathds{Z}\) basis, \(Y=1\) and the
trial is used for key generation if both parties measure 
in the \(\mathds{X}\) basis, and \(Y=2\) and the trial is useless
otherwise.  

For constructing $\EF$s, we examine a generic trial with RVs $X\geq b=0$ and $Y\in \{0,1,2\}$.  
Since we can observe the RV $X$ only when $Y=0$,  the $\EF$ $T(X,Y)$ 
must take the following form
\begin{equation}
T(X,Y)=
\begin{cases}
T'(X), &  \text{if } Y=0, \\
 t_1,  & \text{if } Y=1, \\
 t_2,  & \text{if } Y=2, 
 \label{eq:extended_ef_form}
\end{cases}
\end{equation}
where the function $T'(X)$, along with the parameters $t_1$ and $t_2$,
is required to be non-negative. Suppose that $Y$ takes values $0, 1$,
and $2$ with probabilities $\omega_0$, $\omega_1$, and $\omega_2$,
respectively, where $\omega_0>0$, $\omega_1>0$, and
$\omega_0+\omega_1+ \omega_2=1$.  We assume that $\omega_0$, 
$\omega_1$, and $\omega_2$ are fixed.  If they
  vary within given bounds, the techniques of
  Sect.~\ref{sect:EFs_with_bias} can be applied to construct estimation
  factors.  Under the free-choice assumption,
the estimation-factor inequality in Eq.~\eqref{eq:uncond_efineq}
for allowed distributions \(\nu\) of \(X\) and \(Y\) simplifies to
\begin{equation} \label{eq:extended_ef_simp}
\omega_1 t_1\Exp_{\nu}(e^{-\beta U}) +\omega_2 t_2+ \omega_0  \Exp_{\nu}\big( T'(X) \big) \leq 1.
\end{equation}
Here, \(U\) can be either \(X\) or \(\Exp(X|\Past)\).  The observations presented 
in the main text still apply, and the extremal $\EF$s must be of the form
\begin{equation}
T(X,Y) =
\begin{cases}
(1-\omega_1 t_1e^{-\beta X} - \omega_2 t_2)/\omega_0, &  \text{if } Y=0, \\ 
 t_1,  & \text{if } Y=1, \\
 t_2,  & \text{if } Y=2. 
 \label{eq:extended_opt_ef_form}
\end{cases}
\end{equation}
To ensure that the estimation factor is non-negative, the parameters $t_1$ and $t_2$ are constrained 
such that $t_1, t_2 \geq 0$ and $1-\omega_1t_1  - \omega_2 t_2\geq 0$. 
When \(Y=2\), the RV \(X\) is not observed nor is the associated $U$
  relevant. This motivates the specific choice \(t_{2}=1\).  With
this choice, the $\EF$s in Eq.~\eqref{eq:extended_opt_ef_form} 
are characterized by two parameters $\beta>0$ and
\(t_{1}\in [0, (1-\omega_{2})/\omega_{1}] = [0,1+\omega_{0}/\omega_{1}]\).
We can rewrite the upper bound  on \(t_{1}\) as
\(1+\omega_{0}/\omega_{1} =1/\qty(1-\omega_0/(\omega_{0}+\omega_{1}))\). 
If we define \(\omega=\omega_0/(\omega_{0}+\omega_{1})\) and \(t=t_{1}\), then
\(T(X,Y)\) restricted to the cases where \(Y=0\) or \(1\) has the same form as the
extremal estimation factors in Eq.~\eqref{eq:t_maximal}. Since the case
\(Y=2\) contributes \(\ln(t_{2})=0\) to the confidence bound,
these $\EF$s can be constructed using the techniques 
presented in Sects.~\ref{sect:numerical_opt} and~\ref{sect:analytical_opt}.

\subsection{Spot-checking with early stopping}
\label{sect:early_stop}

In the main text, we showed that it is possible to stop the
experiment early once a desired number \(m\) of unchecked trials
has been obtained.  Formally, this is achieved by setting the 
spot-checking probabilities and estimation factors for all future trials to \(1\),
and choosing the maximum number of trials \(n\) to be performed to be
sufficiently large. For this purpose, it is in principle possible
to choose \(n=\infty\), in which case the experiment stops with exactly 
\(m\) unchecked trials with probability \(1\).  In practice, however, there is always
an effective upper limit on the maximum number of trials \(n\) that can be
executed. In other words, at some point the experimenter gives up
and stops the experiment, declaring failure if fewer than \(m\) unchecked trials have been observed.
When the spot-checking probability in each executed trial is 
\(\omega\), the probability of failure can be made arbitrarily small
by choosing \(n\) sufficiently large compared to
\(m/(1-\omega)\).  For example, to ensure that the probability of
failure is at most \(e^{-\gamma}\) for some large \(\gamma\), we
can use the Hoeffding inequality~\cite{hoeffding1963} to show that
it suffices to choose \(n\) large enough such that
\(e^{-\gamma} \geq \exp(-2(n-n\omega-m)^2/n)\). Solving for \(n\),
for spot-checking with early stopping, we set 
\begin{equation} \label{eq:n_lb_earlystop}
  n\geq \frac{m}{(1-\omega)} \times \qty(1+\frac{\gamma+\sqrt{\gamma^2+8m(1-\omega)\gamma}}{4m(1-\omega)}).
\end{equation}   

In the presence of possible failures, the coverage property of
confidence bounds \(S_{\lb}\), as described in
Sect.~\ref{sect:EF_method}, needs to be modified.  Let
$C_n=\sum_{i=1}^{n} Y_i$ be the number of unchecked trials after
executing \(n\) trials. Then the early stopping protocol as
implemented above succeeds if \(C_{n}\geq m\).  By
Thm.~\ref{thm:ef_confbound} and the fact that
\(\Prob(C_{n}\geq m, S_{U}< S_{\lb})\leq \Prob(S_{U}< S_{\lb})\), we
have that the probability of succeeding but not covering satisfies
\(\Prob(C_{n}\geq m, S_{U}< S_{\lb})\leq \epsilon\).  In general, this
does not directly imply the coverage property conditional on
success. This is because even moderate probabilities of succeeding can
often be adversarially exploited to force worse conditional coverage.
However,  for spot-checking with early stopping, the coverage probability 
conditional on success satisfies
  \begin{align}
    \Prob(S_{U}\geq S_{\lb}|C_{n}\geq m)
    &=1-\Prob(S_{U}< S_{\lb}|C_{n}\geq m)
      \notag\\
    &=1-\Prob(S_{U}< S_{\lb},C_{n}\geq m)/\Prob(C_{n}\geq m)
      \notag\\
    &\geq 1-\Prob(S_{U}<S_{\lb},C_{n}\geq m)/(1-e^{-\gamma})
      \notag\\
    &\geq 1-\Prob(S_{U}<S_{\lb})/(1-e^{-\gamma})
      \notag\\
     &\geq 1-\epsilon/(1-e^{-\gamma}).
  \end{align}
  Spot-checking with early stopping therefore retains good coverage
  probability conditional on success for moderately large \(\gamma\),
  with little overhead in the required number of trials \(n\).

\section{Comparison to other protocols}
\label{sect:comparison}

In this section, we review two existing approaches to spot-checking
and compare them with the estimation-factor method. First, in
Sect.~\ref{sect:Gocanin2022}, we describe and clarify the
spot-checking approach introduced in Ref.~\cite{Gocanin2022}.  Next,
in Sect.~\ref{sect:serfling}, we discuss an alternative approach based
on the Serfling inequality. Finally, in Sect.~\ref{sect:comps}, we
provide the details of the comparison among  these methods,
as summarized in the two figures presented in the main text.
Additional comparison results are given in Sect.~\ref{sect:comps}. 

\subsection{Spot-checking according to Ref.~\cite{Gocanin2022}}
\label{sect:Gocanin2022}

In the scenario investigated in Ref.~\cite{Gocanin2022}, each trial
\(i\) involves a binary 0/1-valued RV \(X_{i}\) and a spot-checking RV
\(Y_{i}\) as considered here, satisfying the free-choice
assumption. Below we slightly generalize the scenario by allowing \(X_{i}\)
to take values in an arbitrary two-element set
\(\{x_{\lb},x_{\ub}\}\subset \rls \) with \(x_{\lb}<x_{\ub}\).  The
results of Ref.~\cite{Gocanin2022} assume that the trials are
independent and that the spot-checking probability is fixed to be
\(\omega\) for each trial.  Moreover, in each trial \(i\), the mean of
\(X_{i}\) is assumed to be a fixed but unknown parameter
\(\theta_{i}\). Slightly more generally, the \(\theta_{i}\) can be
random but must satisfy that they are determined by \(\Past_{1}\), the
past of the entire experiment.  We remark that Ref.~\cite{Gocanin2022}
claims more generality, but the independence requirement in the
paragraph containing Eqs. (24) and (25) therein implies the
independence stated above.

To describe the main result of Ref.~\cite{Gocanin2022}, we take the
means \(\theta_{i}\) as fixed but unknown.  The goal is to obtain a
lower confidence bound on \(S_{n}=\sum_{i=1}^{n}Y_{i}\theta_{i}\)
based on the observed values of \(\sum_{i=1}^{n}(1-Y_{i})X_{i}\) and
\(\sum_{i=1}^{n} Y_i\). For this purpose, Ref.~\cite{Gocanin2022}
constructs tests of the null hypotheses
\(Q_{n}=\sum_{i=1}^{n} \theta_{i} \leq
n\bar{\theta}_{\nll}\vphantom{\Big|}\) for given values of
  \(\bar\theta_{\nll}\).    Inverting these tests yields lower
confidence bounds for \(Q_{n}\), which are then converted to lower
confidence bounds for \(S_{n}\). Below, we detail each step of this
approach and justify the reasoning.  Unlike the confidence bounds
obtained with estimation factors, the confidence bounds on \(S_{n}\)
obtained this way do not hold for
\(S'_{n}=\sum_{i=1}^{n}Y_{i}X_{i}\), even under the assumption of
  independent trials.

Consider values \(\bar \theta_{\nll}\) and \( \bar{\theta}_{\obs}\)
satisfying
\(x_{\lb}\leq \bar\theta_{\nll}< \bar{\theta}_{\obs}\leq x_{\ub}\).
Define \(p_{\obs}=(\bar{\theta}_{\obs}-x_{\lb})/(x_{\ub}-x_{\lb})\)
and
\(p_{\nll}=(\bar{\theta}_{\nll}-x_{\lb})/(x_{\ub}-x_{\lb})\). These
are the probabilities of observing \(X_{i}=x_{\ub}\) if the mean of
\(X_{i}\) is \(\bar{\theta}_{\obs}\) and \(\bar{\theta}_{\nll}\),
respectively.  Under the hypothesis that
\(Q_{n} \leq n \bar{\theta}_{\nll}\), according to Eq.~(26) of
Ref.~\cite{Gocanin2022} we have
\begin{align} \Prob\qty(\sum_{i=1}^{n} (1-Y_i) X_i \geq
  \sum_{i=1}^{n} (1-Y_i) \bar{\theta}_{\obs}) \leq \qty(1-\omega+\omega
  e^{-D_{\KL}(p_{\obs}|p_{\nll})})^n, \label{eq:prxq_main}
\end{align} where 
\(D_{\KL}(p_{\obs}|p_{\nll})=p_{\obs}
\ln(p_{\obs}/p_{\nll})+(1-p_{\obs})\ln\big((1-p_{\obs})/(1-p_{\nll})\big)\)
is the Kullback–Leibler (KL) divergence. 
For \(\bar{\theta}_{\obs}>x_{\ub}\), the probability on the left-hand 
side of the above inequality  is \(0\).
The derivation of this probability bound in
Ref.~\cite{Gocanin2022} requires the assumption of independent
trials.

From Eq.~\eqref{eq:prxq_main}, given \(\bar{\theta}_{\nll} <x_{\ub}\),
we derive a hypothesis test \(H(\bar{\theta}_{\nll},\epsilon)\) for
the null hypothesis \(Q_{n}\leq n\bar{\theta}_{\nll}\) at significance
level \(\epsilon \in \qty(0,1)\).  We choose a threshold
\(\bar{\theta}_{\obs}\) depending on \(\bar{\theta}_{\nll}\) such that
the test rejects the null hypothesis if
\(\sum_{i=1}^{n}(1-Y_{i})X_{i} >
\sum_{i=1}^{n}(1-Y_{i})\bar{\theta}_{\obs}\).  For this, we aim to
solve the equation
\(\epsilon=P(\bar{\theta}_{\obs}) \coloneqq \qty(1-\omega+\omega
e^{-D_{\KL}(p_{\obs}|p_{\nll})})^n\) for \(\bar{\theta}_{\obs}\) with
\(\bar{\theta}_{\obs}>\bar\theta_{\nll}\), provided a solution exists.
We observe that \(D_{\KL}(p_{\obs}|p_{\nll})\) is a continuous and
monotonically increasing function of \(p_{\obs}\) when
\(p_{\obs}\geq p_{\nll}\).  Specifically, we have
\(D_{\KL}(p_{\nll}|p_{\nll})=0\), so that 
\(P(\bar{\theta}_{\nll})=1\), and
\(D_{\KL}(1|p_{\nll}) = -\ln(p_{\nll})\), so that 
  \(P(x_{\ub})=(1-\omega+\omega p_{\nll})^n\).  In addition, we have
that both \(p_{\obs}\) and \(p_{\nll}\) are linear with positive slope
in \(\bar{\theta}_{\obs}\) and \(\bar{\theta}_{\nll}\), respectively,
which implies that \(P(\bar{\theta}_{\obs})\) is monotonically
decreasing.  Thus, if \(\epsilon\geq P(x_{\ub})\), a unique solution
for \(\bar{\theta}_{\obs}\) exists.  Conversely, if
\(\epsilon < P(x_{\ub})\), we set \(\bar\theta_{\obs}=x_{\ub}\). In
this case, the null hypothesis is never rejected.  Moreover,
since \(D_{\KL}(p_{\obs}|p_{\nll})\) is monotonically decreasing 
and continuous  in \(p_{\nll}\in(0, p_{\obs}] \) for fixed \(p_{\obs}\), the
relationship between \(\bar{\theta}_{\nll}\) and the corresponding
threshold \(\bar{\theta}_{\obs}\) in this construction is monotonic
and continuous; that is, \(\bar{\theta}_{\obs}\) is a non-decreasing
and continuous function of \(\bar{\theta}_{\nll}\).

The hypothesis test described in the previous paragraph can be
inverted to obtain a lower confidence bound on \(Q_{n}\) according to
the standard hypothesis-test inversion procedure (see Sect.~7.1.2 of
the textbook~\cite{Shao2003}). Given the observed values of the
RVs and a significance level \(\epsilon\), the lower
confidence bound is given by the infimum \(n\theta_{\lb}\) over all
values \(n\bar{\theta}_{\nll}\) for which the hypothesis test
\(H(\bar{\theta}_{\nll},\epsilon)\) does not reject the null
hypothesis \(Q_{n}\leq n\bar{\theta}_{\nll}\). The confidence
bound \(n\theta_{\lb}\) is an RV through its dependence on the
instantiated RVs. Because of the monotonic and continuous relationship
mentioned at the end of the previous paragraph, \(\theta_{\lb}\) can
be obtained as follows: First, we set
\(\bar{\theta}_{\obs}= \sum_{i=1}^{n}(1-y_{i})
x_{i}/\sum_{i=1}^{n}(1-y_{i})\) if \(\sum_{i}(1-y_{i})\neq 0\), and
set \(\bar{\theta}_{\obs}=x_{\lb}\) otherwise.  We then set the
  value of \(\theta_{\lb}\in [x_{\lb},\bar{\theta}_{\obs}]\) as
  follows:  If there is a value of \(\bar\theta_{\nll}\) that
  solves the equation
  \(\epsilon=\qty(1-\omega+\omega
  e^{-D_{\KL}(p_{\obs}|p_{\nll})})^n\), we set \(\theta_{\lb}\) to
  this value.  A solution exists if \(\bar{\theta}_{\obs}>x_{\lb}\) and
\(\epsilon\geq(1-\omega)^{n}\).  Otherwise, we set 
  \(\theta_{\lb}=x_{\lb}\).  To solve the equation in the
  first case, we rewrite it as 
\begin{align}
  -\ln(1-\omega\qty(1-e^{-D_{\KL}(p_{\obs}|p_{\nll})})) &= \ln(1/\epsilon)/n.
                                                          \label{eq:goc_solve}
\end{align}
In this identity, \(p_{\obs}\) depends 
linearly on \(\bar\theta_{\obs}\), and we solve for \(p_{\nll}\), from
which the sought-for solution for \(\bar{\theta}_{\nll}\) can be
recovered via a linear transformation. With this
strategy, \(\theta_{\lb}\) satisfies 
\begin{equation} \label{eq:thetalb_sol}
  \theta_{\lb}=f(p_{\obs})(x_{\ub}-x_{\lb}) + x_{\lb},
\end{equation}
where \(f(p) \coloneqq \inf\qty{q | p \geq q \geq 0,  (1-\omega+\omega e^{-D_{\KL}(p|q)})^{n}\geq\epsilon}\).

The confidence bound \(n\theta_{\lb}\) on \(Q_{n}\) can be converted
to a lower confidence bound on \(S_{n}\).  To obtain a tighter lower
confidence bound on \(S_{n}\), we additionally consider  a known 
upper bound \(\theta_{\max}\) on all \(\theta_{i}\), where
\(x_{\lb}\leq \theta_{\max}\leq x_{\ub}\).  If
\(\theta_{\max}=x_{\ub}\), this upper bound adds no new 
information and becomes trivial.  A
non-trivial such upper bound is used in the comparison in
Sect.~\ref{sect:comps}. We can now apply the inequality
\begin{align}
  Q_{n}&=
         \sum_{i=1}^{n} \theta_i\notag\\
       &\leq \sum_{i=1}^{n} Y_i\theta_i + \sum_{i=1}^{n} (1-Y_i) \theta_{\max} \notag \\
       & =S_{n}+ (n-C_{n})\theta_{\max}, \label{eq:prxq_bound}
\end{align}
where \(C_n=\sum_{i=1}^{n} Y_i\).  This inequality corresponds to 
Eqs. (28) and (29) of Ref.~\cite{Gocanin2022}.  As a consequence,
\begin{align}
  S_{\lb,\rm{G}}&=n\theta_{\lb}- (n-C_{n})\theta_{\max}
                  \label{eq:goc_lb}
\end{align}
is a lower confidence bound on \(S_{n}\) with confidence level
\((1-\epsilon)\).  Since \(C_{n}\) is observed, the confidence
bound \(S_{\lb,\rm{G}}\) is fully determined by the experimental
observations.  The inequality in Eq.~\eqref{eq:prxq_bound}
suggests that the performance of this confidence bound depends on
the value of the spot-checking probability \(\omega\): a higher
\(\omega\) implies a typically larger gap between \(n\) and \(C_{n}\), leading
to a more conservative bound. Conversely, when \(\omega\) is close
to \(0\), \(C_{n}\) is typically close to \(n\), tightening the bound.

To evaluate the finite-data efficiency of the method in
Ref.~\cite{Gocanin2022}, it is necessary to estimate the expectation
of \(S_{\lb,\rm{G}}\), which is given by
\(\Exp(S_{\lb,\rm{G}}) = n\qty(\Exp(\theta_{\lb})-
\omega\theta_{\max})\).  We now treat variables such as
\(\theta_{\lb}\), \(\bar\theta_{\obs}\) and \(p_{\obs}\) defined in
the previous paragraphs as random variables depending on the
experimental observations.  For comparison purposes, we consider the
case where \(\theta_{i}=\theta\) for all \(i\). For binary RVs
\(X_i\), this implies that the trials are not just independent but
also identical.  Let
\(p_{\theta}=(\theta-x_{\lb})/(x_{\ub}-x_{\lb})\).  Computing
\(\Exp(\theta_{\lb})\) is complicated by the fact that
\(\theta_{\lb}\) is a non-linear function of the observed quantity
\(p_{\obs}\) (or equivalently, of \(\bar{\theta}_{\obs}\)),
parametrized by \(\ln(1/\epsilon)\) and \(n\). For the present
purposes, it suffices to determine \(\Exp(\theta_{\lb})\) to within
order \(O(1/n)\).  We claim that to this order, \(\Exp(\theta_{\lb})\)
can be obtained by replacing \(\bar{\theta}_{\obs}\) and \(p_{\obs}\)
with \(\theta\) and \(p_{\theta}\) in the computation of
\(\theta_{\lb}\).  Except for the bias resulting from the case where
no trials are spot-checked, \(\theta\) and \(p_{\theta}\) are the
expectations of \(\bar\theta_{\obs}\) and \(p_{\obs}\),
respectively.  We first
consider the boundary cases.  If \(\epsilon<(1-\omega)^{n}\), then
regardless of values of \(\bar{\theta}_{\obs}\) and \(p_{\obs}\), we
have \(\theta_{\lb}=x_{\lb}\), and the claim is satisfied.  Hence,
  we assume from now on that \(\epsilon\geq (1-\omega)^{n}\). If
  \(\theta=x_{\lb}\), then \(p_{\obs}=0\) and \(\bar\theta_{\obs}=x_{\lb}\) 
  with probability \(1\), and therefore
\(\Exp(\theta_{\lb})=x_{\lb}\), consistent with the claim.
If \(\theta=x_{\ub}\), then with probability
\(1-(1-\omega)^{n}=1-e^{-O(n)}\), at least one trial is
  spot-checked and so \(\bar\theta_{\obs}=x_{\ub}\), which makes
\(p_{\obs}= p_{\theta}=1\), again consistent with the
claim.  For the case where \(x_{\lb}<\theta<x_{\ub}\) and
\(\epsilon\geq (1-\omega)^{n}\), the probability that
\(\bar\theta_{\obs}=x_{\lb}\) is bounded by the sum of the probability
\( (1-\omega)^n\) that no trial is spot-checked and
the probability that the outcomes of all spot-checked trials are
\(x_{\lb}\). This probability is \(\sum_{i=0}^{n}\binom{n}{i} (1-p_{\theta})^{i}\omega^{i}(1-\omega)^{n-i}=(1-\omega p_{\theta})^{n}\).
Since in this case \(\omega p_{\theta}<1\), the probability that
\(\bar\theta_{\obs}=x_{\lb}\) is exponentially small in \(n\) and
contributes negligibly to the expectation. We therefore neglect this
contribution for the remainder of the argument when calculating
expectations and variances.   To complete the proof of the claim, in
view of Eq.~\eqref{eq:thetalb_sol}, it suffices to prove that
\begin{equation}\label{eq:fp_exp}
  \Exp(f(p_{\obs})) = f(p_{\theta}) + O(1/n)
\end{equation}
for \(p_{\theta}\in (0,1)\).  If \(p>0\), then \(f(p)\) is the unique
solution \(q\) to the equation
\(\epsilon=(1-\omega+\omega e^{-D_{\KL}(p|q)})^{n}\) in the interval
\(0\leq q\leq p\).  This means that the graph of \(f\) consisting of
the points \((p,f(p))\) is a subset of the level curve of
\(D_{\KL}(p|q)\) consisting of the points \((p,q)\) where
\(D_{\KL}(p|q)=d\) with
\(d=-\ln((\epsilon^{1/n}-(1-\omega))/\omega)\).  The function
\(D_{\KL}(p|q)\) is jointly convex in its two arguments. It follows
that the region \(\{(p,q) \mid D_{\KL}(p|q) \leq d \}\) is a convex
set in the \((p,q)\)-plane. Since \(D_{\KL}(p|q)=0\) iff \(q=p\), this
convex set contains the diagonal line \(q=p\). The graph of \(f\)
forms the lower boundary of this set, which is below the diagonal
line.  Geometrically, the lower boundary of a convex set defines a
convex function.  Thus \(f(p)\) is convex.  Moreover, since
\(D_{\KL}(p|q)\) is smooth for \(0<p,q<1\) and by the implicit
  function theorem, \(f(p)\) is continuously second differentiable
for \(0<p<1\) and continuous at the boundaries. We next exploit the approximation
\(D_{\KL}(p|p-\delta)=\frac{\delta^{2}}{2p(1-p)}\qty(1 + O(\delta))\),
where the suppressed constant depends on \(p\) and can be estimated by
means of Eq.~(85) of Ref.~\cite{wills2020performance}.  Let
\(\delta(p)=p-f(p)\). Then \(\delta(p)\) is concave.
Eq.~\eqref{eq:fp_exp} is equivalent to the claim that
\(\Exp(\delta(p_{\obs})) = \delta(p_{\theta}) + O(1/n)\).  From the
given approximation of the KL divergence and the identity
\begin{align}
  D_{\KL}(p|p-\delta(p))
  &=-\ln((\epsilon^{1/n}-(1-\omega))/\omega)
    \notag\\
  &=-\ln(1-(1-\epsilon^{1/n})/\omega)
    \notag\\
  &= -\ln(1-(1-e^{\ln(\epsilon)/n})/\omega)
    \notag\\
  &= -\ln(1+\ln(\epsilon)/(n\omega)(1+O(1/n)))
    \notag\\
  &=\ln(1/\epsilon)/(n\omega)(1 + O(1/n)),
\end{align}
we can deduce that
\(\delta(p)^2 =
\frac{2p(1-p)\ln(1/\epsilon)}{n\omega}(1+O(1/n^{1/2}))\). Consequently
\(\delta(p_{\theta})\) is proportional to \(1/\sqrt{n}\) to lowest
order. Finally, we show that the difference between
\(\Exp(\delta(p_{\obs}))\) and \(\delta(p_{\theta})\) is at the next
order, \(O(1/n)\). By concavity of \(\delta(p)\),
\(\Exp(\delta(p_{\obs})) \leq \delta(p_{\theta})\).  Since
\(\delta(p)\) is continuously second differentiable, concave, and
non-negative, there exists a positive constant \(c\) such that
\(\delta(p_{\theta}+u)\geq \delta(p_{\theta})+ \delta'(p_{\theta})u -c
u^{2}\) for \(p_{\theta}+u\in [0,1]\).  Substituting \(p_{\obs}-p_{\theta}\)
 for \(u\) and taking the expectation of both sides gives
\(\Exp(\delta(p_{\obs}))\geq \delta(p_{\theta}) - c\Var(p_{\obs})\).
We calculate \(\Var(p_{\obs})\) by applying the law of total variance with respect
to the number of spotchecked trials \(I=n-C_n\):
  \begin{align}
  \Var(p_{\obs})
  &= \sum_{i=0}^{n}\Prob(I=i)\Var(p_{\obs}|I=i)
         +\Var(\Exp(p_{\obs}|I)).
  \end{align}
  For \(i>0\), \(\Exp(p_{\obs}|I=i)=p_{\theta}\), and for \(i=0\),
  \(\Exp(p_{\obs}|I=0)=0\).  Because \(\Prob(I=0)\) is exponentially
  small, so is \(\Var(\Exp(p_{\obs}|I))\). We have
  \(\Var(p_{\obs}|I=0)=0\) and for \(i>0\),
  \(\Var(p_{\obs}|I=i) = p_{\theta}(1-p_{\theta})/i\), because
  \(p_{\obs}\) is the sample mean of \(i\) independent Bernoulli trials with
  mean \(p_{\theta}\) given each set of \(i\) spot-checked trials. Since
  \(\Prob(I=i)= \binom{n}{i}\omega^{i}(1-\omega)^{n-i}\), we obtain,
  up to an exponentially small correction
  \begin{align}
  \Var(p_{\obs})
  &= \sum_{i=1}^{n} \omega^{i} (1-\omega)^{n-i} \binom{n}{i}p_{\theta}(1-p_{\theta})/i
    \notag\\
  &\leq
    \frac{2}{\omega (n+1)}
    p_{\theta}(1-p_{\theta})\sum_{i=1}^{n+1}\omega^{i+1} (1-\omega)^{n-i} \binom{n+1}{i+1}
    \notag\\
  &\leq \frac{2}{\omega (n+1)}
    p_{\theta}(1-p_{\theta}).
\end{align}
The claim follows.  

As with estimation-factor-based confidence bounds, one can determine
the minimum number of i.i.d. trials required to achieve, on average, a
desired gap between \(S_n\) and \(S_{\lb,\rm{G}}\), given a target
distribution that governs each trial and a specific significance level
\(\epsilon\).  Here we describe, justify, and generalize the solution given in
 Ref.~\cite{Gocanin2022} by extending the arguments leading to Eq.~(27) 
 in that reference. Let \(\theta\) be the mean of \(X\) in a generic
trial according to the target distribution, and let
\(n(1-\omega)\delta_{\thresh}\) be the desired gap.  The expected gap
between \(S_n\) and \(S_{\lb,\rm{G}}\) can be computed to good
approximation as follows.  First, when solving for
  \(\theta_{\lb}\), we replace \(\bar\theta_{\obs}\) and \(p_{\obs}\)
  with \(\theta\) and \(p_{\theta}\), which are their expectations up
  to exponentially small neglected terms. This yields an
approximation of \(\Exp(\theta_{\lb})\) accurate to order \(O(1/n)\),
as justified above.  Specifically, according to
Eq.~\eqref{eq:thetalb_sol}, we set
\(\Exp(\theta_{\lb})=f(p_{\theta})(x_{\ub}-x_{\lb}) + x_{\lb}\).
Second, in view of Eq.~\eqref{eq:goc_lb} and the fact that
\(\Exp(C_n)=n(1-\omega)\), we get the expected lower confidence bound
\begin{equation}
  \Exp(S_{\lb,\rm{G}})=
  n\Exp(\theta_{\lb})-n\omega\theta_{\max}.
\end{equation}
Consequently, the expected gap is 
\begin{align}
  \Exp(S_n-S_{\lb,\rm{G}})&= \Exp(S_n)- \Exp(S_{\lb,\rm{G}})\notag \\
                          &= n(1-\omega)\theta-n\Exp(\theta_{\lb})+n\omega\theta_{\max}.
\end{align}
To ensure that the expected gap is less than or equal to the desired
gap, we require
\(\Exp(\theta_{\lb})\geq
(1-\omega)(\theta-\delta_{\thresh})+\omega\theta_{\max}\).
In terms of \(f(p_{\theta})\), the inequality becomes
  \begin{align}
    f(p_{\theta}) \geq ((1-\omega)(\theta-\delta_{\thresh})+\omega\theta_{\max} - x_{\lb})/(x_{\ub}-x_{\lb}).
  \end{align}
  The quantity \(f(p_{\theta})\) depends implicitly on \(n\).
  The minimum number of trials \(n\) required is determined by finding the
  minimum \(n\) for which \(f(p_{\theta})\) satisfies this inequality.

\subsection{The Serfling inequality}
\label{sect:serfling}

One approach to spot-checking is to randomly select, without
replacement, a \(k\)-element index subset
\(J\subseteq \{1,\ldots, n\}\) of the \(n\) trials for checking. The
average of the unobserved \(X_{i}\) for \(i\notin J\) can then be
estimated by means of the Serfling inequality~\cite{Serfling1974},
provided that the RVs \(X_{i}\) are bounded. Applying the inequality
requires assuming that the trials are as-if parallel. 
That is, it requires that the index subset \(J\) to be selected is conditionally
independent of the sequence of RVs \(\bm{X}=(X_{i})_{i=1}^{n}\) given
the pre-experiment past \(\Past_{1}\). Under this assumption, we show
below that one can obtain a lower confidence bound at significance level
\(\epsilon\) on \(S'_{n}=\sum_{i\notin J}X_{i}\).  However, with
random sampling without replacement, it is not possible to replace
\(X_{i}\) with \(\Exp(X_{i}|\Past_{i})\) and estimate the
corresponding sum.

To highlight the relationship to bounds previously obtained, we define
the spot-checking RVs as \(Y_{i}=\knuth{i\notin J}\) and set
\(\omega=k/n\). Suppose that for each trial \(i\) the RV \(X_i\) is
bounded with \(X_{i}\in[x_{\lb},x_{\ub}]\).  Then we have the bound
\begin{align}
  \Prob\qty(
  (1-\omega)\sum_{i=1}^{n} (1-Y_i) x_i \geq   \omega\sum_{i=1}^{n} Y_{i}x_{i}
  + \delta
  \Bigg|\bm{X}=\bm{x})
  &\leq e^{-2 \delta^{2}/\qty(n\omega (x_{\ub}-x_{\lb})^{2}(1-\omega+1/n))}.
    \label{eq:serfling1}
\end{align}
To prove the bound, we add \(\omega\sum_{i=1}^{n}(1-Y_{i})x_{i}\)
to both sides of the inequality in the event, thus rewriting the 
left-hand side as
\begin{align}
  \Prob\qty(
  \sum_{i=1}^{n} (1-Y_i) x_i \geq   \omega\sum_{i=1}^{n} x_{i}
  + \delta
  \Bigg|\bm{X}=\bm{x}).
\end{align}
The bound then follows from the Serfling
inequality~\cite{Serfling1974}, specifically Cor. 1.1 in the reference, 
with \(N\), \(n\), and \(t\) there replaced here by \(n\),  \(k=n\omega \), 
and \(\delta/k\), respectively. Since the bound on the right-hand side
  of Eq.~\eqref{eq:serfling1} does not depend on the specific conditioning 
  event \(\bm{X}=\bm{x}\), we can extend it to an arbitrary sequence of 
  RVs \(\bm{X}\), yielding
\begin{align}
  \Prob\qty(
  (1-\omega)\sum_{i=1}^{n} (1-Y_i) X_i \geq   \omega\sum_{i=1}^{n} Y_{i}X_{i}
  + \delta)
  &\leq e^{-2 \delta^{2}/\qty(n\omega (x_{\ub}-x_{\lb})^{2}(1-\omega+1/n))}.
    \label{eq:serfling2}
\end{align}  
This is similar to the bound obtained in Lem.~6 of
Ref.~\cite{Tomamichel2017}, but it applies without restricting the
\(X_i\) to be binary RVs.

Setting the right-hand side of Eq.~\eqref{eq:serfling2} equal to the
error bound \(\epsilon\), we solve for \(\delta\) to obtain
\begin{align}
  \delta&
          =\sqrt{\frac{n\omega(1-\omega+1/n)\ln(1/\epsilon)}{2}}(x_{\ub}-x_{\lb}).
\end{align}
As a result,  a lower confidence bound with confidence level \((1-\epsilon)\) for
\(\sum_{i}Y_{i}X_{i}\) is 
\begin{align}
  S_{\lb,\text{Serf}}
  & = \frac{1-\omega}{\omega}\sum_{i}(1-Y_{i}) X_{i}
    - \sqrt{\frac{n(1-\omega+1/n)\ln(1/\epsilon)}{2\omega }}(x_{\ub}-x_{\lb}).
    \label{eq:serflinglb}
\end{align}

For the case where the \(X_{i}\) are i.i.d., the expectation of
\(S_{\lb,\text{Serf}}\) satisfies
\begin{align}
  \Exp(S_{\lb,\text{Serf}})
  & = \frac{1-\omega}{\omega}\sum_{i}\Exp \qty((1-Y_{i}) X_{i})
    - \sqrt{\frac{n(1-\omega+1/n)\ln(1/\epsilon)}{2\omega }}(x_{\ub}-x_{\lb}) \notag \\
  &=  \frac{1-\omega}{\omega}\sum_{i}\Exp(1-Y_{i}) \Exp(X_{i})
    - \sqrt{\frac{n(1-\omega+1/n)\ln(1/\epsilon)}{2\omega }}(x_{\ub}-x_{\lb}) \notag \\
  &=  \Exp(S'_{n})
    - \sqrt{\frac{n(1-\omega+1/n)\ln(1/\epsilon)}{2\omega }}(x_{\ub}-x_{\lb}),     
    \label{eq:exp_serflinglb}
\end{align}  
where the second line follows from the independence between \(X_i\)
and \(Y_i\) given the past \(\Past_{1}\).  Therefore, the expected gap
between \(S'_{n}\) and \(S_{\lb,\text{Serf}}\) is precisely the
subtracted term on the right-hand side of
Eq.~\eqref{eq:exp_serflinglb}.  This gap may be compared to the one
given in Eq.~\eqref{eq:expgap} using estimation factors.  The latter
is smaller by a factor of \(2\sigma/(x_{\ub}-x_{\lb})\)
when \(n\) is large.  This factor is never greater than \(1\)
in view of the variance bound \(\sigma^2\leq (x_{\ub}-x_{\lb})^2/4\), and it achieves the
maximum value of \(1\) only when each \(X_i\) is binary and takes values
\(x_{\ub}\) and \(x_{\lb}\) with equal probability.  Moreover, to
ensure that the expected gap between \(S'_{n}\) and
\(S_{\lb,\text{Serf}}\) does not exceed a desired threshold
\(n(1-\omega)\delta_{\thresh}\), it suffices to require that the subtracted
term on the right-hand side of Eq.~\eqref{eq:exp_serflinglb} be at
most \(n(1-\omega)\delta_{\thresh}\). Accordingly, one can determine
the minimum number of i.i.d. trials required.

\subsection{ Explanations and calculations for the main-text figures and additional comparisons}
\label{sect:comps}

For the figures presented in the main text, we consider the situation
where two parties receive sequentially generated pairs of quantum
systems. The parties, \(\pA\) and \(\pB\), aim to estimate the
average, device-independent Bell-state extractability of the unchecked
pairs, where they randomly spot-check the
pairs by performing Bell tests on the pairs being
  spot-checked. The spot-checking RV is privately shared between the
  parties ahead of time.  Since Bell tests are
  destructive, performing these tests only on the spot-checked trials
  allows the parties to save the states from the remaining trials while
  still estimating their average extractability.  To this end, the parties
estimate the average violation of a Bell inequality in the unchecked
pairs, for a Bell inequality whose violation can be related to
extractability. For the figures, we specifically consider the Bell inequality
constructed by Clauser, Horne, Shimony and Holt in
1969~\cite{chsh:1969},  also known as the CHSH inequality.  To estimate
the violation of this inequality,  in every spot-checked trial, each
party randomly and independently selects a measurement setting
and applies it to their own system,
obtaining one of two possible outcomes.  Let \(\pC\) denote either
party.  For a generic trial, we use \(S_{\pC}\) to denote the setting
choice of party \(\pC\), where the choice is either \(S_{\pC}=1\) or
\(S_{\pC}=2\). For the figures, each setting is selected with equal
probability, that is, \(\Prob(S_{\pC}=1)=1/2\).  The party's outcome
is denoted by \(O_{\pC}\) and has values in \(\{-1,1\}\).  The
complete record of the trial is given by the RV
\(R=(O_{\pA},O_{\pB},S_{\pA},S_{\pB})\).  The correlation between
outcomes given the settings \(S_{\pA}=s_{\pA}\) and
\(S_{\pB}=s_{\pB}\) is
\(\text{Cor}(s_{\pA}, s_{\pB}) =
\Exp(O_{\pA}O_{\pB}|s_{\pA}s_{\pB})\).  The CHSH inequality can be
expressed in terms of the correlations as
\begin{equation} 
\text{Cor}(1, 1)+\text{Cor}(1, 2)+\text{Cor}(2, 1)-\text{Cor}(2, 2)\leq2. \label{eq:CHSH_ineq}
\end{equation}
 The inequality is satisfied by all local realistic distributions of \(R\).

The CHSH inequality as written above does not depend on the joint
distribution of measurement settings and holds for any distribution
with full support. However, for spot-checking and hypothesis 
testing, the specific distribution over settings matters. Given the 
trial record \(R=(O_{\pA},O_{\pB}, S_{\pA},S_{\pB})\), we define the function 
$I(R)=4(1-2\knuth{S_\text{A}=2} \knuth{S_\text{B}=2}) O_\text{A} O_\text{B}$.
As a function of the RV $R$ and for the uniform distribution of measurement settings, 
the CHSH inequality in Eq.~\eqref{eq:CHSH_ineq} can be written as 
\begin{equation}
\hat I \leq 2, \label{eq:proper_CHSH_ineq}
\end{equation}
where we defined \(\hat I=\Exp\qty(I(R))\). We refer to \(\hat I\) as
the CHSH value, where \(\hat{I}>2\) implies a violation of the CHSH inequality. 
The function $I(R)$ is a RV that takes only two possible values, $-4$ or $+4$, in each trial~\cite{Zhang2011}.

The value of \(\hat I\) for a given trial can be related to
extractability of an ideal Bell state from the pair of systems
shared in the trial. Bell-state extractability \(\Xi\) is defined as
the maximum fidelity with respect to a Bell state that can be
obtained by local isometries applied to each
system~\cite{supic:2020,kaniewski:2016}.  According to
Eq. (10) of Ref.~\cite{kaniewski:2016}, a lower bound on \(\Xi\)
based on \(\hat I\) is given by \(\check\Xi\) defined as
\begin{equation}
  \check \Xi  =
  \frac{1}{2}+\frac{\max(\hat{I}-I_{\thresh},0)}{2(2\sqrt{2}-I_{\thresh})} \label{eq:extract_bound}
\end{equation}
with $I_{\thresh}=(16+14\sqrt{2})/17\approx 2.1058$.  Here we
have extended Eq. (10) of Ref.~\cite{kaniewski:2016} to allow for every
value of \(\hat{I}\) permitted by quantum mechanics according to the
Tsirelson bounds \( -2\sqrt{2}\leq \hat{I}\leq 2\sqrt{2}\),
taking advantage of the fact that the extractability is always at
least \(1/2\).  The lower bound $\check \Xi $ is plotted as the
solid curve in Fig.~\ref{main-fig:tight_bound} of the main text.
To learn \(\check\Xi\) by spot-checking,
we use the binary RV  \(X\) defined as
$X=
\frac{1}{2}+\frac{I(R)-I_{\thresh}}{2(2\sqrt{2}-I_{\thresh})}$,
with values
$x_{\lb}=\frac{1}{2}-\frac{4+I_{\thresh}}{2(2\sqrt{2}-I_{\thresh})}$
or
$x_{\ub}=\frac{1}{2}+\frac{4-I_{\thresh}}{2(2\sqrt{2}-I_{\thresh})}$.
In terms of \(X\), \(\check\Xi\) can be expressed as
\begin{equation}
  \check\Xi=\max(\Exp(X),1/2).  \label{eq:extract_bound2}
\end{equation}

The above applies to a generic trial.  We now introduce trial indices
and account for potential dependence of a trial on its past. Let
\(\Xi_{i}\) be the past-conditional  extractability
of the pair shared in trial \(i\), and let
\(\Theta_i=\Exp(X_{i}|\Past_i)\).  According to
Eq.~\eqref{eq:extract_bound2}, we have
\(\Xi_{i}\geq \check\Xi_{i}=\max(\Theta_{i},1/2)\).  This relationship
makes it possible to convert a lower confidence bound
\(S_{\lb}\) on the sum \(\sum_{i}Y_{i}\Theta_{i}\) over unchecked
trials into confidence bounds on the sum of the past-conditional
extractabilities \(\sum_{i}Y_{i}\Xi_{i}\) over unchecked
trials.  As in the previous section, let \(C_{n}=\sum_{i}Y_{i}\).  Since
  \(\Prob\qty(\sum_{i}Y_{i}\Theta_{i}\geq S_{\lb})\geq
  (1-\epsilon)\), with probability at least \((1-\epsilon)\) we have
\begin{align}
  \max(S_{\lb}, C_{n}/2)
  &\leq \max(\sum_{i}Y_{i}\Theta_{i}, \sum_{i}Y_{i}/2)\notag\\
  &\leq \sum_{i}Y_{i}\max(\Theta_{i}, 1/2)\notag\\
  &\leq  \sum_{i}Y_{i}\Xi_{i}.
\end{align}
It follows that \(\max(S_{\lb}, C_{n}/2)\) can be interpreted as a lower confidence
bound on the sum \(\sum_{i}Y_{i}\Xi_{i}\) with confidence level \((1-\epsilon)\).  
The confidence bound \(S_{\lb}\) can be obtained
according to Eq.~\eqref{eq:conf_lb} for spot-checking with estimation
factors.

Since \(X_{i}\) is a binary RV, we can also apply the method of
Ref.~\cite{Gocanin2022} described in Sect.~\ref{sect:Gocanin2022}.
This method requires independent trials and the replacement of the
past-conditional quantities \(\Theta_i\) with
\(\theta_{i}=\Exp(X_{i})\). To compute a confidence bound on
\(\sum_{i}Y_{i}\theta_{i}\) according to Eq.~\eqref{eq:goc_lb}, we set
\(\theta_{\max}=\frac{1}{2}+\frac{I_{\max}-I_{\thresh}}{2(2\sqrt{2}-I_{\thresh})}\)
with \(I_{\max}=2\sqrt{2}\) in view of the Tsirelson
bound. Consequently, \(\theta_{\max}=1\).

Fig.~\ref{main-fig:tight_bound} of the main text compares
lower confidence bounds on the average extractability obtained by estimation-factor
spot-checking protocols with those obtained by the protocol of
Ref.~\cite{Gocanin2022} described in Sect.~\ref{sect:Gocanin2022} and
by the Serfling inequality given in Sect.~\ref{sect:serfling} for the
case where the actual trials are i.i.d.  This
corresponds to the typical situation used to establish the
completeness of the protocols.  In contrast to the latter two methods,
estimation-factor protocols are sound without any assumptions on trial
independence.  The soundness and applicability of the other two protocols are
discussed in their respective sections.  For the plot, we set the
total number of trials to \(n=10^{5}\), the confidence level to
\((1-\eps) = 99\SI{\%}\), and the spot-checking probability to
\(\omega=0.1\). We varied the CHSH value \(\hat I\) across the
range \([2,2\sqrt{2}]\). Because the RV \(X\) defined just above
Eq.~\eqref{eq:extract_bound2} is binary with possible values
\(x_{\lb}\) or \(x_{\ub}\), the CHSH value \(\hat I\) completely
determines its distribution.  For the estimation-factor protocols, to
account for the fact that \(x_{\lb}\) is not \(0\), we modified both
estimation-factor constructions and confidence-bound estimates
according to the standardization described right after the proof of
Lem.~\ref{lem:ucond_efineq}, as done explicitly in the main text.  We
performed \(N_{\text{rep}}=1000\) independent simulations of \(n\)
trials according to each chosen \(\hat I \in [2,2\sqrt{2}]\).
From each simulation, we obtained an instance \((\bm{x},\bm{y})\) of
\((\bm{X},\bm{Y})\).  For each instance, we determined four different
lower confidence bounds on \(\sum_{i}Y_{i}\Exp(X_{i}|\Past_{i})\) as
follows:
\begin{itemize}
\item \textbf{Calibrated estimation factors.}  The estimation factor
  \(T\) was constructed according to
  Protocol~\ref{prot:EF_construction} using the setup-phase
  alternative~\ref{prot:numerical}, and the confidence bound
  \(S_{\lb,1}\) was computed according to
  Protocol~\ref{prot:SC_analysis}.  The estimation factor was
  optimized numerically using \(n_c = 100\) calibration trials, where
  in each calibration trial we simulated and observed an
  instantiation of the RV \(X\) prior to the spot-checking phase.
  This added \(n_{c}\) simulated trials to each of the
  \(N_{\text{rep}}\) simulation runs.  We note that we also
  implemented the setup-phase alternative~\ref{prot:variance},
  where we analytically constructed an estimation factor using
  the mean and variance estimated from calibration trials.  The
  resulting confidence bound was found to be indistinguishable from
  the one above for the parameters we considered.  This suggests that
  in practice, analytical construction based on calibration trials
  can yield nearly optimal performance.
\item \textbf{Fixed estimation factors.}  
  The estimation factor was determined via Protocol~\ref{prot:EF_construction} 
  with the setup-phase alternative~\ref{prot:fixed}, and the corresponding 
  confidence bound \(S_{\lb,2}\) was obtained using Protocol~\ref{prot:SC_analysis}. 
  No calibration trials were required for this approach. 
\item \textbf{Method of Ref.~\cite{Gocanin2022}.}
  The confidence bound was computed as 
  \(S_{\lb,3} = n\theta_{\lb} - (n - C_n)\theta_{\max}\) (see Eq.~\eqref{eq:goc_lb}), 
  where \(\theta_{\max}=1\)  
   and \(\theta_{\lb}\) was derived via hypothesis-test inversion 
  as detailed in Sect.~\ref{sect:Gocanin2022} (see Eq.~\eqref{eq:thetalb_sol}).
\item \textbf{Serfling inequality.}  The confidence bound
  \(S_{\lb,4}\) corresponds to \(S_{\lb,\text{Serf}}\) as defined in
  Eq.~\eqref{eq:serflinglb}.  We emphasize that this bound is included
  for comparison only, as \(S_{\lb,\text{Serf}}\) provides a
  confidence bound on \(\sum_{i}Y_{i}X_{i}\), rather than on
  \(\sum_{i}Y_{i}\Exp(X_{i}|\Past_{i})\).
\end{itemize}
For each confidence bound \(S_{\lb,k}\), we computed the corresponding 
confidence bound \(\bar{\Xi}_{\lb,k}\) on the average extractability \(\bar \Xi\)
according to \(\bar \Xi_{\lb,k} = \max(S_{\lb,k}/C_{n}, 1/2)\) if
  \(C_{n}>0\) and \(\bar \Xi_{\lb,k}=1/2\) otherwise. In the latter
  case, there are no unchecked states, and we set the average of the
  empty set to the minimum extractability.
  Write \(\bar \Xi_{\lb,k}^{(i)}\) for the confidence bound obtained in the
\(i\)'th simulation run for \(i=1,\ldots, N_{\text{rep}}\). 
For each simulated CHSH value \(\hat I\), the
plotted point is the mean value \(\sum_{i}\bar \Xi_{\lb,k}^{(i)}/N_{\text{rep}}\).
For \(k=1\) (calibrated estimation factors) and \(k=3\) (according
to Ref.~\cite{Gocanin2022}) and at selected CHSH values, 
we also show violin plots to illustrate the histogram of the 
differences \(\bar \Xi_{\lb,1}^{(i)}-\bar \Xi_{\lb,3}^{(i)}\) with 
\(i=1,\ldots, N_{\text{rep}}\).

Fig.~\ref{main-fig:finite_efficiency} of the main text shows the
minimum number of trials required for the expected lower confidence
bound to exceed a specified threshold. For this, we assume a true CHSH
value of \(\hat{I}=2.7\), a value anticipated to be achievable in
the near future given the detection-loophole-free Bell tests reported
in Refs.~\cite{nadlinger:2022, zhang:2022}.  According to
Eqs.~\eqref{eq:extract_bound} and~\eqref{eq:extract_bound2}, the
corresponding true mean \(\Exp(X)\) is \(0.9111\). To ensure that the
expected lower confidence bound on \(\Exp(X)\) exceeds a desired
threshold, a minimum number of trials is required. For this analysis,
the spot-checking probability is again set to \(\omega=0.1\), and the
confidence level to \((1-\eps)=99\SI{\%}\). For estimation factors, a
general strategy to determine the minimum number of trials is
described in Sect.~\ref{sect:numerical_opt_nmin}.
Specifically, the minimum number of trials here was 
determined according to the optimization problem in
Eq.~\eqref{eq:min_exp_num} as formulated in
Sect.~\ref{sect:efficiency}, which yields results indistinguishable
from the general strategy of Sect.~\ref{sect:numerical_opt_nmin} for this study.
To determine the minimum number of trials required by the method of
Ref.~\cite{Gocanin2022} and the Serfling inequality, we used the
strategies described in the concluding paragraphs of
Sect.~\ref{sect:Gocanin2022} and Sect.~\ref{sect:serfling},
respectively.

Figs.~\ref{main-fig:tight_bound}  and~\ref{main-fig:finite_efficiency} of the main text 
demonstrate the advantage of estimation factors over the method of Ref.~\cite{Gocanin2022}. 
This advantage becomes more pronounced as the spot-checking probability \(\omega\) 
increases. To illustrate this trend, we show the performance of each 
method for the case  \(\omega=0.5\) in Figs.~\ref{SI-fig:tight_bound} and~\ref{SI-fig:finite_efficiency}.

\begin{figure}[htb!]
  \begin{center}
     \includegraphics[scale=0.45]{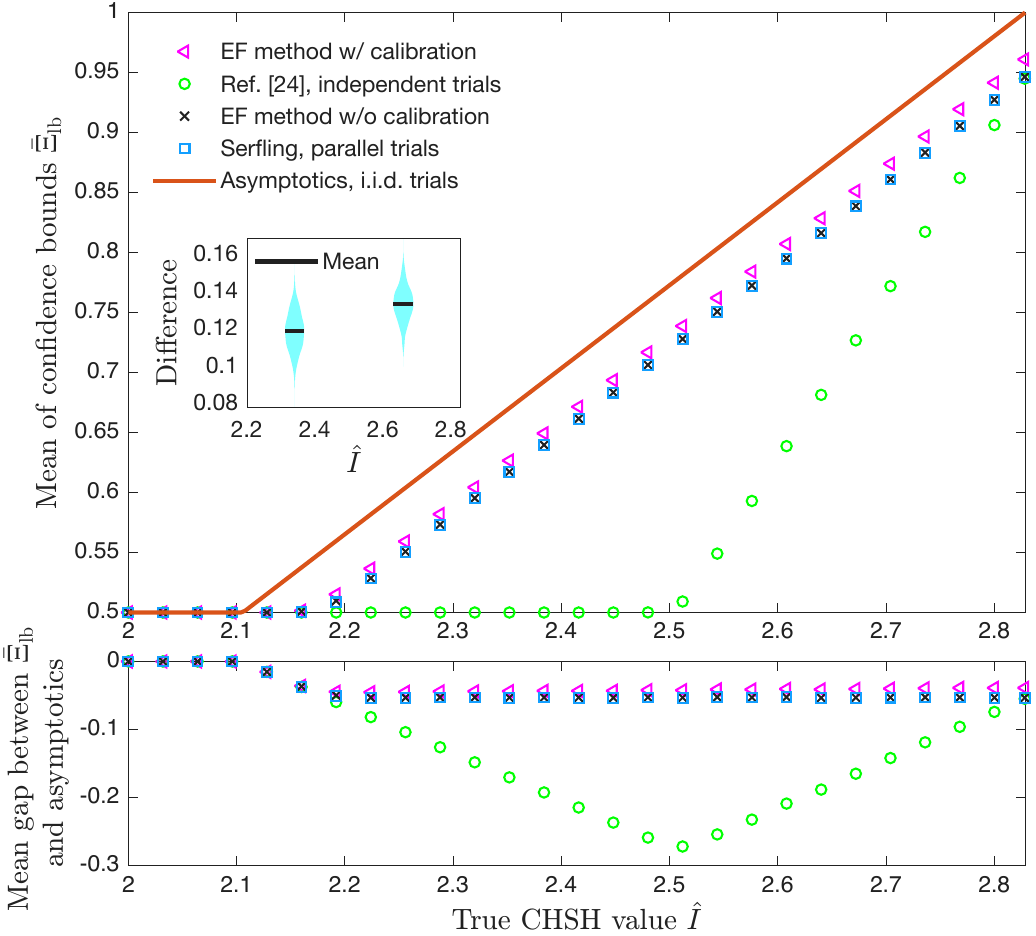}
  \end{center}
  \vspace{-0.5cm}
  \caption{Lower confidence bounds $\bar \Xi_{\lb}$ on the average
    extractability $\bar \Xi$ of the Bell state from the
    past-conditional states generated in the unchecked trials as a
    function of the true CHSH value $\hat I$.  The parameters used for
    the underlying simulations are the same as those specified for
    Fig.~\ref{main-fig:tight_bound} of the main text, except that the
    spot-checking probability is set to \(\omega=0.5\). The confidence
    bounds were obtained at the confidence level
    \((1-\epsilon)=99\SI{\%}\) following the procedure described in
    this section. }
  \label{SI-fig:tight_bound}
\end{figure}

\begin{figure}[htb!]
  \begin{center}
   \includegraphics[scale=0.45]{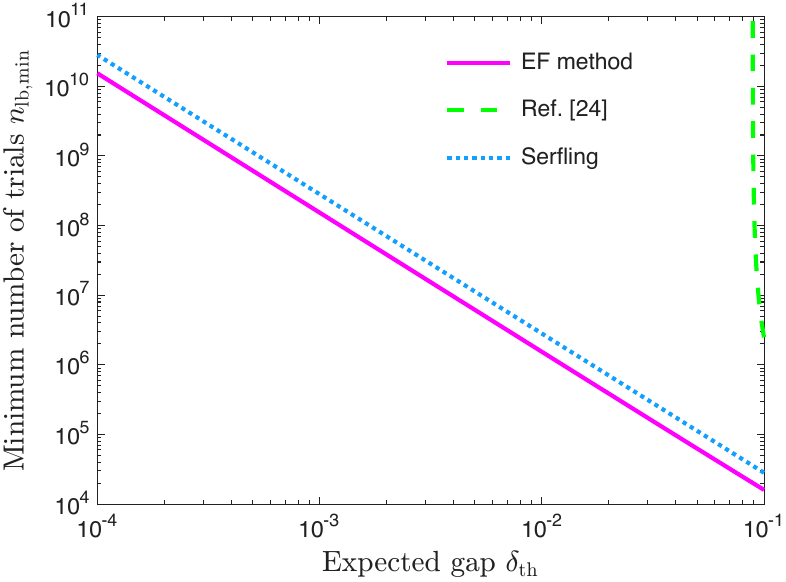}
  \end{center}
  \vspace{-0.5cm}
    \caption{Minimum number of trials $n_{\text{lb}, \min}$ required to 
  ensure that the expected lower confidence bound on $\bar \Xi$ exceeds
  \(\check \Xi-\delta_{\text{th}}\), where \(\check \Xi=0.9111\) is the 
  asymptotic lower bound on $\bar \Xi$ corresponding to a true 
  CHSH value \(\hat{I}=2.7\) 
  and \(\delta_{\text{th}}\) is the varying expected gap. We set the
      confidence level to \((1-\epsilon)=99\SI{\%}\) and assumed i.i.d. trials
      with spot-checking probability $\omega=0.5$.}
  \label{SI-fig:finite_efficiency}
\end{figure}

Finally, we illustrate the performance of each method in the limit
where \(n \to \infty\) while the expected number of spot-checked
trials, given by \(\bar{n}_s=n\omega\), remains constant. In this
limit, the actual number of spot-checked trials follows a Poisson
distribution with mean \(\bar{n}_s\). For consistency with
Fig.~\ref{main-fig:tight_bound} in the main text, we set
\(\bar{n}_s=10^4\).  For numerical computation of confidence bounds
using a sequence of trials, we need to set the total number of trials
\(n\) in the sequence to be sufficiently large but finite. In our
simulation, we set \(n=10^9\) after observing that further increasing
the value of \(n\) did not significantly affect the computed
confidence bounds.  We performed \(N_{\text{rep}}=1000\) independent
simulations of \(n\) trials according to each chosen
\(\hat I \in [2,2\sqrt{2}]\).  Following the procedure described in
this section, from each simulation we computed a confidence bound on
the average extractability using each protocol considered. We then
computed the mean of these confidence bounds for each protocol. The
results are presented in Fig.~\ref{SI-fig:limit_bound}. These
numerical results indicate that in the limit where \(n \to \infty\)
while \(\bar{n}_s=n\omega\) remains constant, the
estimation-factor-based confidence bounds still perform well, yielding
results comparable to those illustrated in
Fig.~\ref{main-fig:tight_bound} of the main text.  This observation is
consistent with the asymptotic behavior analyzed in
Sect.~\ref{sect:constant_omega_n}.  Fig.~\ref{SI-fig:limit_bound} also
shows that in this limit, the method of Ref.~\cite{Gocanin2022}
performs comparably to the estimation-factor method with calibration.
However, we remark that for spot-checking non-binary RVs, the bounds
according to Ref.~\cite{Gocanin2022}, if proven applicable, would
perform significantly worse in general than those obtained via
calibrated estimation factors.

\begin{figure}[htb!]
  \begin{center}
     \includegraphics[scale=0.45]{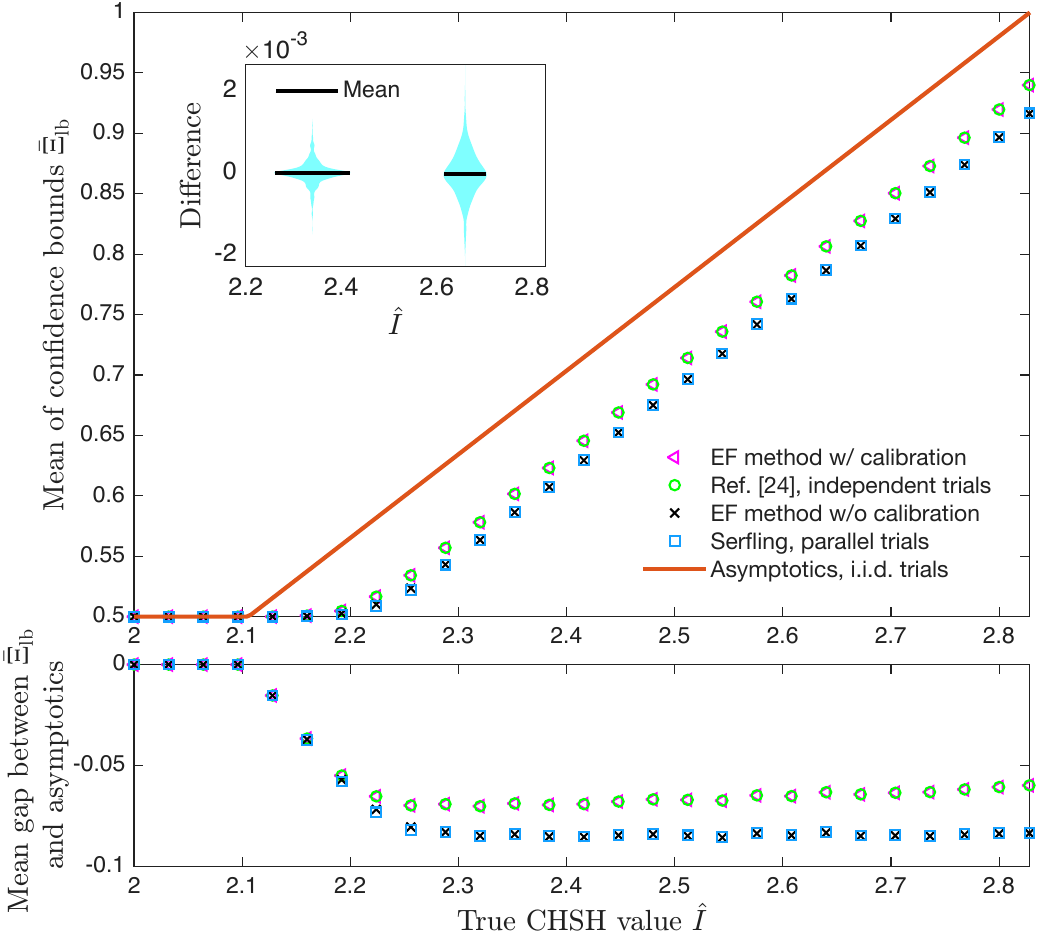}
  \end{center}
  \vspace{-0.5cm}
  \caption{Lower confidence bounds $\bar \Xi_{\lb}$ on the average
    extractability $\bar \Xi$ of the Bell state from the
    past-conditional states generated in the unchecked trials as a
    function of the true CHSH value $\hat I$.  The parameters used for
    the underlying simulations are the same as those specified for
    Fig.~\ref{main-fig:tight_bound} of the main text, except that we
    analyze the limit situation \(n \to \infty\) and \(\omega \to 0\)
    while the expected number of spot-checked trials is kept constant
    at \(\bar{n}_s=n\omega=10^{4}\).  For the plot, we set \(n=10^{9}\)
    after verifying that increasing \(n\) further did not significantly 
    affect the results. The confidence bounds were
    obtained at the confidence level \((1-\epsilon)=99\SI{\%}\)
    following the procedure described in this section.}
  \label{SI-fig:limit_bound}
\end{figure}

\section{Application to QKD}
\label{sect:appls}

Spot-checking strategies are abundant in protocols such as QKD, and
the Serfling inequality is frequently used for this purpose.  As an
example, we consider the QKD protocols analyzed in
Refs.~\cite{Tomamichel2012,Tomamichel2017}. The protocols are
specified in Box 1 of Ref.~\cite{Tomamichel2012} and Table 2 of
Ref.~\cite{Tomamichel2017} . In the latter protocol, the two parties
\(A\) and \(B\) agree on a subset \(J\) of \(k\) trials randomly
selected from a total of \(n\) trials to perform ``parameter
estimation'' by comparing their 0/1-valued measurement outcomes.  Let
\(O_{A,i}\) and \(O_{B,i}\) be the measurement outcomes of \(A\) and
\(B\) in trial \(i\), assuming both parties used the basis
complementary to the key-generation basis.  In the implementations of
the QKD protocols described in
Refs.~\cite{Tomamichel2012,Tomamichel2017}, this complementary basis
is used by both parties for all trials in the randomly selected subset
\(J\).  The error for trial \(i\) is defined as
\(\bar X_{i}=O_{A,i}\oplus O_{B,i}\), which takes the value \(1\) when 
\(O_{A,i}\neq O_{B,i}\) and \(0\) otherwise. Parameter estimation involves
determining whether the average error rate
\(\frac{1}{n-k}\sum_{i\notin J} \bar X_{i}\) is below a threshold
\(\delta\). The protocol aborts if it is not. Let
\(X_i=1-\bar X_{i}\), which takes the value of \(1\) when there is no
error.  The proof of security in
Refs.~\cite{Tomamichel2012,Tomamichel2017} is based on the Serfling
inequality. It amounts to applying the Serfling inequality to
establish a lower confidence bound of at least \((n-k)(1-\delta)\) on
\(\sum_{i\notin J} X_{i}\) with error bound \(\epsilon\), where \(\epsilon\)
is a small security parameter. Instead of applying the Serfling
inequality, one can use the estimation-factor method to establish the
confidence bound.  This has two advantages. First, since estimation-factor-based
bounds are generally tighter than those obtained via the Serfling
inequality, it improves the success probability of the protocol or
allows a higher error-tolerance threshold \(\delta\).  This
improvement is apparent from the comparison of
Eq.~\eqref{eq:serflinglb} with Eq.~\eqref{eq:expgap}, particularly in
the case of interest here, where \(X_{i}\) is highly biased toward
\(1\).  Second, the application of the Serfling inequality requires
the assumption that the trials are parallel or ``as-if'' parallel. See
the input specification in Table 2 of Ref.~\cite{Tomamichel2017} and
the ``commuting measurements'' assumption given before the table. This
assumption can be significantly weakened when using the
estimation-factor method, making it more suitable for practical QKD
implementations in which the trials are generated sequentially.

\ifSubfilesClassLoaded{%
    \bibliographystyle{unsrt} 
    \bibliography{./spotcheck_est} 
}{}

\end{document}
